%% file: ConfMedEqHenningMaja_111116.tex
\documentclass{fac}
\usepackage{amssymb}
\usepackage{amsmath}
\newtheorem{definition}{Definition}
\newtheorem{theorem}{Theorem}
\newtheorem{proposition}{Proposition}
\newtheorem{lemma}{Lemma}

\newtheorem{example}{Example} 
\usepackage{hyperref} 
\usepackage{graphicx}
\usepackage{caption}
\usepackage{subcaption}
\usepackage[all]{xy}
\usepackage{todonotes}
\usepackage{color,soul}

\def\true{\mathit{true}}

\def\failure{\mathit{failure}}
\def\error{\mathit{error}}

\def\exe{\mathit{Exe}}

\def\vars{\mathit{vars}}

\def\id{\mathit{id}}
\def\post{\mathit{POST}}

\def\mapsfrom{\mathrel{\mbox{$\leftarrow$\hskip -0.2ex\vrule width 0.15ex height 1ex\hbox to 0.3ex{}}}}

\def\ourmapsto{\mathrel{\mbox{\hbox to 0.2ex{}\vrule width 0.15ex height 1ex\hskip -0.2ex$\rightarrow$}}}

\def\mapstostar{\stackrel{*}\ourmapsto}
\def\mapsfromstar{\stackrel{*}\mapsfrom}

\newcommand{\imply}{\Rightarrow}

\newcommand{\relRT}{\overset{\!\!*}{\rightarrow}}
\newcommand{\irelRT}{\overset{\,\,*}{\leftarrow}}

\def\lift#1{\setbox1=\hbox{#1\hbox to 4pt{}}\vbox to 1.4\ht1{\hrule width 1\wd1 height .8pt\hbox to 1\wd1{\vrule height 3pt width .8pt\hfil\vrule height 3pt width .8pt}\vss\hbox to 1\wd1{\hfil#1\hfil}}}

\def\metachr{\textsc{MetaCHR}}
\def\where{\mathrel{\mbox{\textsc{where}}}}

\def\meta{\textsc{Meta}}
\def\metagr{\textsc{Meta}^{\mathit{Gr}}}
\def\chr{\textsc{Chr}} 

\def\denotes#1{[\![#1]\!]}
\def\denotesgr#1{[\![#1]\!]^{\mathit{Gr}}}

\def\bigdenotes#1{\bigl[\!\!\bigl[#1\bigr]\!\!\bigr]}
\def\bigdenotesgr#1{\bigl[\!\!\bigl[#1\bigr]\!\!\bigr]^{\mathit{Gr}}}

\def\allrars{\textit{all-relevant-app-recs}}
\def\mexe#1#2{\widehat\exe(#1,#2)}
\def\mexenoargs{\widehat\exe} 
\def\commonvars{\textit{common-vars}}

\def\whiteghost#1{{\setbox9=\hbox{#1}\hbox to
\wd9{\vrule width 0dd depth
\dp9 height \ht9 \hfil}}}
\begin{document}

\correspond{Henning Christiansen, Roskilde University, Denmark. Email: henning@ruc.dk}

\title[On Proving Confluence Modulo Equivalence for CHR]
{On Proving\\Confluence Modulo Equivalence\\for Constraint Handling Rules}
  \author[Christiansen and Kirkeby]
      {Henning Christiansen\thanks{The project
      is supported by The Danish Council for Independent Research, Natural Sciences,
      grant no.~DFF 4181-00442} and Maja H.~Kirkeby$^{1,}$\thanks{The second author's contribution  received funding from
the European Union Seventh Framework Programme (FP7/2007-2013)
under grant agreement no.~318337,
ENTRA - Whole-Systems Energy Transparency.}\\
       Research group PLIS: Programming, Logic and Intelligent Systems\\
      Department of People and Technology\\
      Roskilde University, Denmark}

\maketitle

\begin{abstract}
Previous results on proving confluence
for Constraint Handling Rules  are extended
in two ways in order to allow a larger and more realistic class of CHR programs to be considered confluent.
Firstly, we introduce the relaxed notion of confluence modulo equivalence into the context of CHR: while
confluence for a terminating program means that all alternative derivations for a 
query lead to the exact same final state, confluence modulo equivalence only requires
the final states to be equivalent with respect to an equivalence relation tailored for the given program.
Secondly, we allow non-logical built-in predicates such as \texttt{var}/1 and incomplete ones such as \texttt{is}/2,
that are ignored in previous work on confluence.

To this end, a new operational semantics for CHR is developed which includes such predicates.
In addition, this semantics differs from earlier approaches by its simplicity 
without loss of generality, and it may also be recommended for future studies of CHR.

For the purely logical subset of CHR, proofs can be expressed in first-order logic,
that we show is not sufficient in the present case.
We have introduced a formal meta-language that allows reasoning about abstract states
and derivations with meta-level restrictions that reflect the non-logical and incomplete predicates.
This language represents subproofs as diagrams, which facilitates a systematic
enumeration of proof cases, pointing forward to a mechanical support for such proofs.
\end{abstract}


\section{Introduction}
Constraint Handling Rules, CHR~\cite{fruehwirth-98,fru_chr_book_2009},
is a programming language consisting of guarded rewriting rules over constraint stores.
CHR inherits its nomenclature from the logic programming tradition; constraints are first-order atoms,
and the language has a declarative semantics based on a logical reading of the rules.
It has become important as a general language for knowledge representation and reasoning
as well as
for expressing algorithms in a high-level fashion;
see, e.g.,~\cite{FruehwirthRaiserEds2011,DBLP:series/lncs/5388,DBLP:journals/tplp/SneyersWSK10}.

A foundation for applying confluence in the analysis and verification of CHR programs
has been laid in earlier work, and the overall theoretical issues are well understood~\cite{DBLP:conf/cp/Abdennadher97, DBLP:conf/cp/AbdennadherFM96,DBLP:journals/constraints/AbdennadherFM99,FruehwirthRaiserEds2011}.
The confluence notion goes longer back in the traditions of term and abstract rewriting systems;
see more details in Section~\ref{sec:rewrite}.
There are, however, still
severe limitations in the results for CHR that impede its
practical application to realistic programs.
The present paper aims at filling part of the gap, by 
\begin{itemize}
\item the introduction for CHR of confluence \emph{modulo equivalence} 
that allows a much larger and interesting class of programs to enjoy the advantages of confluence;
\item extending to a larger subset of CHR that includes non-logical and incomplete\footnote{In this paper, we use the term  \emph{incomplete} for a built-in predicate whose (established) implementation
produces runtime errors for selected calls. Examples of such calls are \texttt{4} \texttt{is} \texttt{2+X} and \texttt{X>1}.
The precise definition is found in section~\ref{sec:preliminaries}.} built-in predicates
(e.g., \texttt{var}/1, resp.\ \texttt{is}/2) that have been ignored in previous work.
\end{itemize}
While confluence of a program means that all derivations from a common initial state end in the same final state,
 the  ``modulo equivalence'' version relaxes this
such that final states need not be strictly identical,
but only equivalent with respect to a given equivalence relation.
The following motivating example is used throughout this paper.

\begin{example}[\cite{DBLP:conf/lopstr/ChristiansenK14}]\label{ex:collect}
The following CHR program, consisting of a single rule, collects a number of separate items
into a (multi-) set represented as a list of items.
\begin{verbatim}
set(L), item(A) <=> set([A|L]).
\end{verbatim}
This rule will apply repeatedly, replacing constraints matched by the left-hand side by those indicated to the right.
The query
\begin{verbatim}
?- item(a), item(b), set([]).
\end{verbatim}
may lead to two  different final states, $\{\texttt{set([a,b])}\}$ and  $\{\texttt{set([b,a])}\}$,
both representing the same set.
We introduce a state equivalence relation  $\approx$ implying that $\{\texttt{set($L$)}\}\approx\{\texttt{set($L'$)}\}$, whenever $L$ is a permutation of $L'$.
The program is not confluent when identical end states are required,
but it will be shown to be confluent modulo $\approx$ in Section~\ref{sec:set-example-all-details} below.
\end{example}
The relevance of confluence modulo equivalence is also demonstrated for dynamic optimization programs that produce an arbitrary, optimal solution among a collection of equally good ones;
the Viterbi algorithm
expressed in CHR is considered in Section~\ref{sec:viterbi}.

To model non-logical and incomplete predicates, we need to introduce a new operational semantics for CHR. To be interesting for studies of confluence, this semantics maintains
nondeterminism for choice of the next rule to be applied to the current state.
In addition to treating a larger language, this semantics differs from earlier approaches by its simplicity
without  loss of generality.
Various redundancies have been removed so that a program state has only two components, a constraint store and a bookkeeping device to handle well-known termination
problems for the propagation rules of CHR; a simple observation shows that global variables are unnecessary;
execution of built-in predicates are modelled by substitutions applied to the state immediately,
which is more in line with how a practical CHR system works (as opposed to earlier proposals'
additional store of ``processed'' built-ins and their evaluation explained by logical entailment).
A detailed comparison and references to previous operational semantics are given in Section~\ref{sec:CHR}
below.

Reasoning about derivations is more difficult in the context of non-logical/incomplete built-ins.
Basically, all earlier proof methods for the purely logical subset of CHR rely on a subsumption principle
that any property shown about derivations between states also
holds when more constraints are added and substitutions applied to the states;
as a consequence of this, confluence proofs can be reduced to considering a finite number
of cases that can be 
checked in an automatic way.
This principle breaks down when non-logical predicates are introduced, e.g.,
the predicate \texttt{var(X)} succeeds but the instance  \texttt{var(7)}
fails.
To cope with this, we have introduced a formal meta-language $\metachr$ to represent abstract states, derivations
and proofs as diagrams, with powerful parametrization and meta-level constraints that limit the allowed
instances. The following is an example of an abstract term in the meta-language,
$\texttt{var($a$)} \where \mathit{variable}(a)$. Here, $a$ is a meta-variable ranging over terms
and  $\mathit{variable}$ is a meta-level constraint on such terms, allowing only substitutions to names of such variables. This abstract term 
is said to cover all instances that satisfy the meta-level constraint, i.e.,
 \texttt{var(X)} but not  \texttt{var(7)}.
$\metachr$ allows us to reason about such abstract terms in a way so that properties shown at this
level are guaranteed to hold for all such permissible instances.
We can demonstrate that proofs of confluence can be reduced to considering only a finite number of
abstract proof cases, but the additional complexity given by an equivalence relation
(and state invariant; below) may in some cases require an unfolding into an infinite number
of subcases, each requiring a differently shaped proof diagram.

The notion of observable confluence~\cite{DBLP:conf/iclp/DuckSS07} for CHR considers only states that satisfy a given invariant.
We include such invariants, as we consider them to be central in CHR programming practice:
a program is typically developed with a particular class of queries in mind, often strongly biased, 
so only queries in this class lead to meaningful computations.

\begin{example}\label{ex:collect-inv}\emph{(Example~\ref{ex:collect}, continued)}
The one-line program above reflects a tacitly assumed state invariant:
only one \texttt{set} constraint is allowed.
If we open up for a query such as
\begin{verbatim}
?- item(a), item(b), set([]), set([c]).
\end{verbatim}
we obtain a collection of different answers, representing different ways of splitting
$\{\texttt{a},\texttt{b},\texttt{c}\}$ into two disjoint subsets. 
However, this may not be intended, and the program is not confluent modulo the indicated equivalence
relation unless the invariant is taken into account.
The relevant invariant may specify that all constraints must be ground, and that a state must include exactly
one \texttt{set}/2 constraints whose argument is a list.
\end{example}
The earlier approach~\cite{DBLP:conf/iclp/DuckSS07} for showing observable confluence (for logical built-ins only) 
sticks to the above mentioned logical subsumption principle.
As shown by~\cite{DBLP:conf/iclp/DuckSS07} and explained below, this leads to infinitely many proof cases for even simple invariants such as groundedness;
our meta-language approach handles such examples in a more satisfactory way.

Confluence modulo equivalence was mentioned in relation to  CHR in a previous conference
paper~\cite{DBLP:conf/lopstr/ChristiansenK14} that also gave a first version of the
operational semantics. 
The present paper provides theoretical foundations for studying
confluence modulo equivalence for CHR, and introduces a formal meta-language that supports systematic proofs. This may also point forward towards (partly) mechanized proof systems for confluence modulo equivalence.

The results in the present paper may carry over in a useful way to other systems
with nondeterminism in which confluence has to be studied. This may be active rules in databases~\cite{DBLP:conf/sigmod/AikenWH92}, concurrent constraint
programming~\cite{DBLP:journals/tcs/FalaschiGMP97}
and theoretical models of concurrency such as
$\pi$- and $\rho$-calculi~\cite{DBLP:journals/jfp/Niehren00,DBLP:conf/ccl/NiehrenS94}.

%



Section~\ref{sec:background} reviews previous work on confluence in term rewriting
and general rewriting systems, including fundamental results concerning confluence modulo
equivalence, that has not been utilized for CHR before, and we give an overview
of the state of the art for CHR.
Section~\ref{sec:CHR} gives 
our operational semantics for CHR, first introduced in~\cite{DBLP:conf/lopstr/ChristiansenK14},
intended for reasoning about confluence for programs with  non-logical built-in predicates,
and various properties related to confluence are introduced; we also make a comparison with operational
semantics used in earlier work on confluence for CHR.
In Section~\ref{sec:conf-mod-eq-in-CHR} we generalize earlier results on critical pairs for CHR, now including
the larger set of built-in predicates, and taking invariant and equivalence into account;
we can also show that such pairs -- or corners as we call them (since we include the common ancestor state) --
are not suited for proofs of confluence in our more general case due to this subsumption principle; we also add some more detailed comments
on previous work on confluence for CHR.
%

Our main results are presented in Section~\ref{sec:abstract}. The meta-language \metachr{} is introduced in which proofs of joinability are reified as abstract diagrams. 
A proof of confluence modulo equivalence can be split into a finite set of proof cases, each given by an abstract corner.
As opposed to the results of \cite{DBLP:conf/cp/Abdennadher97,DBLP:conf/cp/AbdennadherFM96} it is not necessary for confluence (modulo equivalence) that
each such abstract corner is joinable. A property called split-joinable is introduced, occasionally leading to  infinite sets of corners to be checked for joinability.
We show that when the abstract corners are either joinable or split-joinable, local confluence is guaranteed and confluence is guaranteed for terminating programs.

In Section~\ref{sec:examples}, we demonstrate the applicability of the suggested approach,  by giving proofs of confluence modulo
equivalence for selected programs: the program of Example~\ref{ex:collect} that demonstrates
an equivalence indicating a redundant data representation,
a version of the Viterbi algorithm in CHR that exemplifies dynamic programming
algorithms with pruning, and finally an example with a splitting into infinitely many cases.
%
Section~\ref{sec:discussion}
provides for a summary, and a discussion of possible
directions for future work.

\section{Background and Related work}\label{sec:background}
Confluence modulo trivial identity is well-studied in Rewriting Systems,
see, e.g.,~\cite{BaderNipkow1999} for an overview. Since the 1990es, the proof methods have been adapted to the more complex system of Constraint Handling Rules~\cite{fruehwirth-98,fru_chr_book_2009}, most notably~\cite{DBLP:conf/cp/Abdennadher97,DBLP:conf/cp/AbdennadherFM96,DBLP:conf/iclp/DuckSS07}.
Confluence modulo equivalence has been studied in general rewriting systems~\cite{DBLP:journals/jacm/Huet80} and was only recently introduced to CHR~\cite{DBLP:conf/lopstr/ChristiansenK14}.
%
%
%
%
%
%
%
%
%
\subsection{Confluence for General Rewriting Systems and Term Rewriting Systems}\label{sec:rewrite}
A binary \emph{relation} $\rightarrow$ on a set $A$ is a subset of $A \times A$, where $x \rightarrow y$ denotes membership of $\rightarrow$. A \emph{rewriting system} is  a pair $\langle A, \rightarrow\rangle$; it is \emph{terminating} if there is no infinite chain $a_0 \rightarrow a_1 \rightarrow \cdots$.
The \emph{reflexive transitive closure} of $\rightarrow$ is denoted $\relRT$. The \emph{inverse relation} $\leftarrow$ is defined by $\{(y,x) \mid x \rightarrow y\}$. 
An \emph{equivalence (relation)} $\approx$ is a binary relation on $A$ that is reflexive, transitive and symmetric. We say that $x$ and $y$ are \emph{joinable} if there exists a $z$ such that $x \relRT z $ and a $z \irelRT y$.
%
%
%
%
%

A rewriting system $\langle A, \rightarrow\rangle$ is \emph{confluent} if and only if $y' \irelRT x \relRT y \imply \exists z.\ y' \relRT z \irelRT y$, and is \emph{locally confluent} if and only if $y' \leftarrow x \rightarrow y \imply \exists z.\ y' \relRT z \irelRT y$.
In 1942, Newman showed his fundamental Lemma~\cite{Newman42}:
\emph{A terminating rewriting system is confluent if and only if it is locally confluent.} An elegant proof of Newman's lemma was provided by Huet \cite{DBLP:journals/jacm/Huet80} in 1980.
%


The more general notion of \textit{confluence modulo equivalence} was introduced in 1972 by Aho et~al~\cite{Aho72} in the context of the Church-Rosser property.
\begin{definition}[Confluence modulo equivalence]\label{def:confModEq}
A relation $\rightarrow$ is confluent modulo an equivalence $\approx$ if and only if
\begin{align*} 
\forall\, x,y,x', y'. \quad y' \irelRT x' \approx x \relRT y \qquad  \imply \qquad \exists\, z, z'. \quad  y' \relRT z' \approx z \irelRT y.
\end{align*}
\end{definition}
Given an equivalence relation $\approx$, we say that $x$ and $y$ are \emph{joinable modulo equivalence} if there exists $z,z'$ such that $x \relRT z $, $z' \irelRT y$ and $z \approx z'$.
This is shown as a diagram in Fig.~\ref{fig:ConfModEq}.
In 1974, 
Sethi~\cite{DBLP:journals/jacm/Sethi74} 
studied confluence modulo equivalence for bounded rewriting systems, that are systems for which there exists an upper bound for the number of possible rewrite steps for all terms. He showed that confluence modulo equivalence for bounded systems is equivalent to the following properties, $\alpha$ and $\beta$, also shown in Fig.~\ref{fig:LocalConfModEq}. 
\begin{definition}[$\alpha$ \& $\beta$]\label{def:alfaBeta}
A relation $\rightarrow$ has the $\alpha$ property and the $\beta$ property
with respect to an equivalence $\approx$ if and only if it satisfies the $\alpha$ and $\beta$ conditions, respectively:
\begin{align*}
\alpha:\qquad & \forall x,y, y'. \quad y' \leftarrow x \rightarrow y \qquad \Longrightarrow \qquad \exists z, z'. \quad  y' \relRT z' \approx z \irelRT y \\
\beta: \qquad & \forall x, y', y. \quad y' \approx x \rightarrow y  \qquad\, \Longrightarrow \qquad \exists z, z'. \quad  y' \relRT z' \approx z \irelRT y
\end{align*}
\end{definition}
In 1980,~Huet \cite{DBLP:journals/jacm/Huet80} generalized this result to any terminating system.
%
\begin{definition}[Local confl.~mod.~equivalence]\label{def:LconfModEq}
A rewriting system is \emph{locally confluent modulo an equivalence} $\approx$ if and only if it has the $\alpha$ and $\beta$ properties.
\end{definition}
\begin{theorem}\label{thm:confLconfModEq}\textbf{(Huet,~\cite{DBLP:journals/jacm/Huet80})}\quad
Let $\rightarrow$ be a terminating rewriting system. For any equivalence $\approx$, $\rightarrow$ is confluent modulo $\approx$ if and only if $\rightarrow$ is locally confluent modulo $\approx$.
\end{theorem}
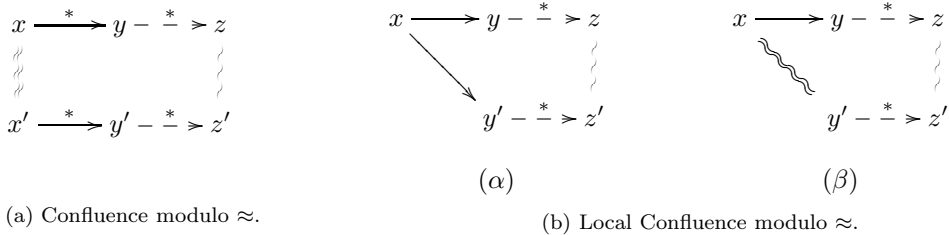
\begin{figure}
\centering
\begin{subfigure}{0.3\textwidth}
   \begin{displaymath}
        \begin{array}{c}
     \xymatrix{ x \ar[r]^{*} \ar@2{~}[d] & y  \ar@{-->}[r]^{*} & z \\
                x' \ar[r]^{*}   	 & y' \ar@{-->}[r]^{*} & z' \ar@{~}[u] }
                \\ \\ 
     \end{array}
   \end{displaymath} 
   \subcaption{Confluence modulo $\approx$.}
   \label{fig:ConfModEq}
 \end{subfigure}
 \begin{subfigure}{0.6\textwidth}
   \begin{subfigure}{0.45\textwidth}
     \begin{displaymath}
     \begin{array}{c}
	  \xymatrix{ x \ar[r] \ar[dr] & y  \ar@{-->}[r]^{*} & z \\
				      & y' \ar@{-->}[r]^{*} & z' \ar@{~}[u] }\\ \\
	 ( \alpha  )
     \end{array}
     \end{displaymath}
   \end{subfigure}
   \begin{subfigure}{0.45\textwidth}
   \begin{displaymath}
     \begin{array}{c}
	\xymatrix{  x \ar[r]^{} \ar@2{~}[dr]	& y    \ar@{-->}[r]^{*}       & z \\
						& y'   \ar@{-->}[r]^{*}       & z' \ar@{~}[u]      }\\ \\
	 ( \beta  )
     \end{array}
     \end{displaymath}
   \end{subfigure}
 \subcaption{Local Confluence modulo $\approx$.}
   \label{fig:LocalConfModEq}
 \end{subfigure}  
%
%
\caption{Diagrams for the fundamental notions.
A dotted arrow (single wave line) indicates an inferred step (inferred equivalence).}
\label{fig:alphaBeta}
\end{figure}
%
Term rewriting systems have been studied extensively, and
 terminology and several important results carry over to CHR, as we will see below.
In the following, we assume the reader familiar with the notions of terms over some signature and variables,
substitutions and most general unifiers.

\begin{definition}[Term Rewriting System; semi-formal version adapted from~\cite{BaderNipkow1999}]
A \emph{term rewriting system (TRS)} consists of a finite set of
rules of the form $(l,r)$ in which any variable in $r$ also appears in $l$.
The application of such a rule to a term $s$ to obtain another term $t$, written $s\rightarrow t$
is obtained by 1) find a substitution $\theta$, such that $l\theta$ is a subterm of $s$, and 2)
$t$ is given by replacing that subterm in $s$ by $r\theta$.
\end{definition}
The following notion of critical pairs represents cases in which two rules both can apply in the same
subterm, but if one is applied, the second one cannot be applied successively.

\begin{definition}[TRS Critical Pair; adapted from~\cite{BaderNipkow1999}]
Consider $R_k=(l_k,r_k), k=1,2$ (assumed renamed apart so they have no variable in common)
for which there is a most general unifier $\sigma$ of $l_2$ and a non-variable subterm of $l_1$.
Then $\langle t_1, t_2\rangle$ is a \textit{critical pair}, whenever
$l_1\sigma\rightarrow t_k$ using $R_k$, $k=1,2$.
\end{definition}
For example, the two rules $(f(a),b)$ and $(a,c)$ give rise to the critical pair $\langle b, f(c)\rangle$;
both are derived from the common ancestor term $f(a)$, i.e.,
$b\leftarrow f(a)\rightarrow f(c)$.

In 1970, Knuth and Bendix~\cite{KnuthBendix1970} developed the following, fundamental properties,
later elaborated by Huet~\cite{DBLP:journals/jacm/Huet80}.
We bring them in detail as very similar properties
holds for CHR.

\begin{lemma}[Critical Pair Lemma for TRS~\cite{DBLP:journals/jacm/Huet80,KnuthBendix1970}]\label{def:criticalPairsTRS}
Let a TRS be given and assume terms $s,t_1,t_2$ such that $t_1\leftarrow s \rightarrow t_2$.
Then either 
\begin{itemize}
  \item $t_1$ and $t_2$ are joinable, or
  \item there exists an instance $\langle u_1,u_2\rangle$ or  $\langle u_2,u_1\rangle$ of a critical pair
  and a specific subterm $s'$ of $s$ such that\\
      $t_k$ is a copy of $s$ in which $s'$ is replaced by $u_k$, $k=1,2$.
\end{itemize}
\end{lemma}

\begin{theorem}[Critical Pair Theorem for TRS~\cite{DBLP:journals/jacm/Huet80,KnuthBendix1970}]\label{thm:criticalPairs}
 A  TRS is locally confluent if and only if all its critical pairs are joinable.
\end{theorem}
This theorem in combination with Newman's lemma leads to a desired result: \emph{A terminating TRS is confluent if and only if all its critical pairs are joinable.}
Furthermore, confluence of a finite terminating TRS is decidable
(as there is only a finite number of critical pairs and finitely many finite derivations to test out from their states).

Mayr and Nipkow~\cite{DBLP:journals/tcs/MayrN98} studied 
confluence modulo equivalence for a subset of higher-order rewriting systems (that extend term rewriting to $\lambda$-terms).
They used an alternative version of Theorem~\ref{thm:confLconfModEq} in which the $\beta$ property is replaced
by a $\gamma$ property, as shown below.
It applies when the equivalence $\approx$ is specified as the transitive closure
of
a symmetric relation $\vdash\!\dashv$; such a relation may, e.g., be generated by
a set of equations.
\begin{lemma}[$\alpha$ \& $\gamma$ Confluence \cite{DBLP:journals/jacm/Huet80}]~\label{lemma:alpga-and-gamma}
 Let $\vdash\!\dashv$ be a symmetric relation and  ${\approx}  = (\vdash\!\dashv)^{*}$.
 Let $\rightarrow$ be any relation such that the composition $\rightarrow \cdot \approx$ is terminating. 
 Then $\rightarrow$ is confluent modulo $\approx$ if and only if the conditions $\alpha$ and $\gamma$ are satisfied:
 \begin{align*}
\alpha:\qquad & \forall x,y, y'. \quad y' \leftarrow x \rightarrow y \qquad \Longrightarrow \qquad \exists z, z'. \quad  y' \relRT z' \approx z \irelRT y \\
\gamma: \qquad & \forall x, y', y. \quad y' \vdash\!\dashv x \rightarrow y  \qquad\, \Longrightarrow \qquad \exists z, z'. \quad  y' \relRT z' \approx z \irelRT y
\end{align*}
\end{lemma}
%
%
We do not use this lemma in the present paper, but possible
applications are
discussed in the concluding section.

\subsection{Confluence for Constraint Handling Rules}\label{sec:backgroundConfCHR}
Constraint Handling Rules, CHR, can be understood as a rewrite system over states
that are multisets of constraints
as shown in Example~\ref{ex:collect} above, p.~\pageref{ex:collect}.\footnote{The rule
in the example program is a so-called simplification rule.
CHR also includes other types of rules, that
do not introduce additional conceptual difficulties
in relation to confluence, although they 
imply an extra notational  overhead.}

The known results on confluence for CHR are very similar to those on term rewriting systems shown above.
Similar  critical pairs of states  may appear when two instances of rules
can apply to overlapping constraints;
the precise definition is given in Section~\ref{sec:CHR} below.
The following shows the construction of
such a critical pair for an overlap of two different instances of the only rule in the program
of Example~\ref{ex:collect}, p.~\pageref{ex:collect}, above.
$$
\{\texttt{item(Y)},\texttt{set([X|L])}\} \leftarrow \{\texttt{item(X)},\texttt{item(Y)},\texttt{set(L)}\}\rightarrow \{\texttt{item(X)}, \texttt{set([Y|L])}\}
$$
The first publications by Fr\"uhwirth on CHR appeared in 1993--4~\cite{DBLP:conf/iclp/Fruhwirth93,DBLP:journals/lncs/Fruhwirth94}. Soon after, around 1996, the central results
on confluence for CHR were developed by Abdennadher and others~\cite{DBLP:conf/cp/Abdennadher97,DBLP:conf/cp/AbdennadherFM96}, however, only for the subset of CHR 
with logical built-ins and neither invariant nor equivalence.
The concepts and results from the area of term rewriting can be transferred to CHR so that
the following results hold; CHR$^0$ refers to the indicated subset of CHR.
\begin{itemize}
  \item A CHR$^0$ program is locally confluent if and only if all its critical pairs are joinable.
  \item The set of critical pairs is finite and
  local confluence is decidable; automatic checkers of this property has been developed 
for CHR$^0$, e.g.,~\cite{Raiser-Langbein2010}
  \item A terminating CHR$^0$ program is confluent if and only if all its critical pairs are joinable.
\end{itemize}
%
These results are based on the previously mentioned subsumption principle which essentially boils down to the following. 
\begin{itemize}
  \item[(*)] whenever a (e.g., critical) pair of CHR$^0$ states $x,y$ are joinable,
it holds for any substitution $\theta$ and constraint set $s$ that $x\theta\cup s$ and $y\theta\cup s$
are joinable.
\end{itemize}
Section~\ref{sec:earlier-CHR-confl-etc}, p.~\pageref{sec:earlier-CHR-confl-etc},
gives a precise analysis and also shows that these results do not generalize directly to the larger subset of CHR considered in the present paper.

In 2007, Duck et al~\cite{DBLP:conf/iclp/DuckSS07} argued for the introduction of state invariants;
a state invariant $I(\cdot)$ is a property that is preserved by the derivations of the current program,
and it may, e.g., be defined by reachability from a set of intended queries.
We define an \emph{$I$-state} $x$ as a state for which $I(x)$ holds.
%
The precise definitions and arguments are given in Section~\ref{sec:semantics}, respectively~\ref{sec:earlier-CHR-confl-etc}.
They define \emph{(local) observable confluence} for  CHR$^0$
 as above,  considering only derivations between $I$-states.

While  this generalization of confluence is highly relevant from a practical
point of view, it is inherently more difficult, as the property (*) above does not generalize.
For this discussion, we refer to a state  $x\theta\cup s$ (pair $\langle x\theta\cup s, y\theta\cup s\rangle$) as an \emph{extension} of state $x$ (pair $\langle x, y\rangle$).
A state $x$ (e.g., in a critical pair) may not be an $I$-state in itself, but some of its extended states may be $I$-states; the other way round, some extensions of an $I$-state may not be $I$-states.
Duck et al~\cite{DBLP:conf/iclp/DuckSS07}
 considered cases where, for each critical pair
$\langle x, y\rangle$
, a collection 
of most general extensions $\{\langle x_i, y_i\rangle\}_{i\in\mathit{Inx}}$ exists,
such that any such $\langle x_i, y_i\rangle$ and any extension of it consists of $I$-states.

For a given program $\Pi$ and invariant $I$, let ${\mathcal M}^{I,\Pi}$ be the set of all
such most general extensions for all critical pairs.
Then the following holds.
\begin{itemize}
  \item A CHR$^0$ program is locally observably confluent w.r.t.\ $I$ if and only if all  pairs in
     ${\mathcal M}^{I,\Pi}$ are joinable.
  \item A terminating CHR$^0$ program is observably confluent w.r.t.\ $I$ if and only if all  pairs in 
  ${\mathcal M}^{I,\Pi}$ are joinable.
\end{itemize}
Decidability is lost, and~\cite{DBLP:conf/iclp/DuckSS07} shows that even a standard invariant such
as groundedness leads to infinite ${\mathcal M}^{I,\Pi}$ sets.
The characterization of ${\mathcal M}^{I,\Pi}$ is 
complicated, and no practically
relevant methods have been proposed. 
In the present paper, we cope with these problems by 
introducing a 
meta-language in which we can reason about abstract versions of critical pairs and their joinability,
and in which the invariant is treated as a meta-level constraint.

We are not aware of other work than our own on confluence for CHR  that includes
non-logical predicates or takes an equivalence relation into account.
Confluence for nonterminating CHR programs has been studied by~\cite{DBLP:journals/tplp/Haemmerle12,RaiserTacchella2007}, and~\cite{DBLP:conf/lopstr/AbdennadherF03} has 
considered
how the integration of two programs known to be confluent can be made confluent by adding
new rules.

The choice of an operational semantics for CHR, i.e., a definition of the derivation relation for CHR,
influences the set of programs recognized as confluent
and the amount of notational overhead
needed for the proofs.
We postpone a comparison with selected other operational semantics until Section~\ref{sec:commentOpSem},
following the introduction of the necessary technical apparatus.

\section{Constraint Handling Rules}\label{sec:CHR}
In the following, we introduce CHR and our new operational semantics as a rewriting system.
 We highlight the differences in comparison with previous
semantics used for the study of confluence for CHR.
Ours differs most essentially in that it can describe
non-logical and incomplete built-ins, and we have also succeeded in introducing several
simplifications without loss of generality (apart from a subtle mathematical consequence
implied by some earlier semantics exposed in Example~\ref{ex:different-confluence}, p.~\pageref{ex:different-confluence}).

\subsection{Preliminaries}\label{sec:preliminaries}
We extend the basic concepts and notation introduced in Section~\ref{sec:rewrite}.
Derivation steps are labelled so we can distinguish how they are produced with reference to the
CHR program in question (letters $D$ and $d$ are typically used for such labels,
indicating a \emph{d}escription of the step).
We also introduce the notions $\alpha$- and $\beta$-corners to give a representation
of cases where the $\alpha$- and $\beta$ conditions may (or may not) hold.

\begin{definition}\label{def:corners-etc}
A \emph{derivation system} $\langle S, D, \ourmapsto, I, \approx\rangle$ consists of a set $S$, called
\emph{states}, a set of \emph{labels} $D$, a ternary \emph{derivation relation}
${\mapsto}\subseteq S\times D\times S$, an \emph{invariant} $I\subset S$,
and an \emph{equivalence} ${\approx}\subseteq S\times S$.

A fact $\langle x,d,y\rangle \in {\mapsto}$ is written $x \stackrel d\ourmapsto y$, in which case we also write
$x \ourmapsto y$, thus projecting it to a binary relation;
as usual 	$\mapstostar$ denotes the reflexive, transitive closure of $\ourmapsto$,
and \emph{derivation} is a successive sequence of zero or more, perhaps infinitely many, derivation steps.
For brevity, we may use $x\mapstostar y$ to indicate a derivation from $x$ to $y$, with labels understood.
The invariant property of $I$ means that $I(x)\land (x\stackrel d \ourmapsto y)$ implies $I(y)$;
a state $x$ with $I(x)$ is an \emph{$I$-state} and \emph{$I$-derivation} (step)s are those that involve only $I$-states.

An \emph{$\alpha$-corner} is a structure of the form $y'\mapsfrom x\mapsto y$ where $x,y,y'$ are states and
$y'\mapsfrom x$, $x\mapsto y$ are derivation steps;
a \emph{$\beta$-corner} is of the form $y'\approx x\mapsto y$ where $x,y,y'$ are states, 
 $y'\approx x$ holds and $x\mapsto y$ is a derivation step.  We may use the symbol $\Lambda$ to denote a corner.
 In both cases, the state $x$ is referred to as the \emph{common ancestor state} for the \emph{wing states} $y'$ and $y$.
Two $\alpha$-corners  $y'\mapsfrom x\mapsto y$ and  $y\mapsfrom x\mapsto y'$  are considered identical.
An $\alpha$- ($\beta$-) corner is called an \emph{$\alpha$- ($\beta$-) $I$-corner} when its states are $I$-states.

A \emph{joinability diagram} (modulo $\approx$) for an $\alpha$- or $\beta$-corner 
$$y'  \mathrel{\mathit{Rel}} x\stackrel {d_2}\ourmapsto y$$
(thus $\mathit{Rel}$ is one of $\stackrel {d_1}\mapsfrom$  or $\approx$) is a  structure
of the form
$$
z'\mapsfromstar   y'  \mathrel{\mathit{Rel}} x\stackrel {d_2}\ourmapsto y   \mapstostar z$$
where
$z'\mapsfromstar   y'$ and $y   \mapstostar z$  are derivations such that the  equivalence
 $z'\approx z$ holds. A diagram is sometimes denoted by the symbol $\Delta$.
A given corner is \emph{joinable} modulo $\approx$ whenever there exists a joinability diagram for it. An $\alpha$-corner of the form $y \mapsfrom x \ourmapsto y$ is called \emph{trivially joinable} (modulo $\approx$).

A derivation system $\langle S, D, \ourmapsto, I, \approx\rangle$ is \emph{confluent modulo $\approx$ (with respect to $I$)}
if and only if, for all $I$-states $y',x,y$: $y' \mapsfromstar x \mapstostar y \imply \exists z,z'.\ y' \mapstostar z'\approx z \mapsfromstar y$, and is \emph{locally confluent modulo $\approx$  (with respect to $I$)} if and only if all its $I$-corners are joinable modulo $\approx$.
\end{definition}
Joinability diagrams may be shown as in Figure~\ref{fig:LocalConfModEq},
and notions of (local) ($I$-) confluence (modulo $\approx$) $I$-termination apply as already introduced. 
We can  reformulate Theorem~\ref{thm:confLconfModEq} as  follows. 
\begin{theorem}\label{thm:I-joinable-corners}
An $I$-terminating derivation system is $I$-confluent modulo $\approx$ if and only if
all its $I$-corners (of type $\alpha$ as well as $\beta$) are
joinable  modulo $\approx$.	
\end{theorem}
We assume standard notions of first-order logic such as predicates, atoms and terms.
For any expression $E$, $\vars(E)$ refers to the set of variables occurring in $E$.
A \emph{substitution} is a mapping from a finite set of variables to terms, e.g., the substitution $[x/t]$ replaces variable $x$ by term $t$. 
For substitution $\sigma$  and expression $E$, $E\sigma$
(or $E\cdot\sigma$)
denotes the expression that arises when $\sigma$ is applied to $E$; composition of two substitutions $\sigma, \tau$
is denoted $\sigma\circ\tau$.
Special substitutions $\failure$ and $\error$ are assumed, the first one representing falsity and the second
one runtime errors; a substitution different from these two is called a \emph{proper substitution}.

Two disjoint sets of \emph{(user) constraints} and  \emph{built-in} predicates are assumed.
Our semantics for built-ins differs from previous approaches
by mapping them immediately to a unique substitution. This makes
it possible to handle non-logical devices such as Prolog's
\texttt{var}/1 and run-time errors as they may arise from incomplete built-ins such as \texttt{is}/2.

An evaluation procedure $\exe$ for built-in atoms $b$ is assumed,
such that $\exe(b)$ is either a (possibly identity) substitution to a subset of $\vars(b)$ or 
one of $\failure$ and  $\error$.
It extends to sequences of built-ins as follows.
$$
\exe((b_1,b_2)) =
\begin{cases}
\exe(b_1) & \text{when $\exe(b_1)\in\{\failure,\error\}$},\\
\exe(b_2\cdot\exe(b_1)) & \text{when otherwise $\exe(b_2\cdot\exe(b_1))$}\\
& \text{~~~~~~~~~~~~~~~~~~~$\in\{\failure,\error\}$},\\
 \exe(b_1)\circ\exe(b_2\cdot\exe(b_1)) & \text{otherwise\label{text-def:exe}}
\end{cases}
$$
A built-in $b$ or sequence of such is \emph{satisfiable}
whenever there exists a substitution $\theta$ such that $\exe(b\theta)$ is a proper substitution.
A subset of built-in predicates are the \emph{logical} ones, whose meaning is given by a
first-order theory $\mathcal{B}$.
For a logical atom $b$ with $\exe(b)\neq\error$, the following conditions must hold. 
\begin{itemize}
  \item  Partial correctness: $\mathcal{B}\models \forall_{\vars(b)}(b \leftrightarrow \exists_{\vars(\exe(b))\setminus\vars(b)}\exe(b))$.
  \item Instantiation monotonicity:  $\exe(b\cdot\sigma)\neq \error$ for all substitutions $\sigma$.
\end{itemize}
A built-in predicate $p$ is \emph{incomplete} if there exists an atom $b$ with predicate $p$ for which $\exe(b)=\error$; any other built-in predicate is \emph{complete}.
Any built-in predicate which is not logical is called \emph{non-logical}.
A \emph{most general instance} of a built-in predicate $p/n$ is
an atom $p(v_1,\ldots,v_n)$ where $v_1,\ldots,v_n$ are new and unused variables. 
The following predicates are examples of built-ins, and
the list can be extended if needed.

\begin{definition}\label{def:built-ins}
The following list of built-in predicates are assumed
with their meaning as indicated; $\epsilon$ is the identity substitution.
\begin{enumerate}
  \item $\exe(t\mathrel{\texttt{=}}t')=\sigma$ where $\sigma$ is a most general unifier of $t$ and $t'$;
  if no such unifier exists, the result is $\failure$.
  \item $\exe(\texttt{true})$ is $\epsilon$.
  \item $\exe(\texttt{fail})$ is $\failure$.
  \item $\exe(\texttt{$t$ is $t'$}) = \exe(t\mathrel{\texttt{=}}v)$
  whenever $t'$ is a ground term that can be interpreted as an arithmetic expression with  value $v$; if no such
  $v$ exists, the result is $\error$.
   \item $\exe(\texttt{$t$ >= $t'$})$ 
   is $\epsilon$ whenever $t,t'$ are ground terms 
   that can be interpreted as arithmetic expressions with values $v,v'$ where $v \geq v'$;
   if such values exist but  $v < v'$, the result is $\failure$; otherwise, the result is  $\error$.
   \item $\exe(\texttt{var($t$)})$ is $\epsilon$ if $t$ is a variable and $\failure$ otherwise.
  \item $\exe(\texttt{nonvar($t$)})$ is $\epsilon$ when $t$ is not a variable and $\failure$ otherwise.
  \item $\exe(\texttt{ground($t$)})$ is $\epsilon$ when $t$ is ground and  $\failure$ otherwise.
  \item $\exe(\texttt{constant($t$)})$ is $\epsilon$ when $t$ is a constant and  $\failure$ otherwise.
  \item $\exe(t\mathrel{\texttt{==}}t')$ is $\epsilon$ when $t$ and $t'$ are identical and $\failure$ otherwise.
  \item $\exe(t\mathrel{\texttt{\char92=}}t')$ is $\epsilon$ when $t$ and $t'$ are non-unifiable and  $\failure$ otherwise.
\end{enumerate}\end{definition}
The first three predicates in Definition~\ref{def:built-ins} above are logical and complete;
``\texttt{is}'' and ``\texttt{>=}'' are logical but not complete.
The remaining ones are non-logical.

For the representation of CHR execution states, we introduce \emph{indices:}
an \emph{indexed set} $S$ is a set of items of the form $i{:}x$ where $i$ belongs to some index set and each such $i$ is unique in $S$.
When clear from context, we may identify an indexed set $S$
with its cleaned version $\{x\mid i{:}x\in S\}$.
Similarly, the item $x$ may identify the indexed version $i{:}x$. We extract the indices by $\id(i{:} x) = i$.

%

\subsection{Operational Semantics}\label{sec:semantics}
The following operational semantics is based on principles introduced
in~\cite{DBLP:conf/lopstr/ChristiansenK14};
it differs from those used in previous work in several ways that we discuss in Section~\ref{sec:commentOpSem} below.

As custom in recent theoretical work on CHR,
we use the  \emph{generalized simpagation} form~\cite{fru_chr_book_2009} as a common
representation for the rules of CHR.
The guards can modify variables that also occur in rule bodies, but not variables
that occur in the constraints matched by the head rules.

\begin{definition}\label{def:rules}
A \emph{rule} $R$ is of the form
$$r\colon\; H_1 \setminus H_2 \;\mathtt{<=>}\; g\mid C,$$
where $r$ is a unique identifier for the rule, $H_1$ and $H_2$ are sequences of constraints, 
forming the \emph{head} of  the rule,
$g$ is a \emph{guard} being a sequence of built-ins,
and $C$ 
is 
a sequence of constraints and built-ins called the \emph{body} of $R$.
Any of $H_1$ and $H_2$, but not both, may be empty.
A \emph{program} is a finite set of rules.

A \emph{most general pre-application instance} of rule $R$ is an indexed variant $R'$ of $R$ containing new
and fresh variables.

An \emph{application instance} of rule $R$ is a structure of the form
$$R'' = R'\sigma = (r\colon\; H'_1\sigma \setminus H'_2\sigma \;\mathtt{<=>}\; g'\sigma\mid C'\sigma)
$$
where $R'$ is a most general pre-application instance,
$\sigma$ is a substitution for the variables of $H'_1, H'_2$ and
  $\exe(g'\sigma)$ is a proper substitution such that\footnote{The condition
  indicates that the guard's substitution is not allowed
  to instantiate the variables in the head part.}
 $$(H_1'\uplus H_2')\sigma = (H_1' \uplus H_2')\sigma\,\exe(g'\sigma).$$
The part $g'$ ($g'\sigma$) is referred to as the \emph{guard of} $R'$ ($R''$). 
The \emph{application record} for  $R'$ ($R''$), denoted $\mathrm{applied}(R')$ ($\mathrm{applied}(R'')$)
is the structure
$$r\, @\, i_1 \ldots i_n$$
where $i_1 \ldots i_n$ is the sequence of indices of $H_1, H_2$ in the order they occur.

A rule is a \emph{simplification} when $H_1$ is empty,
a \emph{propagation} when $H_2$ is empty;
in both cases, the backslash is left out, and for a propagation, the  arrow symbol is written $\mathtt{==>}$ instead.
Any other rule is a \emph{simpagation}.
\end{definition}
Following~\cite{RaiserEtAl2009},
an execution state is defined in terms of a suitable equivalence class  that abstracts away irrelevant details
concerning which actual variables and indices are used.
%

\begin{definition}\label{def:state}
A \emph{(CHR) state representation} is a pair $\langle S, T\rangle$, where
\begin{itemize}
  \item $S$ is a finite, indexed set of atoms called the \emph{constraint store},
  \item $T$ is a set of relevant application records called the \emph{propagation history},
\end{itemize}
where a \emph{relevant} application record is one in which each index refers to an index in $S$.
Two state representations $S_1$ and $S_2$ are \emph{variants}, denoted $S_1\equiv S_2$, whenever one can be obtained
from the other by a renaming of variables and a consistent replacement of indices (i.e., by a 1-1 mapping).
When $\Sigma$ is the set of all state representations, a \emph{(CHR) state} is an element of
$\Sigma/\!_\equiv\cup\{\failure, \error\}$, i.e., an equivalence class in $\Sigma$ induced by $\equiv$ or one
of two special states; applying the $\failure$ ($\error$) substitution to a state yields the $\failure$ ($\error$) state.
To indicate a given state, we may for simplicity mention one of its representations.
A state different from $\failure$ and $\error$ is called a \emph{proper state}.
A \emph{query} $q$ is a conjunction of constraints, which is also identified with an initial state
$\langle q', \emptyset\rangle$ where $q'$ is an indexed version of $q$.

Assuming a fixed program, the function $\allrars$ from constraint stores to the powerset of application records is defined as
\begin{eqnarray*}
\allrars(S) & = & \{ r@i_1\ldots i_n\mid \mbox{$r$ identifies a propagation rule and} \\
 &   & \phantom{\{ r@i_1\ldots i_n\mid\mbox{}} \mbox{$i_1\ldots i_n$
        are indices of constraints to which the rule can apply} \}
\end{eqnarray*}
\end{definition}
To simplify notation when we make statements involving
several states or other entities involving components of states, we may do so referring to selected state representations,
considering recurrence of indices and variables significant.
For example, in the context of a program that includes the rule $r$: \texttt{p} \texttt{==>} \texttt{q}, we consider the following as a true statement.
$$ST = \langle\{1{:}\texttt{p},2{:}\texttt{q}\}, \emptyset\rangle \;\;\land\;\; ST= \langle S,\emptyset\rangle \;\;\land\;\; \allrars(S) = \{r@1\}$$

\begin{definition}\label{def:derivations}
A \emph{derivation step} $\ourmapsto$ from one $I$-state to another can be of two types: by rule application
instance $\stackrel{R}\ourmapsto$ or by built-in  $\stackrel{b}\ourmapsto$, defined as follows.
\begin{description}
  \item[Apply:] $\langle S\uplus H_1\uplus H_2, T\rangle \stackrel{R}\ourmapsto
\langle S\uplus H_1 \uplus \big(C\cdot\exe(g)\big), T'\rangle$\\
whenever there is an application instance $R$
of the form $r\colon\; H_1 \setminus H_2\;\mathtt{<=>}\; g\mid C$ with
$\mathrm{applied}(R)\not\in T$,
and $T'$ is derived from $T$ by
1) removing any application record
having an index in $H_2$ and
2) adding $\mathrm{applied}(R)$ in case
$R$ is a propagation.
\item[Built-in:] $\langle \{b\} \uplus S, T\rangle\stackrel{b}\ourmapsto \langle S, T \rangle\cdot\exe(b)$.
\end{description}
%
\end{definition}
Notice that the removal of application records in \textbf{Apply} steps ensures
that no non-relevant propagation record remains in the new state (i.e., the result \emph{is} a state).

\begin{example}
Consider a program consisting of the following two rules.

\medskip\noindent
\hbox to 2em{$r_1$:\hfil}\verb"p(X) \ q(Y) <=>  X=Y | r(X)."\\
\hbox to 2em{$r_2$:\hfil}\verb"r(X) ==> s(X)."

\medskip\noindent
The following is an application instance of $r_1$.
$$ R_1^{\texttt{a},\texttt{a}}\;\;=\;\; \bigl(r_1\colon 1{:}\texttt{p(a) \char92\ }2{:}\texttt{q(a) <=> a=a |\ }
        3{:}\texttt{r(a)}
  \bigr)$$
\medskip\noindent
It can be used in an \textbf{Apply} derivation step as follows.
$$ \bigl\langle\{1{:}\texttt{p(a)}, 2{:}\texttt{q(a)}\},\emptyset\bigr\rangle\stackrel{R_1^{\texttt{\scriptsize a},\texttt{\scriptsize a}}}\ourmapsto
    \bigl\langle\{1{:}\texttt{p(a)}, 3{:}\texttt{r(a)}\},\emptyset\bigr\rangle$$
However, the indexed instance of $r_1$, $\bigl(r_1\colon 1{:}\texttt{p(Z)} \texttt{\char92}2{:}\texttt{q(a)} \texttt{<=>} \texttt{Z=a} \texttt{|}
3{:}\texttt{r(Z)}\bigr)$ is not an application instance as  the guard, when executed, will bind
the head variable \texttt{Z}.

The rule $r_2$ is a propagation rule, and we show an application instance for it and
an \textbf{Apply} derivation step; here  the propagation history is checked before the step and modified by the step.
$$R_2^{\texttt{a}} \;\;=\;\; \bigl(r_2\colon 1{:}\texttt{r(a) ==>\ }
        4{:}\texttt{s(a)}
  \bigr)$$
$$ \bigl\langle\{1{:}\texttt{r(a)}, 2{:}\texttt{r(b)}, 3{:}\texttt{s(b)}\},\{r_2@2 \}\bigr\rangle\stackrel{R_2^{\texttt{\scriptsize a}}}\ourmapsto
    \bigl\langle\{1{:}\texttt{r(a)}, 2{:}\texttt{r(b)}, 3{:}\texttt{s(b)}, 4{:}\texttt{s(a)}\},\{r_2@2, r_2@1 \}\bigr\rangle$$
\end{example}
The following example shows how an incomplete predicate is treated when occurring in a guard and when executed by 
a \textbf{Built-in} step.
\begin{example}
Consider a program that includes the following rule
having the incomplete  ``\texttt{is}'' predicate in its guard. Furthermore, assume that ``\texttt{is}''
can appear in \textbf{Built-in} steps, i.e., can also appear in a state.

\medskip\noindent
\hbox to 2em{$r_1$:\hfil}\verb"p(X) ==> Y is X+2 | q(Y)."

\medskip\noindent
An attempt to \textbf{Apply} it to some state by matching the head with $1{:}\texttt{p(2)}$
may yield the application instance
$$R_2^{\texttt{a}} \;\;=\;\; \bigl(r_1\colon 1{:}\texttt{p(2) ==> Z is 2+2 |\ }
        7{:}\texttt{q(Z)}
  \bigr).$$
The guard evaluates to the substitution $[\texttt{Z}/\texttt{4}]$ and
the new state includes the instantiated body constraint $7{:}$\texttt{q(4)}.
The rule cannot apply by matching the head with 2{:}\texttt{p(Z)} as
the guard evaluates to the $\error$ substitution -- but no $\error$ state is produced.
A \textbf{Built-in} step, on the other hand, for \texttt{Z} \texttt{is} \texttt{2+A}
leads to the $\error$ substitution (by Definition \ref{def:built-ins}) and in turn to the $\error$ state.
\end{example}
We observe the following immediate consequence of the definition,
namely a functional dependency from a state plus label of a possible step
to the resulting state.

\begin{proposition}\label{prop:deriv-step-label-and-subst}
For any state $\Sigma$ and derivation step label $d$, there is at most
one state $\Sigma'$ such that $\Sigma\stackrel d\ourmapsto\Sigma'$.
\end{proposition}
The following distinctions become useful later when we reason about
derivation steps and the built-ins involved.
As it appears in Definition~\ref{def:derivations} above, built-ins  evaluating to $\error$ (representing runtime error)
are treated differently in the two sorts of derivation steps:
in a guard, $\error$ and $\failure$ both means that a rule cannot apply (corresponding to
no runtime error reported in an implemented system);
when such a built-in (coming from the query or a rule body) is applied
to a state, it gives rise to a derivation step leading to the relevant of an $\error$ or a $\failure$ state.

\begin{definition}\label{def:state-built-in-etc}
In the context of a state invariant $I$, a built-in predicate is a \emph{state built-in predicate}
whenever it can appear in an $I$-state.
A logical built-in predicate $p$ is \emph{$I$-complete} whenever $\exe(b)\neq\error$ for any atom $b$ with predicate $p$
that may occur in an $I$-state or in the guard of an application instance that can
apply to an $I$-state.

A guard in a rule is \emph{logical} if it contains only logical predicates; otherwise, it is \emph{non-logical}.
A logical guard is  \emph{$I$-complete} if it contains only $I$-complete predicates;
otherwise, it is \emph{$I$-incomplete}.
\end{definition}
\begin{example}
The built-in ``\texttt{is}/2'' is logical and, while incomplete, it is $I$-complete with respect to an invariant that guarantees
the second argument to be a ground arithmetic expression.
\end{example}

\subsection{A Few Comments on Earlier Operational Semantics for CHR}\label{sec:commentOpSem}
Our operational semantics for CHR  differs from other known and formally specified ones
e.g.,~\cite{DBLP:conf/cp/Abdennadher97,DBLP:conf/cp/AbdennadherFM96,DBLP:journals/constraints/AbdennadherFM99,DuckSBH04,DBLP:conf/iclp/DuckSS07}
by handling  also non-logical and incomplete built-ins.

We do that by ``executing'' built-ins immediately in terms of  substitutions
applied to the state, which we claim is more compatible with practical CHR systems than earlier approaches;
Apt et al's semantics for Prolog with such predicates~\cite{DBLP:journals/aaecc/AptMP94} applies the same principles (with the small difference that they do not distinguish between error and failure).
The referenced approaches use instead a separate store for built-in constraints (restricted to logical ones)
that have been
processed, with their satisfiability determined by a magic solver that mirrors a first-order semantics;
this excludes the possibility to consider runtime errors and non-logical and incomplete predicates.
The following example highlights the difference.

\begin{example}
Assume a program that includes the following program rule, and assume that \texttt{>=} is also a state built-in predicate.
\begin{verbatim}
p(X) <=> X >= 1 | r(X)
\end{verbatim}
We can point out the difference by the query \texttt{A>=2,} \texttt{p(A)}.
Starting from an empty built-in store (\texttt{true}), the semantics of~\cite{DBLP:conf/cp/Abdennadher97,DBLP:conf/cp/AbdennadherFM96,DBLP:journals/constraints/AbdennadherFM99,DuckSBH04,DBLP:conf/iclp/DuckSS07}
may first ``execute'' \texttt{A>=2} by adding it to the built-in store and keeping \texttt{p(A)} as the remaining query.
Then the program rule above can apply for \texttt{p(A)} -- since the truth of the guard is implied by the built-in store, thus leaving a final constraint store $\{\texttt{r(A)}\}$ constrained by the built-in store $(\texttt{A>=2})$.

With our semantics, the rule cannot apply (as the guard evaluates to $\error$ which is treated the same
way as $\failure$), and evaluating \texttt{A>=2} as part of the query results in the final state $\error$.
\end{example}
The test in our semantics (Definition~\ref{def:rules}) that
prevents a rule from being applied if it otherwise would modify variables in the
constraints matched by the rule head, is implicit in~\cite{DBLP:conf/cp/Abdennadher97,DBLP:conf/cp/AbdennadherFM96,DBLP:journals/constraints/AbdennadherFM99,DuckSBH04,DBLP:conf/iclp/DuckSS07}.
Consider, for example, a rule \texttt{p(X)} \texttt{<=>} \texttt{X=a} \texttt{|}$\cdots$ considered
for the query atom \texttt{p(A)}, assuming a built-in store $B$.
Here the test in the guard would amount to the condition $B\models\forall\texttt{A}.\, \texttt{A}=\texttt{a}$;
this is false when, say $B$ is empty, and holds only when $B$ implies that \texttt{a} is the only possible value for \texttt{A}.

The interpretation of runtime error in guards as failure
is
described by~\cite{DBLP:journals/aai/HolzbaurF00a} for one of the first widespread CHR compilers,
released with earlier versions of SICStus Prolog.
The documentation for the now dominant compiler~\cite{SchrijversDemoen2004} embedded in recent versions of SICStus Prolog
and SWI Prolog\footnote{See \url{http://www.swi-prolog.org}; version 7 checked February 2016.}
is not explicit about this point. A test of SWI Prolog shows that a runtime error in a guard makes the entire
computation  terminate with an error message.
While this limits the completeness results for CHR, it has been chosen for efficiency
reasons (and the fact that the guard can be reformulated to obtain error as failure if necessary)~\cite{personalCommTomSchrijversFeb2016}.

%
%
Furthermore, we disregard so-called global variables
defined as those that appear in the original query.
The mentioned previous approaches introduce a separate state component
to memorize global variables, but this can be shown unnecessary. 
Consider a query \texttt{q(X)};
we translate it into  \texttt{q(X),} \texttt{global('X',X)} where \texttt{'X'} is a constant that serves as the name of
variable global variable \texttt{X}. When a derivation terminates in a proper state, it includes the constraint
\texttt{global('X',$\mathit{val}$)} where $\mathit{val}$ is the value computed for variable \texttt{X}. 

The mentioned semantics uses  a separate state component,
that we will call the \emph{queue}, to hold constraints that
have not yet been entered into the ``active'' constraint store.
Constraints appearing in the body of a rule being applied are first entered into the queue,
and then from time to time moved into the active store by a separate sort of derivation step. Rules are applied by matching constraints within the active store. 
This separation may be relevant as a starting point for imposing strategies for
ordering  application of rules and  search for constraints to be processed (which
is one of the goals of~\cite{DuckSBH04}),
but for studies of confluence it is irrelevant as the set of derivations with or without this
additional mechanics is essentially the same.

Our semantics  has only two state components, in comparison
to, e.g.,~\cite{DBLP:conf/cp/Abdennadher97,DBLP:conf/cp/AbdennadherFM96,DBLP:conf/iclp/DuckSS07}
that need five state components when considering more restricted confluence problems.

There are also differences in how to avoid the potential looping by propagation rules applying
to the same constraints over and over again.
Our semantics (and some others not referenced)
hold a set of records telling which propagations must not apply (because they have
been applied already), while~\cite{DBLP:conf/cp/Abdennadher97,DBLP:conf/cp/AbdennadherFM96,DBLP:conf/iclp/DuckSS07} maintain a set of permissions for
those propagations that may apply.
There is essentially only a notational difference between the two, and the choice is a matter of taste.
An alternative approach is taken by~\cite{DBLP:journals/tplp/BetzRF10},
mixing a set-based and a multiset-based approach:
new constraints produced from the body of a propagation is treated set-wise,
and a propagation is only allowed if it results in adding new constraints.

In earlier work, such as~\cite{DBLP:conf/cp/Abdennadher97,DBLP:conf/cp/AbdennadherFM96,DBLP:conf/iclp/DuckSS07} already discussed, the states  include
specific indices and specific variables.
Thus any reasonable definition of joinability and confluence needed to mention
an equivalence relation
telling two states equivalent if they differ only in systematic replacement of variables and indices
(quite similar to our $\equiv$ in Definition~\ref{def:state}, p.~\pageref{def:state}).
So in some sense, these approaches concern confluence modulo equivalence problems,
but for a very specific equivalence hardcoded into the  proofs
of general properties.\footnote{The same can be said about \cite{DBLP:journals/jfp/Niehren00,DBLP:conf/ccl/NiehrenS94},
studying confluence in a completely different setting.}
In 2009, the paper~\cite{RaiserEtAl2009} gave a satisfactory solution to this problem,
abstracting away concrete indices and variables defining a state as an equivalence
class modulo such a relation $\equiv$, exactly as we have shown 
in our Definition~\ref{def:state} above.

\section{Confluence Modulo Equivalence for CHR}\label{sec:conf-mod-eq-in-CHR}
Here we adapt classical definitions of critical pairs and associated properties for CHR to
include non-logical and incomplete built-ins, as well as an invariant and an equivalence relation.

For the strictly logical case with no invariant,~\cite{DBLP:conf/cp/Abdennadher97}
defines critical pairs consisting of CHR states that may be shown joinable by ordinary CHR derivations.
This is not viable in our more general case as our analogous construction may lead to pairs that do not satisfy the invariant and from which no
derivations are possible (although the set of all relevant instances thereof may be joinable at the level of CHR).
As a first step towards our meta-level counterpart of critical pairs, we introduce
what we call most general critical pre-corners having a CHR state serving as a common ancestor.

We use the following subcategorization, introduced by \cite{DBLP:conf/lopstr/ChristiansenK14}, of  $\alpha$- and $\beta$-corners according to the sorts of derivation steps 
involved,
as
they need to be treated differently.
The earlier results on confluence for CHR concern
only $\alpha_1$-corners.

\begin{definition}\label{def:critical}
Assume a program with equivalence $\approx$ and invariant $I$.
Let $\Lambda_\alpha= (y\stackrel{\gamma}\mapsfrom x \stackrel{\delta}\ourmapsto y')$ be an $\alpha$-$I$-corner and $\Lambda_\beta= (y\approx x \stackrel{\delta}\ourmapsto y')$ a $\beta$-$I$-corner.
\begin{itemize}
  \item  $\Lambda_\alpha$ is  an \emph{$\alpha_1$-$I$-corner} whenever
  $\gamma$ and $\delta$ are rule application instances.
  \item $\Lambda_\alpha$ is  an \emph{$\alpha_2$-$I$-corner} whenever
  $\gamma$ is a rule application instance and $\delta$ a built-in.
  \item $\Lambda_\alpha$ is  an \emph{$\alpha_3$-$I$-corner}
   $\gamma$ and $\delta$ are built-ins.
  \item  $\Lambda_\beta$ is  a \emph{$\beta_1$-$I$-corner} whenever
         $\delta$  is a rule application instance.
  \item $\Lambda_\beta$ is  a \emph{$\beta_2$-$I$-corner} whenever
 $\delta$ a is built-in.
\end{itemize}
The ``-$I$-'' part of the names may be left out when clear from context.
\end{definition}

\subsection{Most General Critical Pre-corners}
According to Proposition~\ref{prop:deriv-step-label-and-subst},
the end state of a derivation step 
 is functionally dependent on
the initial state, and we employ this in the following definition
of most general critical pre-corners in which we leave out wing states.
The reason why we refer to these artefacts as \emph{pre-}corners
is that they may not be corners at all; when the wing states are attempted to be filled in, 
the guards or invariant may not be satisfied.

\begin{example}\label{ex:motivate-pre-corners}
Consider the following program which has non-logical guards;
assume $\approx$ being identity and $I(\cdot)=\true$.

\medskip\noindent
\hbox to 2em{$r_1$:\hfil}\verb"p(X) <=> var(X)    | q(X)."\\
\hbox to 2em{$r_2$:\hfil}\verb"p(X) <=> nonvar(X) | r(X)."\\
\hbox to 2em{$r_3$:\hfil}\verb"q(X) <=> r(X)."

\medskip\noindent
The equality predicate \texttt{=}/2 is here regarded as a state built-in predicate,
i.e., it may appear in a query;
the meaning of this and the other built-ins is
given in
Definition~\ref{def:built-ins}, p.~\pageref{def:built-ins}.
The following is an attempt to construct an $\alpha_2$-corner; there are no propagation rules, so we leave the empty propagation history implicit and identify states by multisets of constraints.
$$ \Lambda\; = \;\bigl(\{ \texttt{r(Z)}, \texttt{X=Y}\} \stackrel{R_2^{\texttt{Z}}}\mapsfrom
            \{ \texttt{p(Z)}, \texttt{X=Y}\}  \stackrel{\texttt{X}=\texttt{Y}}\ourmapsto \{ \texttt{p(Z)}\} \bigr)$$
Here $R_2^{\texttt{Z}}$ is the rule instance $r_2$: \texttt{p(Z)} \texttt{<=>} \texttt{nonvar(Z)} \texttt{|} \texttt{r(Z)};
$R_2^{\texttt{Z}}$ is not an application instance since its guard is false, thus
the hinted derivation step does not exist, and $\Lambda$ is not a corner.
However,  any substitution $\theta$  that grounds $Z$  will lead to a corner.
With $\texttt{Z}\theta=\texttt{a}$, $\Lambda\theta$ is a corner, whereas if
in addition $\texttt{X}\theta=\texttt{b}, \texttt{Y}\theta=\texttt{c}$ we need to replace the right-most derivation step $\cdots\stackrel{\texttt{X}=\texttt{Y}}\ourmapsto \{ \texttt{p(Z)}\!\}$ by
$\cdots\stackrel{\texttt{b=c}}\ourmapsto\failure$.
\end{example}
The following definition of most general critical pre-corners is lengthy as it has one case for each
sort of corners, but it is straightforward when a few elements have been explained.
For $\alpha_1$, the common ancestor state is constructed  from two rules
such that the application of one prevents the subsequent application of the other.
The symbol ``$\circ$'' is an arbitrary placeholder that visually indicates the presence of
some state.

Propagation histories notoriously 
introduce extra
notation and technicalities,
that are explained following the definition.
We recall Definition~\ref{def:state}, that
$\allrars(S)$  is the set
of all application records for rules of the current program taking indices only
from the constraint store $S$.

\begin{definition}[Most General Critical Pre-Corners]\label{def:pre-corners} Assume a program with equivalence $\approx$ and invariant $I$.\footnote{Notice that only $\alpha_2$ and $\alpha_3$ refers to $I$, but we maintain the -$I$- syllable in all cases for homogeneity.}\\
{\large{$\alpha_1$:}}\quad A \emph{most general critical $\alpha_1$-$I$-pre-corner} is a structure of the form
$( \circ\stackrel{R_1\sigma}\mapsfrom
            \langle H_1\sigma\cup H_2\sigma, T\rangle \stackrel{R_2\sigma}\ourmapsto\circ)$
where
\begin{itemize}
\item $R_k=(r_k\colon\;A_k \setminus B_k  \;\mathtt{<=>}\; g_k\mid C_k)$, $k=1,2$,
   are two most general pre-application instances;
\item  let $H_k = A_k \uplus B_k$ and 
   assume two nonempty sets $H'_k \subseteq H_k$ such that the set of indices used in $H'_1$ and $H'_2$ are identical and all other indices
  in $R_1,R_2$ are unique, and let $\sigma$ be a most general unifier of $H'_1$ and (a permutation of) $H'_2$;
\item if $r_1=r_2$, we must have $A_1\sigma \neq A_2\sigma$ or $B_1\sigma \neq B_2\sigma$;
\footnote{We exclude cases with
              $r_1=r_2$ where the rule applies the same way for both derivation steps, i.e.,
              $A_1\sigma = A_2\sigma$ and $B_1\sigma = B_2\sigma$, as the two wing states in any
              subsumed corner would be identical and thus trivially joinable.}
\item $B_1 \sigma \cap H'_2\sigma \neq \emptyset$ or $B_2 \sigma \cap H'_1\sigma \neq \emptyset$;
\item $ T = \allrars(H_1\sigma\cup H_2\sigma)
                \;\setminus\; \{r_1@\id(A_1B_1), r_2@\id(A_2B_2) \} 
         $;
\item there exists a substitution $\theta$ such that, for $k=1,2$, $\exe(g_k\sigma\theta)$ is a proper substitution and $H_k \exe(g_k\sigma\theta)=H_k.$
\end{itemize}
{\large{$\alpha_2$:}}\quad A \emph{most general critical $\alpha_2$-$I$-pre-corner} is a structure of the form
$(\circ \stackrel{R}\mapsfrom
            \langle A\uplus B\uplus\{b\}, T\rangle \stackrel{b}\ourmapsto\circ)$
where
\begin{itemize}
\item $R=(r\colon\;A \setminus B  \;\mathtt{<=>}\; g\mid C)$  is a most general pre-application instance whose guard
$g$ is non-logical or $I$-incomplete;
\item  there exists a substitution $\theta$ such that $\exe(g\theta)$ is a proper substitution,
$\vars(G\theta)\cap\vars(b\theta) \neq \emptyset$, and $H\exe(g\theta)=H$, where
        $H= A \uplus B$;
\item $b$ is a most general instance of a state built-in predicate (i.e., all arg's are fresh variables);
\item $ T = \allrars(A\uplus B)
                \;\setminus\; \{r@\id(AB) \} 
         $.
\end{itemize}
{\large{$\alpha_3$:}}\quad A \emph{most general critical $\alpha_3$-$I$-pre-corner} is a structure of the form
$(\circ \stackrel{b_1}\mapsfrom
            \langle\{b_1,b_2\}, \emptyset\rangle \stackrel{b_2}\ourmapsto\circ)$
where
\begin{itemize}
\item $b_k$, $k=1,2$, are  indexed, most general instances of state built-in predicates, $b_1$ being
     non-logical or $I$-incomplete.
\end{itemize}
{\large{$\beta_1$:}}\quad When $\approx\neq=$, a \emph{most general critical $\beta_1$-$I$-pre-corner} is a structure of the form
$(\circ \approx
           \langle A\uplus B, T\rangle \stackrel{R}\ourmapsto\circ)$\\
\hbox to 2.5em{}where \begin{itemize}
\item $R=(r\colon\;A \setminus B  \;\mathtt{<=>}\; g\mid C)$ is  a most general pre-application instance
whose guard $g$ is satisfiable;
\item $ T =\allrars(A\uplus B)
                \;\setminus\; \{r@\id(AB) \} 
         $.
\end{itemize}
{\large{$\beta_2$:}}\quad When $\approx\neq=$, a \emph{most general critical $\beta_2$-$I$-pre-corner} is a structure of the form
$(\circ  \approx
           \langle \{b\}, \emptyset\rangle \stackrel{b}\ourmapsto\circ)$
where
\begin{itemize}
\item $b$ is a most general instance of a state built-in predicate.
\end{itemize}
Any two most general critical $I$-pre-corners are considered the same whenever they differ only by
consistent renaming of indices and variables and swapping of the left and right parts.
The $I$ part of the names may be left out when clear from context.
\end{definition}
%
The propagation history constructed for $\alpha_1$-pre-corners
is similar to that of earlier work, e.g.,~\cite{DBLP:conf/cp/Abdennadher97},
for building critical pairs.\footnote{It makes only a syntactic difference that~\cite{DBLP:conf/cp/Abdennadher97}
maintains a set of application records for rules that may be applied,
whereas we maintain a set for those that may not be applied.} 
It tells that any other propagation rule, say \textit{Prop}, 
which might accidentally
be applied to constraints in the common ancestor state, is prevented from doing so.
This provides the maximum level of generality of the pre-corner in the sense
that it subsumes (defined below) all concrete corners in which \textit{Prop} can apply as well as those
where it cannot.
The propagation histories for the other sorts of pre-corners can be explained in similar ways.

\begin{example}[continuing Example~\ref{ex:motivate-pre-corners}, p.~\pageref{ex:motivate-pre-corners}]\label{ex:motivate-pre-corners-contd}
The following is an example of a most general critical $\alpha_2$-pre-corner for the rule labelled $r_1$
(whose guard contains \texttt{var}/1) 
and built-in \texttt{=}/2.
$$ \Lambda^{r_1,\texttt{=}}\; = \;\bigl(\circ \stackrel{R^Z_1}\mapsfrom
            \langle\{ \texttt{p(Z)}, \texttt{X=Y}\}, \emptyset\rangle\}  \stackrel{\texttt{X}=\texttt{Y}}\ourmapsto\circ \bigr)$$
Here $R^Z_1$ is the application instance $r_1:$ \texttt{p(Z)} \texttt{<=>} \texttt{var(Z)} \texttt{|} \texttt{q(Z)}.
\end{example}
As opposed to the derivation relation $\ourmapsto$, the equivalence relation is
given in an atomic way, so we need to consider any possible $\beta$-corner as critical,
i.e., its joinability is not a priori given.\footnote{In the concluding section,
we discuss an alternative approach that uses $\gamma$-corners,
as mentioned in Section~\ref{sec:rewrite}, instead of $\beta$-corners,
which makes it possible to subcategorize and perhaps filter away some
of the abstract pre-corners that concern the equivalence.}

By construction, we have the following.

\begin{proposition}\label{prop:precorners-finite}
For any given program with invariant $I$ and equivalence $\approx$, the set of most general
critical $I$-pre-corners is finite.
\end{proposition}
As mentioned, most general
critical pre-corners are intended to provide a finite characterization of the set
of actual corners that are not per se joinable.
To express this, we introduce the following notion of subsumption.

\begin{definition}[Subsumption by Most General Critical Pre-Corners]\label{def:subsumption-by-gcpc}
Let 
$\Lambda=(\circ\mathrel{\mathit{Rel}_1}\langle S, T\rangle \mathrel{\mathit{Rel}_2}\circ)$
be a most general critical pre-corner.
An $I$-corner $\lambda =(\langle s_1,t_1\rangle \mathit{rel}_1 \langle s,t\rangle \mathit{rel}_2  \langle s_2,t_2\rangle)$
is \emph{subsumed by} $\Lambda$, written $\Lambda<\lambda$,
whenever there exists a substitution $\theta$, a set of indexed constraints $s^+$ and sets of application instances
$t^+$ and $t^\div$ such that
\begin{itemize}
\item $s=S\theta\uplus s^+$,
\item $t=T\uplus t^+ \setminus t^\div$,
\item $t^+\subseteq\allrars(S\theta\uplus s^{+})\setminus\allrars(S\theta)$\\(i.e., a set of application records, each containing an 
index in $s^+$),
\item $t^\div\subseteq  \allrars(S\theta)$
\item $\mathit{rel}_k=\mathit{Rel}_k\,\theta,\quad k=1,2$.
\end{itemize}
If, furthermore, $\Lambda$ is an $\alpha_2$-pre-corner
$(\circ \stackrel{R}\mapsfrom
            \langle\Sigma, T\rangle \stackrel{b}\ourmapsto\circ)$,
where $R$ has guard $g$, and $b$ a built-in,
it is required that $\vars(g\theta)\cap\vars(b\theta)\neq\emptyset$.
\end{definition}
This definition is guilty in a slight abuse of usage due to the additional
requirements for $\alpha_2$
in that only ``really critical'' instances
of the pre-corners are counted:
if the indicated variable overlaps are not observed, the two derivation steps commute so that joinability 
is guaranteed.

\begin{example}[continuing Examples~\ref{ex:motivate-pre-corners}, \ref{ex:motivate-pre-corners-contd}]\label{ex:motivate-pre-corners-contd-contd}
Consider the following  $\alpha_2$-corner for the program given  in Example~\ref{ex:motivate-pre-corners},
$$
\lambda \;=\; \bigl( \langle\{\texttt{q(A)}, \texttt{A=a}\}\cup S, T\rangle
    \stackrel{R_1^\texttt{\scriptsize{A}}}\mapsfrom
        \langle\{\texttt{p(A)}, \texttt{A=a}\}\cup S, T\rangle
        \stackrel{\texttt{A=a}}\ourmapsto
        \langle\{\texttt{p(a)}\}\cup S[\texttt{A}/\texttt{a}], T\rangle  \bigr)
 $$
 where ${R_1^\texttt{\normalsize{A}}}$ is the rule instance ($r_1:$ \texttt{p(A)} \texttt{<=>} \texttt{var(A)} \texttt{|} \texttt{q(A)})
 and $S$ ($T$)  a suitable set of indexed constraints (application records).
 It appears that $\lambda$ is subsumed by the most general critical $\alpha_2$-pre-corner
$ \Lambda^{r_1,\texttt{=}}$ introduced in Example~\ref{ex:motivate-pre-corners-contd} above.
To see this, we use the substitution $[\texttt{Z}/ \texttt{A}, \texttt{X}/\texttt{A}, \texttt{Y} / \texttt{a}]$ for $\theta$ in Definition~\ref{def:subsumption-by-gcpc} above,
and check that $\vars(\texttt{var(Z)}\theta)=\{\texttt{A}\}$ and $\vars((\texttt{X=Y})\theta)=\{\texttt{A}\}$ do overlap.
 \end{example}
The following adapts the Critical Pair Lemma~\cite{DBLP:journals/jacm/Huet80,KnuthBendix1970} known from term rewriting 
(and implicit in previous work on confluence for CHR)
to our setting.
\begin{lemma}[Critical Corner Lemma]\label{lem:critCorner}
Assume a program with invariant $I$ and equivalence relation $\approx$,
and let $\lambda$ be an $I$-arbitrary corner.
Then it holds that either
\begin{itemize}
  \item $\lambda$ is $I$-joinable modulo $\approx$, or
  \item $\lambda$ is subsumed by a most general critical pre-corner.
\end{itemize}
\end{lemma}
The proof which is straightforward but lengthy can be found in the appendix.

This leads to the following central theorem.

\begin{theorem}[Critical Corner Theorem]\label{thm:critCorner}
Assume a program $\Pi$ with invariant $I$ and state equivalence relation $\approx$.
Then
  $\Pi$ is locally confluent modulo $\approx$
if and only if
  all $I$-corners subsumed by some most general critical pre-corner for $\Pi$ are joinable.
\end{theorem}
\begin{proof}
The ``only if'' part: Assume the opposite, that $\Pi$ is locally confluent and that there is an $I$-corner 
$\lambda$ subsumed by some critical pre-corner for $\Pi$  which is not joinable modulo $\approx$.
According to Lemma~\ref{lem:critCorner}, $\lambda$ must be joinable; contradiction. 
The ``if'' part follows immediately from Lemma~\ref{lem:critCorner}:
let $\lambda$ be an $I$-corner; if $\lambda$ is subsumed by some critical pre-corner for $\Pi$
we are done by assumption; otherwise the lemma states that it is joinable.
\end{proof}
Combining this result with~Theorem~\ref{thm:I-joinable-corners}, p.~\pageref{thm:I-joinable-corners},
we get the following.

\begin{theorem}\label{thm:conflu-and-critCorner}
Assume a terminating program $\Pi$ with invariant $I$ and state equivalence relation $\approx$.
Then
$\Pi$ is  confluent modulo $\approx$
if and only if
all $I$-corners subsumed by some critical pre-corner for $\Pi$ are joinable.
\end{theorem}

\subsection{Relationship with Earlier Approaches to Proving Confluence}\label{sec:earlier-CHR-confl-etc}
In the following, we reformulate earlier results of Abdennadher et al~\cite{DBLP:conf/cp/Abdennadher97,DBLP:conf/cp/AbdennadherFM96,fruehwirth-98}
for confluence without equivalence and invariant for the purely logical subset of CHR
 and those of Duck et al~\cite{DBLP:conf/iclp/DuckSS07}, who extended with an invariant, as we have described in Section~\ref{sec:backgroundConfCHR}.
Their critical pairs are similar to our most general critical $\alpha_1$-pre-corners, and  the other sorts of corners become either trivially joinable
or non-existing in these special cases.

In order to describe these results, we complement the notion of subsumption introduced above in Definition~\ref{def:subsumption-by-gcpc} with a subsumption ordering for $I$-corners.
\begin{definition}[Subsumption Ordering for $\alpha_1$-$I$-corner]\label{def:subsumption-by-a1icorner}
Assume a program with 
invariant $I$, and 
let $\lambda = (\langle s_1,t_1\rangle
\stackrel{r_1}\mapsfrom
\langle s_0,t_0\rangle 
\stackrel{r_2}\ourmapsto
\langle s_2,t_2\rangle)$ and 
$\lambda' = (\langle s_1',t_1'\rangle \stackrel{r'_1}\mapsfrom \langle s'_0,t'_0\rangle \stackrel{r'_2}\ourmapsto  \langle s'_2,t'_2\rangle)$
 be 
 $\alpha_1$-$I$-corners. 
We say that $\lambda$ \emph{subsumes} $\lambda'$ denoted $\lambda \preceq \lambda'$
whenever there exist a substitution $\theta_k$, a set of indexed constraints $s_k^+$ and sets of application instances
$t^+_k$ and $t^\div_k$ for $k=0,1,2$ such that
\begin{itemize} 
\item $s'_k = s_k\theta_k\uplus s_k^+$,
\item $t'_k = t_k\uplus t_k^+ \setminus t_k^\div$,
\item $t_k^+$ is a set of application records, each containing at least one index appearing in $s_k^+$,
\item $t_k^\div\subseteq  \allrars(S_k)$
\item $r'_i=\mathit{r}_i\,\theta_0,\quad i=1,2$.
\end{itemize}
We write $\lambda\prec\lambda'$ whenever $\lambda\preceq\lambda'$ and $\lambda\neq\lambda'$.
\end{definition}
The following property follows immediately by the direct similarity with Definition~\ref{def:subsumption-by-gcpc}.

\begin{proposition}
Let $\Lambda$ be a most general critical $\alpha_1$-pre-corner and $\lambda$ an $\alpha_1$-corner such that $\Lambda<\lambda$.
Whenever $\lambda'$ is a corner with $\lambda\preceq\lambda'$, it holds that $\Lambda<\lambda'$.
\end{proposition}

\begin{example}[continuing Examples~\ref{ex:collect}, p.~\pageref{ex:collect}, and \ref{ex:collect-inv},
p.~\pageref{ex:collect-inv}]\label{ex:collect-corners}
Consider again the single rule program that collects elements into a list, with the invariant
of groundedness plus exactly one \texttt{set} constraint whose argument is a list (we ignore the state equivalence here).
The following shows a most general critical pre-corner, two corners and their mutual ordering;
$R^{L,A}$ stands for an applications instance for the rule in which variables \texttt{L}, \texttt{A}
are replaced by terms $L$, $A$.
\begin{center}
$\circ \stackrel{R^{\texttt{\scriptsize L},\texttt{\scriptsize A}}}\mapsfrom \{ \texttt{set(L)}, \texttt{item(A)}, \texttt{item(B)}\} \stackrel{R^{\texttt{\scriptsize L},\texttt{\scriptsize B}}}\ourmapsto  \circ$\\
\rotatebox[origin=c]{270}{\hbox to 1ex{}{\large$<$}\hbox to 0.5ex{}}\\
${\{\texttt{set([a|c])}, \texttt{item(b)}\} \stackrel{R^{\texttt{\scriptsize [c]},\texttt{\scriptsize a}}}\mapsfrom \{\texttt{set([c])}, \texttt{item(a)}, \texttt{item(b)}\} \stackrel{R^{\texttt{\scriptsize [c]},\texttt{\scriptsize b}}}\ourmapsto  \{ \texttt{set([b|c])}, \texttt{item(a)}\}}$\\
\rotatebox[origin=c]{270}{\hbox to 1ex{}{\large$\preceq$}\hbox to 0.5ex{}}\\
$\{\texttt{set([a|c])}, \texttt{item(b)}, \texttt{item(d)}\} \stackrel{R^{\texttt{\scriptsize [c]},\texttt{\scriptsize a}}}\mapsfrom \{\texttt{set([c])}, \texttt{item(a)}, \texttt{item(b)}, \texttt{item(d)}\}
\qquad\qquad\qquad\qquad\qquad$\\
$\qquad\qquad\qquad\qquad\quad\qquad\qquad\qquad\qquad\qquad\qquad\qquad\qquad\qquad
\stackrel{R^{\texttt{\scriptsize [c]},\texttt{\scriptsize b}}}\ourmapsto  \{ \texttt{set([b|c])}, \texttt{item(a)}, \texttt{item(d)}\}$
\end{center}
\end{example}
Notice in the example above that the  common ancestor in the
pre-corner does not satisfy the invariant and thus there are no derivation steps
possible from it.
However, when this state is instantiated (and perhaps extended with more constraints)
so that the invariant becomes satisfied, the derivation labels 
denote actual application instances
and corners emerge.

We proceed now as Duck et al~\cite{DBLP:conf/iclp/DuckSS07} and identify a collection of $I$-corners for each pre-corner
which together subsumes all relevant $I$-corners as shown in Theorem~\ref{thm:Icorner} below.

\begin{definition}[Minimal and Least Critical $I$-corners]\label{def:minimal-and-least-I-corners}
Assume a program with invariant $I$.
An $\alpha_1$-$I$-corner $\lambda$ is \emph{minimal (for a most general $\alpha_1$-pre-corner $\Lambda$)} whenever
 \begin{itemize}
  \item $\Lambda < \lambda$, and
  \item $\nexists\, \lambda' > \Lambda \colon \lambda' \prec \lambda$.
 \end{itemize}
When, furthermore
\begin{itemize}
  \item $\forall\, \lambda' > \Lambda \colon \lambda \preceq \lambda'$,
 \end{itemize}
 $\lambda$ is a \emph{least $I$-corner for $\Lambda$}.
%
\end{definition}
In Example~\ref{ex:collect-corners} above, the highest placed $I$-corner is minimal
but not least, as other similar corners exist with other choices of constants.

Theorem~\ref{thm:Icorner} below is similar to the central result of~\cite{DBLP:conf/iclp/DuckSS07},
showing that local $I$-confluence follows from joinability of a specific set of minimal $I$-corners.

\begin{lemma}[Existence of Minimal $I$-corners]\label{lem:icorner-precorner}Assume a program with invariant $I$.
For any $\alpha_1$-$I$-corner $\lambda'$  subsumed by a most general $\alpha_1$-pre-corner $\Lambda$, i.e., $\Lambda<\lambda'  $,  there exists a minimal $I$-corner $\lambda$ for $\Lambda$ such that
$\Lambda<\lambda \preceq \lambda'$.
\end{lemma}

\begin{proof}
First of all, we notice that by construction of subsumption, that (**) there cannot exist infinite chains
$\lambda_1\succ\lambda_2\succ\cdots > \Lambda$.

Consider now $\lambda'>\Lambda$. If $\lambda'$ is minimal, we are done; otherwise
(by Definition~\ref{def:subsumption-by-gcpc}, p.~\pageref{def:subsumption-by-gcpc}) there will be a $\lambda_1$ such that
$\lambda'\succ\lambda_1>\Lambda$; if $\lambda_1$ is minimal, we are done;
otherwise there will be a $\lambda_2$ such that $\lambda'\succ\lambda_1\succ\lambda_2>\Lambda$,
and we continue the same way until we reach a minimal $\lambda_n$ with
$\lambda_1\succ\lambda_2\succ\cdots\succ\lambda_n>\Lambda$; due to observation (**) above, this
process will terminate as indicated.
\end{proof}

%

\begin{lemma}[Minimal $I$-corner Lemma; logical case with invariant; trivial $\approx$]\label{lem:critICorner}~\\
Assume a program with logical and complete built-ins, invariant $I$ and state equivalence $=$,
and let $\lambda$ be an $\alpha_1$-$I$-corner. Then it holds that either
\begin{itemize}
  \item $\lambda$ is $I$-joinable, or
  \item $\lambda$ is subsumed by some minimal $\alpha_1$-$I$-corner.
\end{itemize}
\end{lemma}

\begin{proof}
The lemma is a direct consequence of Lemma~\ref{lem:critCorner} and Lemma~\ref{lem:icorner-precorner}.
\end{proof}

\begin{theorem}[Minimal $I$-corner Theorem; logical case with invariant; trivial $\approx$]\label{thm:Icorner}~\\
For a program $\Pi$ with logical and complete built-ins, invariant $I$ and state equivalence relation $=$, the following properties hold.
\begin{enumerate}
  \item $\Pi$ is locally confluent if and only if all its minimal $I$-corners are joinable. \label{thm:Icorner-partI}
  \item When, furthermore, $\Pi$ is terminating, $\Pi$ is confluent if and only if all its minimal $I$-corners are joinable.\label{thm:Icorner-partII}
  \item  A minimal $I$-corner is not necessarily least, and the set of all minimal $I$-corners is not necessarily finite.\label{thm:Icorner-partIII}
\end{enumerate}
\end{theorem}

\begin{proof}
Part \ref{thm:Icorner-partII} follows  from Newman's Lemma and Part \ref{thm:Icorner-partI}. Proof of Part~\ref{thm:Icorner-partI}: 
``$\Rightarrow$'': It follows directly from the assumption of local confluence.
``$\Leftarrow$'': Assume that all minimal $I$-corners are joinable, but the program is not locally confluent, i.e., there exists an $I$-corner $\lambda'$ that is not joinable. From Lemma~\ref{lem:critICorner} we have that any
$\lambda'$ is subsumed by a minimal $I$-corner $\lambda$ which by assumption is joinable.
Part~\ref{thm:Icorner-partIII} is demonstrated by the following Example~\ref{ex:infinite-minimal-corners}.
\end{proof}
In the general case, Theorem~\ref{thm:Icorner} does not provide an immediate recipe for proving local confluence
due to the potentially infinite number of cases.
The following example demonstrates two ways that this may appear.

\begin{example}\label{ex:infinite-minimal-corners}
Consider a  program that includes the following rules.

\medskip\noindent
\hbox to 2em{$r_1$:\hfil}\verb"p(X) <=> q(X)."\\
\hbox to 2em{$r_2$:\hfil}\verb"p(X) <=> r(X)."\\
\hbox to 2em{$r_3$:\hfil}\verb"p(X) <=> X >= 1 | s(X)."

\medskip\noindent
 We assume the invariant 
$$I(\langle S, T\rangle)\Leftrightarrow S \text{ is ground.}$$ 
There are no propagation rules, so we ignore the propagation history and consider a state as a multiset of constraints.
Rules $r_1$ and $r_2$ give rise to a most general critical $\alpha_1$-$I$-pre-corner
$( \circ\stackrel{R_1^{\texttt{X}}}\mapsfrom
            \{\texttt{p(X)}\} \stackrel{R_2^{\texttt{X}}}\ourmapsto\circ)$.
It has the following infinite set of minimal $\alpha_1$-$I$-corners.
$$
\bigl\{\bigl(\{\texttt{q($t$)}\} \stackrel{R_1^{t}}\mapsfrom
            \{\texttt{p($t$)}\} \stackrel{R_2^{t}}\ourmapsto \{\texttt{q($t$)}\}\bigr) \bigm\vert \mbox{$t$ is a ground term} \bigr\}
$$
Obviously, there is no least $\alpha_1$-$I$-corner for this $\alpha_1$-$I$-pre-corner.
This problem was also noticed by Duck et al in their paper on observable confluence~\cite{DBLP:conf/iclp/DuckSS07}.
An additional consequence of our definitions is that a guard with an incomplete predicate may also give rise
to an infinite set of minimal $\alpha_1$-$I$-corners, even when we relax the invariant to equality.
Now, rules $r_1$ and $r_3$ give rise to a most general critical $\alpha_1$-$I$-pre-corner
$( \circ\stackrel{R_1^{\texttt{X}}}\mapsfrom
            \{\texttt{p(X)}\} \stackrel{R_3^{\texttt{X}}}\ourmapsto\circ)$.
It has the following infinite set of minimal $\alpha_1$-$I$-corners.
$$
\bigl\{\bigl(\{\texttt{q($t$)}\} \stackrel{R_1^t}\mapsfrom
            \{\texttt{p($t$)}\} \stackrel{R_3^t}\ourmapsto \{\texttt{r($t$)}\}\bigr) \bigm\vert \mbox{$t$ is a ground term that can be read as a numeral $\ge1$} \bigr\}
$$
\end{example}
The solution that we describe in Section~\ref{sec:abstract}, and which has no counterpart in~\cite{DBLP:conf/iclp/DuckSS07},
is to consider each most general critical pre-corner (of which there are only finitely many) one at a time,
lifted to a meta-level where we can reason about their joinability properties without having
to expand them to a set of minimal $I$-corners. 

A partial version of the classical results by Abdennadher et al~\cite{DBLP:conf/cp/Abdennadher97,DBLP:conf/cp/AbdennadherFM96,fruehwirth-98} 
can be described as a special case of
Theorem~\ref{thm:Icorner} with trivial invariant.

\begin{lemma}[Least I-corner Lemma; logical case with trivial invariant and $\approx$]\label{lem:leastIcorner}
Assume a program $\Pi$ with logical and complete built-ins, 
invariant $I(\cdot)\Leftrightarrow\true$
and state equivalence relation $=$.
The  set of minimal $\alpha_1$-corners is finite and consists of least $\alpha_1$-corners,
each of which is produced from an $\alpha_1$-pre-corner as follows:

\begin{itemize}
  \item for each  $\alpha_1$-pre-corner of the form
$( \circ\stackrel{R_1}\mapsfrom
            \Sigma \stackrel{R_2}\ourmapsto\circ)$
construct the unique $\alpha_1$-corner, $(\Sigma_1\stackrel{R_1}\mapsfrom
            \Sigma \stackrel{R_2}\ourmapsto\Sigma_2)$.
\end{itemize}
\end{lemma}

\begin{proof}
Let us consider an $\alpha_1$-pre-corner
$\Lambda=( \circ\stackrel{R_1}\mapsfrom
            \langle S, T\rangle \stackrel{R_2}\ourmapsto\circ)$.
The indicated $\alpha_1$-corners do exist as the indicated derivation steps exist:
by construction of $\Lambda$, the two application instances $R_1,R_2$ exist
(cf.\ Definition~\ref{def:rules}, p.~\pageref{def:rules}: trivial guards and the
condition of not modifying head constraints guaranteed); and they can apply to $\langle S, T\rangle$
as their head constraints are in $S$, and no application record for $R_1$ or $R_2$ is in $T$.
Referring to Proposition~\ref{prop:deriv-step-label-and-subst}, p.~\pageref{prop:deriv-step-label-and-subst}
(functional dependency $\textit{Ancestor-state}\times\textit{Application-instance}\rightarrow\textit{Result-state}$),
it is sufficient only to consider the common ancestor states of the involved (pre-) corners.

Let now $\lambda$ be an $\alpha_1$-corner as stated in the lemma. First, we show that $\lambda$ is minimal by contradiction,
so assume the opposite, namely that there exists a $\lambda\succ\lambda'>\Lambda$;
let $\langle s',t'\rangle$ refer to the common ancestor state of $\lambda'$.
From $\lambda'>\Lambda$ it follows that $s'=s\theta'\uplus s^{\prime+}$ for some $\theta', s^{\prime+}$,
and from $\lambda\succ \lambda'$ that $s=s'\theta\uplus s^{+}$  for some $\theta, s^{+}$;
thus $s=s\theta\theta' \uplus s^{\prime+}\theta\uplus s^{+}$
and hence $s^{\prime+}= s^{+}=\emptyset$ and $\theta,\theta'$ are renaming substitutions.

In a similar way, we obtain
\begin{itemize}
  \item $t'=t\uplus t^{\prime+} \setminus t^{\prime\div}$ where any index of $t^{\prime+}$ is in $s^{\prime+}=\emptyset$, and thus $t^{\prime+}=\emptyset$, and $t^{\prime\div}\in t$ (cf.~Definitions~\ref{def:pre-corners},~\ref{def:subsumption-by-gcpc}); hence $t'=t \setminus t^{\prime\div}$,
  \item $t=t'\uplus t^+ \setminus t^\div$ where any index of $t^{+}$ is in $s^{+}=\emptyset$, and thus $t^{+}=\emptyset$, and $t^{\div}\in t'$ as above; hence $t=t'\setminus t^\div = t\setminus  t^{\prime\div}\setminus t^\div$.
\end{itemize}
It follows now that $t^\div=t'^\div=\emptyset$ and thus $\lambda=\lambda'$. Contradiction.

It remains to show that $\lambda$ is a least corner for $\Lambda$,
i.e., for any $\lambda'>\Lambda$ it holds that $\lambda'\succ\lambda$.
This follows from the fact that an unfolding of these two statements according to their respective definitions,
inserting the same common ancestor state $\langle s,t\rangle$ of $\Lambda$ and $\lambda$,
yields identical results.
\end{proof}
The most significant difference in the confluence results with the different 
semantics 
appears when guards contain incomplete built-ins.
This implies cases where our semantics cannot apply,
but the previous ones can, and thus local confluence is a stronger property with
those semantics.
\begin{example}\label{ex:different-confluence}
Consider a program consisting of the following rules;
invariant and equivalence are trivial and not considered.

\medskip\noindent
\hbox to 2em{$r_1$:\hfil}\verb"p(X) <=> 1 >= X, X >= -1 | q(X)"\\
\hbox to 2em{$r_2$:\hfil}\verb"p(X) <=> r(X)"\\
\hbox to 2em{$r_3$:\hfil}\verb"q(X) <=> 1 >= X, X >= 0  | r(x)"\\
\hbox to 2em{$r_4$:\hfil}\verb"q(X) <=> 0 >= X, X >= -1 | r(x)"

\medskip\noindent
As discussed in Section~\ref{sec:commentOpSem} and Example~\ref{ex:motivate-pre-corners} above, the semantics of~\cite{DBLP:conf/cp/Abdennadher97,DBLP:conf/cp/AbdennadherFM96,DBLP:conf/iclp/DuckSS07,fruehwirth-98}
include a built-in store in  the state, and a rule can fire when its guard is a consequence of the current
built-in store. The built-in store for the common ancestor state of a critical pair is formed by the conjunction of the guards of the involved rules; ignoring global variables and propagation history, we obtain in the mentioned semantics the following critical pair, here shown with the ancestor state for ease of comparison.
$$\lambda^*\;=\;\Bigl(\bigl\langle \{\texttt{q(X)}\}, (\texttt{1>=X}\land \texttt{X>=-1})\bigr\rangle\stackrel{r_1}\mapsfrom
     \bigl\langle \{\texttt{p(X)}\}, (\texttt{1>=X}\land \texttt{X>=-1})\bigr\rangle
     \stackrel{r_2}\ourmapsto   \bigl\langle \{\texttt{r(X)}\}, (\texttt{1>=X} \land \texttt{X>=-1})\bigr\rangle\Bigr)
$$
This critical pair is not joinable as neither $r_3$ nor $r_4$ can apply to the left wing state
since their respective guards are not consequences of the current built-in store.
It follows that the program is not confluent when derivations are defined as by~\cite{DBLP:conf/cp/Abdennadher97} and others.

With our semantics, the program is confluent. There are no corners similar to $\lambda^*$
with \texttt{p}/1 having an uninstantiated variable as its argument (the guard of $r_1$ evaluates to $\error$ so $r_1$ cannot apply). Instead we notice an infinite of family minimal corners for $r_1$ and $r_2$, one for each numeral in the
interval $[-1,1]$; for example:
$$\lambda^{\texttt{0.5}}=\bigl(\{\texttt{q(0.5)}\} \stackrel{R_1^{\texttt{0.5}}}\mapsfrom
     \{\texttt{p(0.5)}\}
     \stackrel{R_2^{\texttt{0.5}}}\ourmapsto  \{\texttt{r(0.5)}\}\bigr)
     \quad\!\mbox{and}\quad
     \lambda^{\texttt{-0.5}}=\bigl(\{\texttt{q(-0.5)}\} \stackrel{R_1^{\texttt{-0.5}}}\mapsfrom
     \{\texttt{p(-0.5)}\}
     \stackrel{R_2^{\texttt{-0.5}}}\ourmapsto  \{\texttt{r(-0.5)}\}\bigr)
$$
We see that $\lambda^{\texttt{0.5}}$ can be extended to a joinability diagram by an application of $r_3$
and the same for  $\lambda^{\texttt{-0.5}}$ by $r_4$. Obviously, the entire
family of minimal corners is joinable, and the program is confluent under our semantics.
\end{example}
As mentioned, we do not intend to reason about infinite sets of minimal corners when it can be avoided. The methods introduced in the following section allows reasoning about abstract corners
that visually resemble $\lambda^*$ in Example~\ref{ex:different-confluence} above,
but in which the combined guard constraints are interpreted as meta-level restrictions
on the intended instantiations of the states involved.
(Such abstract corners are allowed to split, so our abstract version of $\lambda^*$
in the example can split into two halves, one shown joinable using $r_3$ and the other by $r_4$.)

%
%
\section{Proving Confluence Modulo Equivalence using Abstract and Meta-level Constrained Corners and Diagrams}\label{sec:abstract}

The classical approach to proving local confluence for CHR~\cite{DBLP:conf/cp/Abdennadher97,DBLP:conf/cp/AbdennadherFM96} is distinguished by having to consider only a finite number of cases,
each characterized by a critical pair of proper CHR states.
Joinability of each critical pair is then shown by applying CHR rules directly.



As shown above, this does not generalize directly to the more general context with non-logical/incomplete built-in predicates and invariants. 
We introduce abstract states, that embed meta-level constraints derived from the invariant and rule guards, representing exactly the permissible states satisfying these constraints. 
Applications of CHR rules to abstract states are simulated with the meaning that they go only for these permissible states.
%
This makes it possible to describe proofs of local confluence in terms of finitely many proof cases, also for 
 examples where \cite{DBLP:conf/iclp/DuckSS07}
  requires infinitely many. 
  Occasionally, we may need to split a case into sub-cases, each requiring different combinations of CHR rules for showing joinability.

Section~\ref{sec:metachr}  introduces the language \metachr{}, and Section~\ref{sec:abstract-corners-and-confluence} provides our central results
on how confluence modulo equivalence may be shown by considering abstract corners constructed from the  most general critical pre-corners.

\subsection{A Meta-Language for CHR and its Semantics}\label{sec:metachr}
In the following, we assume fixed sets of built-in and constraint predicates, invariant $I$ and state equivalence $\approx$.
The following definition introduces the basic elements of  \metachr\ giving a parameterized representation
for CHR and notions related to its semantics.
Built-in predicates of CHR are lifted into \metachr\ in two ways, firstly by a lifted version of the $\exe$ function
(Section~\ref{sec:preliminaries}, p.~\pageref{text-def:exe}) that is extended with an extra argument intended to hold the entire head of the actual rule instance, so the condition can be checked that a guard of a CHR rule cannot modify variables
in that head.
Secondly, each built-in predicate is represented as a predicate of the same name in \metachr\ expressing
satisfiability of a given atom.
The \metachr\ predicates $\allrars$ and $\commonvars$ introduced below will be used to simulate details of CHR's derivation steps.

Two denotation functions will be defined, first $\denotesgr-$ that maps a ground \metachr{} term
into a specific CHR related object, and next $\denotes-$ that maps a  \metachr{} term parameterized 
by meta-variables into all the objects that it covers (analogous to subsumption above).
Be aware that when  a \metachr\ term is ground, it means that it contains no
\metachr\ variables, although it may denote a non-ground CHR related object. 
\begin{definition}\label{def:metachr}
\metachr{} is a typed logical language; for given type $\tau$,  $\meta_\tau$ ($\metagr_\tau$)
refers to the set of (ground) terms of type $\tau$.
The \emph{(ground) denotation function}  for each type $\tau$
is a function
$$\denotesgr-\colon \metagr_\tau\rightarrow\chr_\tau$$
where  $\chr_\tau$ a suitable set of CHR related objects.

The types and terms of \metachr{} are assumed sufficiently rich
such that any relevant object related to CHR is \emph{denotable}, e.g.,
for any CHR state $s$, there exists a ground term $t$
of  \metachr{} with $\denotesgr t=s$.

Whenever $S, S_1, S_2$ are ground \metachr{}
terms of type \textit{state}, \metachr{} includes the following
atomic formulas: invariant statements of the form $I(S)$, equivalence (statement)s of the form
$S_1\approx S_2$.
We assume similarly polymorphic operators for equality and various operations related to sets
such as $\in$, $\subseteq$ etc.
Each such predicate  has a fixed meaning
defined as follows; 
for any sequence of ground \metachr{} $t_1,\ldots,t_n$ of \metachr{} terms of relevant types,
$$ p(t_1,\ldots,t_n)\qquad\mbox{if and only if}\qquad p(\denotesgr{t_1},\ldots,\denotesgr{t_n})$$
Whenever $H$ is some ground term and $G$ a ground term of type \textit{guard},
the formula $\mexe HG$ holds if and only if
\begin{itemize}
  \item    $\exe(\denotesgr G)$ is a proper substitution $\theta$,
     \item  $\vars(\denotesgr H)\cap \mathit{domain}(\exe{\denotesgr E})=\emptyset$.
\end{itemize}   
For each built-in predicate $\texttt{p}/n$ of CHR, \metachr\ includes a \emph{lifted} predicate
$\texttt{p}/n$ whose arguments are of type \textit{term}, and which has a fixed meaning
defined as follows; for ground terms $T_1,\ldots,T_n$,
$\texttt{p}(t_1,\ldots,t_n)$ holds if and only if 
\begin{itemize}
  \item $\exe(\texttt{p}(\denotesgr{t_1},\ldots,\denotesgr{t_n}))$ is a proper substitution.
\end{itemize}
Wheneover $t_1$ and $t_2$ are ground terms, the predicate $\commonvars(t_1, t_2)$ holds if and only
if
\begin{itemize}
  \item $\vars(\denotesgr{t_1})\cap\vars(\denotesgr{t_2})\neq\emptyset$.
\end{itemize}
\metachr{} includes a lifted version of the function $\allrars$ (Def.~\ref{def:state}) from terms of type \textit{constraint-store} to sets of terms of type \textit{application-record}
defined such that $\allrars(s)=t$ if and only if $\allrars(\denotesgr s)=\denotesgr t$.
\end{definition}
For simplicity of notation, we assume
for each predicate and function symbol, a function
symbol in  \metachr{} of similar arity and type \textit{term} for its arguments,
 written with the same symbols.
For example \texttt{p(a,X)} can be read as a ground \metachr{} term, and $\denotesgr{\texttt{p(a,X)}}=\texttt{p(a,X)}$
is a non-ground CHR term.
To avoid ambiguity, \metachr{} variables are written by \textit{italic} letters; this may occasionally clash
the traditional use such letters for mathematical placeholders, and we add explanations when necessary to avoid confusion.

We extend the notational principle of indicating a state by one of its representations to \metachr{} as demonstrated in the following example.
\begin{example}
Assume a  CHR constraint predicate $\texttt{p}/2$. The following equality between
CHR states holds,
$$\bigdenotesgr{\bigl\langle\{1{:} \mathtt{p(X,a)}\}, \{r@1\}\bigr\rangle} = \bigdenotesgr{\bigl\langle\{2{:} \mathtt{p(Y,a)}\}, \{r@2\}\bigr\rangle}$$
and thus the  \metachr{} formula  $$\bigl\langle\{1{:} \mathtt{p(X,a)}\}, \{r@1\}\bigr\rangle = \bigl\langle\{2{:} \mathtt{p(Y,a)}\}, \{r@2\}\bigr\rangle$$
is true.
\end{example}
The following notion of templates will be used for mapping specific CHR related objects, possibly containing variables,
into a representation in $\metachr$ with new \metachr\ variables, so that application of CHR substitutions are simulated by \metachr\ substitutions.
\begin{definition}\label{def:template}
A \emph{template} $T'$ for a CHR related object $t$ (e.g., term, constraint, rule, state, etc.) is a
\metachr\ term formed as follows: 1) find a \metachr{} term $T$
such that $\denotesgr T= t$, and 2) form $T'$ as a copy of $T$ in which all subterms that are names of CHR variables
are replaced systematically by new and unused \metachr{} variables.
\end{definition}
For example, the \metachr{} term \texttt{p($X$,a)} is a template for the CHR atom \texttt{p(X,a)}.
Similar templates have been used in meta-interpreters for logic programs~\cite{ChristiansenTPLP2005,Hill94meta-programming-in-logic},  based on a lifting of the Prolog
text into a meta-level representation in which Prolog unification is simulated by unification at the meta-level.
The following definition is central. It is the basis for defining
meta-level versions of derivations, corners and diagrams parameterized by  \metachr{}
variables that are constrained in suitable ways.

\begin{definition}\label{def:denotation-and-abstraction}
An \emph{abstraction}  of type $\tau$
is a structure of the form
$$A_\tau\where \Phi,$$
where $A_\tau\in\meta_\tau$ and
$\Phi$ is a formula of \metachr{} referred to as a \emph{meta-level constraint}.
The abstraction is  \emph{ground} if and only if $A_\tau$ is ground and $\Phi$ contains no free variables.
In cases where the meta-level constraint is $\true$, we may leave it out to simplify notation,
i.e., $(A_\tau\where\true)$ is written  as $A_\tau$.

The denotation function $\denotesgr-$ is extended to ground abstractions and arbitrary structures
(e.g., application instances, corners and diagrams) containing such,
in the following way.
\begin{itemize}
  \item For any ground abstraction $A_\tau\where \Phi$,
   $$\denotesgr{A_\tau\where \Phi}=
                  \begin{cases}
                     \denotesgr A & \text{whenever $\Phi$ is satisfied}, \\
                  \bot   & \text{otherwise}.
                \end{cases}
    $$
  \item For any structure $s(A_1,\ldots,A_n)$ including ground abstraction $A_1,\ldots,A_n$,
  $$ \denotesgr{s(A_1,\ldots,A_n)} =
                    \begin{cases}
                      \bot  & \text{if, for some $i$, $\denotesgr{A_i}=\bot$}, \\
                       s(\denotesgr{A_1},\ldots,\denotesgr{A_n}) & \text{otherwise}.
                  \end{cases}
  $$
\end{itemize}
An abstraction or structure with abstractions $\mathbf A$ is said to \emph{cover} a concrete object or structure $C$, whenever
there is a grounding \metachr{} substitution $\sigma$ for which $\denotesgr{\mathbf A\sigma}=C\neq\bot$.
The set of all concrete objects or structures covered by $\mathbf A$ is written
$\denotes{\mathbf A}$.

An abstraction or structure with abstractions $\mathbf A$ is \emph{consistent} whenever $\denotes{\mathbf A}\neq\bot$.

Two abstractions or structures with abstractions, $\mathbf S, \mathbf S'$ are \emph{semantically equivalent}
whenever, for any grounding substitution $\sigma$ that $\denotesgr{\mathbf S\sigma}= \denotesgr{\mathbf S'\sigma}$.
An abstraction of type state is referred to as an \emph{abstract state.}
\end{definition}

\begin{example}[Abstract States]
The abstract objects shown below include lifted versions  of the
CHR built-ins \texttt{constant}/1 and \texttt{var}/1 introduced 
in Definition~\ref{def:built-ins} above.
Notice in the lefthand sides that  $a,x$ are variables of \metachr{}.
\begin{eqnarray*}
\bigdenotes{\bigl\langle \texttt{\{p($a$,$x$)}\}, \emptyset \bigr\rangle\where\mexe{\textit{-}}{(\texttt{constant($a$)},\texttt{var($x$)} }}&=&
    \bigdenotes{\bigl\langle \{\texttt{p($a$,$x$)}\}, \emptyset \bigr\rangle \where \texttt{constant($a$)}\land\texttt{var($x$)}}\\
 & = & \bigl\{\bigl\langle\{\texttt{p(a,X)}\}, \emptyset \bigr\rangle , \ldots, \langle \{\texttt{p(b,Y)}\}, \emptyset \rangle , \ldots\bigr\} \\
  & = & \bigl\{\bigl\langle \{\texttt{p($a,x$)}\} , \emptyset \bigr\rangle \mid \mbox{$a$ is a constant, $x$ a variable}\bigr\}\\
\bigdenotes{\bigl\langle \{\texttt{p($a$)}\} , \emptyset \bigr\rangle \where \texttt{var($a$)} \land \texttt{const($a$)}} & = & \emptyset 
\end{eqnarray*}
In the example above, it was possible to 
turn a sequence of built-ins in a guard (the second argument of $\mexe--$)  into a conjunction.
However, this does not hold in general since different orders in a guard with non-logical or incomplete predicates may give different results.
\end{example}
Next, we introduce various building blocks, leading to abstract corners and joinability diagrams.
\begin{definition}[Abstract $=$, $\approx$ and $I$]\label{def:abstract-eqs-etc}
Let $T, T'$ be abstractions of the same type and $S, S'$ abstract states.
An \emph{abstract equality} is a  formula $T=T'$, an \emph{abstract invariant} a formula $I(S)$
and an \emph{abstract equivalence} a formula $S\approx S'$.
Let $e(T_1,\ldots,T_n)$, $n=1,2$ be an arbitrary such formula;
$e(T_1,\ldots,T_n)$ is defined to be true if and only if it 
the following properties hold.
\setbox3=\hbox{\textbf{(Completeness)\hbox to 0.5em{}}}
\begin{itemize}
\item \noindent\hbox to 1\wd3{\textbf{(Soundness)}\hfil}For any $e(t_1,\ldots,t_n)\in\denotes{e(T_1,\ldots,T_n)}$, it holds that
            $e(t_1,\ldots,t_n)$ is true.
\item \noindent\hbox{\textbf{(Completeness)\hbox to 0.5em{}}}For any $i=1,\ldots,n$ and any $t_i\in\denotes{T_i}$
there exists a true atom $e(t_1,\ldots,t_n)$ with\\
\noindent\hbox to 1\wd3{} $e(t_1,\ldots,t_n)\in\denotes{e(T_1,\ldots,T_n)}$.
\end{itemize}
An abstract state $A$ for which $I(A)$ holds is called an \emph{abstract $I$-state.}
\end{definition}
Notice that we do not require the constituents of abstract statements to be consistent, which means that two inconsistent abstract states 
will satisfy an abstract $\approx$ statement. This is convenient for the formulation of the central Theorems~\ref{thm:abstract-CC-theorem} and~\ref{thm:termination+abstract-joinability=confluence}, below.

The following property is useful when we want to build an abstract $\beta$-corner (defined below).
For any  abstract state, we can always produce an equivalence statement that covers all relevant equivalences
at the level of CHR; this is made precise as follows.
\begin{proposition}\label{prop:equiv-state-exists}
For any
abstract state $A$ there exist an abstract equivalence
$A\approx A'$ such that $\denotes{A'} = \{ s'\mid \mbox{$s'\approx s$ for any $s\in\denotes A$}\}$.\end{proposition}
\begin{proof}
Assume an abstract state of the form $S\where\Phi$, $S$ some \metachr\ term of type \textit{state} and $\Phi$ a \metachr\ formula.
The  following abstract equivalence satisfies the proposition, where $S'$ is a new and unused \metachr\ variable.
$$(S\where\Phi) \approx (S' \where S'\approx S\land \Phi).$$
\end{proof}
As seen above, we have overloaded $\approx$ to simplify notation.
In the following example, we will elucidate the different levels of equivalence.
\begin{example}[Abstract Equivalence; Examples~\ref{ex:collect} and~\ref{ex:collect-inv}, continued]
We consider again the program that collects a set of items into a list with the suggested invariant and equivalence.
For simplicity, we consider here only states containing a single \texttt{set} constraint whose argument is a list of constants.
The propagation history is always empty, and we can ignore both that and the indices.
In this example, and only here, we add  subscripts to distinguish
the different versions of the equivalence symbol $\approx$:
$\approx_\textrm{CHR}$,
$\approx_\textrm{\metachr}$, $\approx_\textrm{ABSTRACT}$.
With these remarks, we can specify the equivalence at the level of CHR as follows:
$$\{\texttt{set($\ell_1$)}\} \approx_\textrm{CHR} \{\texttt{set($\ell_2$)}\} \quad\Leftrightarrow\quad \mathit{perm}(\ell_1,\ell_2),$$ where
$\mathit{perm(\ell_1,\ell_2)}$ is an auxiliary predicate that holds if and only of $\ell_1$ and $\ell_2$ are lists
of constants that are permutations of each other.
We will now demonstrate the construction given by the proof of Proposition~\ref{prop:equiv-state-exists}
for an abstract state $(\{\texttt{set([a,b]}\}\where \true)$ that we will write in the short form $\{\texttt{set([a,b]}\}$.
The proposition suggests the following abstract equivalence that holds between the indicated abstract states.
$$\{\texttt{set([a,b])}\} \approx_\textrm{ABSTRACT} (S' \where S'\approx_\textrm{\metachr} \{\texttt{set([a,b])}\})$$
To clarify the meaning of $\approx_\textrm{ABSTRACT}$, we unfold the right and innermost equivalence $\approx_\textrm{\metachr}$ 
according to the definition, assuming a lifting of
$\mathit{perm}/2$ to \metachr, and we get the following.
 $$\{\texttt{set([a,b])}\} \approx_\textrm{ABSTRACT} (\{\texttt{set($\ell$)}\} \where \mathit{perm}(\ell,\texttt{[a,b]}\})
 $$
The right-hand side, is now in a form that makes it easier to apply abstract derivations.
\end{example}
With the abstract states at hand, we can now define abstract derivation steps. 
\begin{definition}[Abstract Derivation Step]\label{def:abstract-derivation-step}
An \emph{abstract derivation step} is an
abstraction of the form $A \stackrel{D}\ourmapsto A'$ where
$A$ and $A'$ are abstract $I$-states, $D$ an abstract label (i.e., abstract built-in atom or abstract application instance),
with $\vars(A')\cup\vars(D)\setminus\vars(A)$ being fresh and unused variables,
such that the following properties hold.
\setbox3=\hbox{\textbf{(Completeness)\hbox to 0.5em{}}}
\begin{itemize}
\item \noindent\hbox to 1\wd3{\textbf{(Soundness)}\hfil}For any $( a\stackrel{d}\ourmapsto a')\in\denotes{ A \stackrel{D}\ourmapsto  A'}$, it holds that
         $a\stackrel{d}\ourmapsto a'$ is a concrete derivation step.
\item \noindent\hbox{\textbf{(Completeness)\hbox to 0.5em{}}}For any $a\in\denotes{ A}$ (for any $a'\in\denotes{A'}$) there exists a concrete derivation step\\
\noindent\hbox to 1\wd3{}$( a\stackrel d\ourmapsto a')\in\denotes{ A\stackrel D\ourmapsto A'}$.
\end{itemize}
\emph{Abstract $I$-derivations} are defined in the usual way as a sequence of zero or more abstract $I$-derivation steps.
\end{definition}
An abstract derivation step is intended to cover a set of concrete derivation steps, but 
unintended variable clashes in the abstract derivation step can cause undesired limitations on those. This is avoided with the requirement of fresh and unused variables.
The second part of the completeness condition is relevant when we compose derivations and diagrams.
If $\denotes{A'}$ includes an element $a'$ for which the indicated step does not exist, $A'$ is so to speak
too big, and the next step (or equivalence statement) from $A'$  would have to take care of too many
irrelevant concrete states.

The following property follows immediately from the Definition~\ref{def:abstract-derivation-step}.

\begin{proposition}[Abstract Derivation]\label{prop:abstract-derivation-covering}
For any 
abstract  $I$-derivation 
$$\Xi  =(A_0\stackrel{D_1}\ourmapsto A_1\stackrel{D_2}\ourmapsto\cdots\stackrel{D_n}\ourmapsto A_n) \quad,\; n\geq0.$$
the following
properties hold.\setbox3=\hbox{\textbf{(Completeness)\hbox to 0.5em{}}}
\begin{itemize}
\item \noindent\hbox to 1\wd3{\textbf{(Soundness)}\hfil}Any element of $\denotes{\Xi}$ is a concrete derivation.
\item \noindent\hbox{\textbf{(Completeness)\hbox to 0.5em{}}}For any $a_0\in\denotes{ A_0}$ (for any $a_n\in\denotes{ A_n}$)  there exists a concrete derivation\\
 \noindent\hbox to 1\wd3{}$( a_0\stackrel{d_1}\ourmapsto a_1\stackrel{d_2}\ourmapsto\cdots\stackrel{d_n}\ourmapsto a_n)\in\denotes{ \Xi}$.
\end{itemize}
\end{proposition}

\begin{proof}
By induction. Base case $n=0$ is trivial; the step follows directly from Definition~\ref{def:abstract-derivation-step}.
\end{proof}
Analogously to concrete corners, abstract corners are constructed using abstract derivations, abstract equivalence and abstract invariants, as follows.
%
\begin{definition}[Abstract Corners]\label{def:variable-regular-beta-corner}
An \emph{abstract $I$-corner} is a structure of the form $(A_1 \mathrel{\mathit{Rel}_1} A \mathrel{\mathit{Rel}_2} A_2)$ 
where $(A_1 \mathrel{\mathit{Rel}_1} A)$ is an abstract $I$-derivation step or abstract equivalence, and
$(A \mathrel{\mathit{Rel}_2} A_2)$ an abstract $I$-derivation step such that
$\vars(A_1)\cap\vars(A_2)\subseteq\vars(A)$.
\end{definition}
We will refer to an abstract corner as  an abstract $\alpha_1$-,  $\alpha_2$-,  $\alpha_3$-, $\beta_1$- or $\beta_2$-corner
according to the  relationships  involved, analogous to what we have done for concrete corners
(Definition~\ref{def:critical}, above).
The following property
is a consequence of what we have shown so far.
\begin{proposition}\label{prop:abstract-corner-covering}
For any 
abstract $I$-corner 
$$\mathbf{\Lambda}  =(A_1 \mathrel{\mathit{Rel_1}} A \mathrel{\mathit{Rel_2}} A_2),$$
the following
properties hold.
\setbox3=\hbox{\textbf{(Completeness)\hbox to 0.5em{}}}
\begin{itemize}
\item \noindent\hbox to 1\wd3{\textbf{(Soundness)}\hfil}Any element of $\denotes{\mathbf{\Lambda}}$ is a concrete corner.
\item \noindent\hbox{\textbf{(Completeness)\hbox to 0.5em{}}}For any $a\in\denotes{ A}$ (for any $a_1\in\denotes{ A_1}$) (for any $a_2\in\denotes{ A_2}$) there exists a concrete corner\\
\noindent\hbox to 1\wd3{}$( a_1 \mathrel{\mathit{rel}_1} a \mathrel{\mathit{rel}_2}  a_2)\in\denotes{\mathbf{\Lambda}}$
\end{itemize}
\end{proposition}
\begin{proof}
The proposition is an immediate consequence of the soundness and completeness conditions in Definitions~\ref{def:abstract-eqs-etc} and~\ref{def:abstract-derivation-step}.
\end{proof}
Finally, we introduce abstract joinability diagrams for abstract corners, allowing us to treat a perhaps infinite set of corners in a single proof case.
%
\begin{definition}
An \emph{abstract joinability diagram} (modulo $\approx$) for an abstract $I$-corner is a structure of the form
$$\mathbf{\Lambda}  =( A_1 \mathrel{\mathit{Rel_1}} A \mathrel{\mathit{Rel_2}} A_2),$$
is a  structure
of the form
$$\boldsymbol{\Delta}=(A'_1 \mapsfromstar A_1 \mathrel{\mathit{Rel_1}} A \mathrel{\mathit{Rel_2}} A_2 \mapstostar A_2')$$
where
$A_1'\mapsfromstar   A_1$ and $A_2   \mapstostar A_2'$  are abstract derivations  such that the abstract equivalence
 $A_1'\approx A_2'$ holds.
A given abstract corner is \emph{(abstractly) joinable modulo $\approx$} whenever there exists an abstract joinability diagram for it.
\end{definition}

\begin{proposition}\label{prop:abstract-diagram-covering}
Let $\boldsymbol{\Delta}$ be an
abstract joinability diagram for an
abstract $I$-corner $\mathbf{\Lambda}$.
Then the following
properties hold.
\setbox3=\hbox{\textbf{(Completeness)\hbox to 0.5em{}}}
\begin{itemize}
\item \noindent\hbox to 1\wd3{\textbf{(Soundness)}\hfil}Any element of $\denotes{\boldsymbol{\Delta}}$ is a concrete joinability diagram.
\item \noindent\hbox{\textbf{(Completeness)\hbox to 0.5em{}}}For any $\lambda\in\denotes{{\mathbf{\Lambda}}}$
 there exists a concrete joinability diagram for $\lambda$ in $\denotes{\boldsymbol{\Delta}}$.
\end{itemize}
\end{proposition}

\begin{proof}
The proposition is an immediate consequence of soundness and completeness properties
given by Propositions~\ref{prop:abstract-derivation-covering} and~\ref{prop:abstract-corner-covering}
and
Definition~\ref{def:abstract-eqs-etc}.
\end{proof}
Combining this with the Critical Corner Theorem, Theorem~\ref{thm:critCorner}, p.~\pageref{thm:critCorner},
we get immediately the following.

\begin{lemma}[Abstract Corner Lemma]\label{thm:abstractCornerLemma}
Assume a program $\Pi$ with invariant $I$ and state equivalence relation $\approx$,
and let $\mathcal D$ be a family of 
abstract $I$-corner that together covers all
concrete corners that are subsumed by some general critical pre-corner for $\Pi$.
Then
  $\Pi$ is locally confluent modulo $\approx$
if and only if
  all diagrams in $\mathcal D$ is joinable.
\end{lemma}
In the following, we consider how to construct a family of abstract corners as required in Lemma~\ref{thm:abstractCornerLemma}.

\subsection{Proving Confluence Modulo Equivalence using Abstract Joinability Diagrams}\label{sec:abstract-corners-and-confluence}
Here we will show how a set of most general critical pre-corners can be lifted to a
set of abstract corners, and we identify necessary and sufficient conditions for
confluence modulo equivalence.
A program with invariant $I$ and equivalence $\approx$ is assumed.

A pre-corner (Definition~\ref{def:pre-corners}, p.~\pageref{def:pre-corners}) is a common ancestor state whose wing states are indirectly characterized 
by their relationships to the ancestor state. Our way to lift it, to be defined below, consists of first lifting
the common ancestor state, and then applying abstract versions of the indicated relationships (i.e., rule application, built-in
or equivalence) to obtain abstract wing states, constrained at the meta-level by restrictions
induced by guards and invariant.

As part of this, we need the following construction, which, for a given abstract ancestor state and type
of derivation step, provides the resulting abstract state.
For convenience, we combine derivation steps and $\approx$.
We recall Proposition~\ref{prop:deriv-step-label-and-subst}, p.~\pageref{prop:deriv-step-label-and-subst},
stating for the concrete case, that there is at most one resulting state for a derivation step.

\begin{definition}~\label{def:post}
Let $A$ be an abstract $I$-state and $\mathit{Rel}$ either an abstract derivation step
or $\approx$.
An \emph{abstract post state for $A$ with respect to $\mathit{Rel}$} is an abstract state $A'$ such that
$A \mathrel{\mathit{Rel}} A'$ holds.
Such a state $A'$ is indicated as $\post(A,\mathit{Rel})$.
\end{definition}
Proposition~\ref{prop:equiv-state-exists}, p.~\pageref{prop:equiv-state-exists},
shows that $\post(A,\approx)$ can be found in a straightforward manner,
although in practice it may be useful to unfold the definition of $\approx$.

For the definition to be useful
for derivation steps, we assume that  $\metachr$ is sufficiently rich as to express
an abstract state $\post(A,\mathit{Rel})$.
One way to obtain this is to include  $\post$ as a function in the language,
whose meaning were defined semantically as indicated in Definition~\ref{def:post},
but it will be more useful to define a procedure that produces an abstract state
in terms of plain \metachr{} predicates and terms.

A general $\post$ procedure that can handle all built-in predicates will be quite complex;
Drabent's~\cite{drabent-report-1997} analysis of a predicate transformer
for unification of arbitrary terms demonstrates this.
However, in many cases, the invariant and the selection state built-ins
reduce the complexity.
In all the examples we have considered, it has been straightforward to produce all necessary post
states by hand; see, e.g., the larger example in Section~\ref{sec:viterbi} below.
%
%

\begin{example}\label{ex:drabent}
Let $\Sigma$ be the abstract state $(\langle \{{\tt p}(x), x \; {\tt is }\; y\}, \emptyset \rangle \mathrel{\where} variable(x))$, and we will construct a state
$\post(\Sigma, x\; {\tt is }\; y)$.

This example is especially tricky as the built-in is incomplete and there are no restrictions on $y$,
so the post state should capture both the $\error$ and proper states.
We solve the problem, suggesting the following state;
we assume two auxiliary \metachr\ predicates
$\mathit{arithmetic}(t)$ being true for any ground $t$ for which $\denotesgr t$ can be evaluated as an arithmetic expression, and $\mathit{eval}(t_1,t_2)$ being true for any ground $t_1,t_2$
for which $\denotesgr{t_1}$ can be evaluated as an arithmetic expression with value $\denotesgr{t_2}$.

\begin{align*}
 \langle \mathit{S}, \emptyset \rangle \mathrel{\where} &\; \big( \mathit{S} = \{{\tt p}(y')\} \land \mathit{eval}(y,y') \land \mathit{arithmetic}(y) \big)\quad\lor \\
 &\; \big( \mathit{S} = \{\mathit{error}\} \land \neg \mathit{arithmetic}(y) \big)
\end{align*}
Here, the meta-variable $x$ has been replaced by $y'$, which represents the value of the arithmetic expression $y$.
Notice that no single rule can apply to this state, and if it happens to arise in
an attempt to produce a joinability diagram, we need to apply the notion of splitting
introduced below in Definition~\ref{def:splitting}.
\end{example}
The following lifting of a pre-corner into an abstract corner is straightforward, although quite detailed
as it includes meta-level versions of conditions for subsumption by pre-corner (Definition~\ref{def:subsumption-by-gcpc}, p.~\pageref{def:subsumption-by-gcpc}).
As we see in our examples, the detailed conditions often reduce to something much simpler,
so the definition below represents, so to speak, worst cases.

\begin{definition}[Lifting Most General Critical Pre-Corners into Abstract Critical $I$-corners]\label{def:lifting-pre-corners}~\\
An \emph{abstract critical $I$-corner for} a most general critical pre-corner $\Lambda = (\circ\mathrel{\mathit{rel}_1} \langle s_0, t_0 \rangle \mathrel{\mathit{rel}_2}\circ)$
is of the form
$$\mathbf{\Lambda}\; =\;
\bigl(\post(A,\mathit{Rel}_1) \;\;\mathrel{\mathit{Rel}_1} \;\; A \;\; \mathrel{\mathit{Rel}_2} \;\post(A,\mathit{Rel}_2)\bigr), $$ 
where $A$ is an abstract state, and $\mathit{Rel}_1, \mathit{Rel}_2$ abstract derivation steps or $\approx$,
specified as follows.
Let firstly
$(\circ\mathrel{\mathit{Rel}_1} \langle S_0, T_0 \rangle \mathrel{\mathit{Rel}_2}\circ)$ be a template (Def.~\ref{def:template}) for $\Lambda$.
The construction of $A$ depends on relationships ${\mathit{Rel}_1}$, ${\mathit{Rel}_2}$, that determine whether the corner is of type $\alpha_1$, $\alpha_2$, etc.

In case  ${\mathit{Rel}_k}$, $k=1,2$ is an application instance, we assume the notation
$ (r\colon\; H_k  \;\mathtt{<=>}\; G_k\mid C_k)$.
The symbols $S^+, T^+, T^\div$ are fresh and unused variables,
and let $S$ stand for $( S_0 \uplus S^+)$ and $T$ for  $T_0 \uplus T^+ \setminus T^{\div}$.
(For the reading of the following, keep in mind that $A, S, T, S_0, T_0, \mathit{Rel}_1, \mathit{Rel}_2, H_k,G_k,C_k$
are not \metachr\ variables, but mathematical placeholders. The predicates used below are elements of \metachr{} (Def.~\ref{def:metachr}))
The common ancestor state $A$ is given as follows for the different cases.
\begin{description}\def\pind{\vrule width 0pt height 2.5ex}\def\ind{\hbox to 8.5em{}}
  \item[$\alpha_1$:\phantom{$, \beta_1$}] $\langle S, T\rangle \where\;
  I(\langle S, T\rangle) \land \mexe{H_1}{G_1} \land \mexe{H_2}{G_2} \land \mbox{}$\\
  \ind
  \pind$T^+ \subseteq \allrars(S_0 \uplus S^+) \setminus \allrars(S_0) \land \mbox{}$\\  
   \ind  
   \pind$T^{\div} \subseteq \allrars(S_0)\setminus\{\mathrm{applied}(\mathit{Rel}_1), \mathrm{applied}(\mathit{Rel}_2) \} $   \\
  
  \item[$\alpha_2$:\phantom{$, \beta_1$}]
  $\langle S, T\rangle \where\; I(\langle S, T\rangle) \land \mexe{H_1}{G_1} \land\mbox{} $\\
  \ind
  \pind$T^+ \subseteq \allrars(S_0 \uplus S^+) \setminus \allrars(S_0) \land\mbox{} $\\
  \ind
   \pind$T^{\div} \subseteq \allrars(S_0)\setminus\{\mathrm{applied}(\mathit{Rel}_1)\} \land \mbox{}$ \\
  \ind
   \pind$\commonvars(G_1,\mathit{Rel}_2) $ \\
  
  \item[$\alpha_3, \beta_2$:]
  $ \langle S, T\rangle \where\; I(\langle S, T\rangle) \land\mbox{} $\\
\ind
  \pind$T^+ \subseteq \allrars(S_0 \uplus S^+) \setminus \allrars(S_0) \land \mbox{}$\\
\ind
   \pind$T^{\div} \subseteq \allrars(S_0) $ \\

  \item[$\beta_1$:\phantom{$, \beta_1$}]
  $\langle S, T\rangle \where\; I(\langle S, T\rangle) \land \mexe{H_1}{G_1} \land\mbox{} $\\
 \ind
  \pind$T^+ \subseteq \allrars(S_0 \uplus S^+) \setminus \allrars(S_0) \land \mbox{}$\\
  \ind
   \pind$T^{\div} \subseteq \allrars(S_0)\setminus\{\mathrm{applied}(\mathit{Rel}_1)\} $ \\

 \end{description}
Whenever $\mathbf{\Lambda}$ is constructed as above, a semantically equivalent
abstract corner $\mathbf{\Lambda}'$ may also be recognized as an
abstract critical $I$-corner for  $\Lambda$.
By \emph{a set of abstract critical corners for program $\Pi$} we mean
a set consisting of one and only one abstract critical corner for  each most general pre-corner for $\Pi$.
\end{definition}
We notice that the set of all abstract critical $I$-corners for a program $\Pi$ is finite  as there exist only a finite number of most general critical pre-corners for $\Pi$,
cf.~Proposition~\ref{prop:precorners-finite}.

\begin{proposition}\label{prop:abstract-critical-corners-finite}
For any given program $\Pi$ with invariant $I$ and equivalence $\approx$, a set of 
abstract critical corners for it is finite.
\end{proposition}

\begin{lemma}[Cover by Abstract Critical Corner $\Leftrightarrow$ Subsumed by Most Gen.~Crit.~Pre-Corner]\label{lem:cover-subsume}
For given program $\Pi$,
let $\mathbf{\Lambda}$ be an abstract critical $I$-corner for a most general critical pre-corner ${\Lambda}$.
Then the set of  $I$-corners covered by $\mathbf{\Lambda}$ is identical to the set
of  $I$-corners subsumed by ${\Lambda}$.
Furthermore,
$$\{\lambda \mid  \mbox{$\exists$ abs.\ crit.\ corner $\mathbf{\Lambda}$ for $\Pi$ .\ $\lambda\in\denotes{\mathbf{\Lambda}}$}\}
   = \{\lambda\mid  \mbox{$\exists$ most gen.\ critical corner $\Lambda$ for $\Pi $ .\ $\Lambda < \lambda$}\}
 $$
\end{lemma}

\begin{proof}
The second part is a direct consequence of the first part.

For the first part, consider firstly an $I$-corner $\lambda$ that is subsumed by a most general critical pre-corner $\Lambda$
(with  notation as
in Definition~\ref{def:lifting-pre-corners} for each case of $\alpha_1$-, $\alpha_2$-, etc.\ corners);
we prove that it is covered by the abstract critical corner $\mathbf{\Lambda}$ for $\Lambda$ as given by the lemma as follows.

Subsumption means (by Definition~\ref{def:subsumption-by-gcpc}, p.~\pageref{def:subsumption-by-gcpc}) that there exists a CHR substitution  $\theta$ and suitable sets  $s^{+}, t^{+}, t^{\div}$ and states $\mathit{post}_1, \mathit{post}_2$ such that
$\lambda = \bigl(\mathit{post}_1 \mathrel{(\mathit{rel}_1\theta)}   \langle s_0\theta \uplus s^+ , \; s_0 \uplus t^{+} \setminus t^{\div} \rangle 
 \mathrel{(\mathit{rel}_2\theta)}  \mathit{post}_2\bigr)$ and the following conditions hold.
\begin{itemize}
\item $s=s_0\theta\uplus s^+$
\item $t=t_0\uplus t^+ \setminus t^\div$
\item $t^+\subseteq\allrars(s_0\theta\uplus s^+)\setminus\allrars(s_0\theta)$
\item $t^\div\subseteq  \allrars(s_0\theta)$
\end{itemize}
Let now $\sigma$ be a \metachr\ substitution such that $\denotesgr{S^+\sigma} = s^+$,
$\denotesgr{T^+\sigma} = t^+$ and $\denotesgr{T^\div\sigma} = t^\div$.
Since $\langle S_0,T_0\rangle$ is a template for $\langle s_0,t_0\rangle$, and $S^+,T^+,T^\div$ do not occur in $\langle S_0,T_0\rangle$,
we can extend $\sigma$ such that
$\denotesgr{\langle S_0,T_0\rangle\sigma}=\langle s_0,t_0\rangle$.
It remains, for each possible case of the corners being $\alpha_1$,  $\alpha_2$,  $\alpha_3$, $\beta_1$ or $\beta_2$, to show
that $\Phi\sigma$ holds where $A = (\langle S,T\rangle\where\Phi)$ (i.e., $\Phi$ stands for the relevant of the alternative, detailed conditions in Definition~\ref{def:lifting-pre-corners}, above).
This can be verified by inspection in each case, referring to 
\begin{enumerate}
  \item the fact that $\lambda$ is an $I$-corner, meaning that the two relationships
$\bigl(\mathit{post}_1 \mathrel{(\mathit{rel}_1\theta)}   \langle s_0\theta \uplus s^+ , \; t_0 \uplus t^{+} \setminus t^{\div} \rangle\bigr)$
and $\bigr( \langle s_0\theta \uplus s^+ , \; t_0 \uplus t^{+} \setminus t^{\div} \rangle 
\mathrel{(\mathit{rel}_2\theta)}  \mathit{post}_2\bigr)$ do hold, i.e., their state arguments each satisfy the additional conditions,
being  the relevant of a derivation step (Definitions~\ref{def:rules} and~\ref{def:derivations}, p.~\pageref{def:rules}-\pageref{def:derivations}) or an equivalence statement, and 
  \item the completeness parts of Definitions~\ref{def:abstract-eqs-etc} and~\ref{def:abstract-derivation-step},
           and the relationship between $\exe$ and $\mexenoargs$ (Definition~\ref{def:metachr}, p.~\pageref{def:metachr}).
\end{enumerate}
The detailed arguments are left out as the $\Phi$ formula in each case is a straightforward lifting of the
similar conditions at the level of CHR. 

The other way round, consider  an $I$-corner $\lambda=(a_1\mathrel{\mathit{rel}_1'} \langle s, t \rangle \mathrel{\mathit{rel}_2'}a_2)$ covered by an abstract critical corner $\mathbf{\Lambda}$
(with  notation as
in Definition~\ref{def:lifting-pre-corners} for each case of $\alpha_1$-, $\alpha_2$-, etc.\ corners);
we prove that it is subsumed by the most general critical pre-corner $\Lambda$ given by the lemma (and notation as in Definition~\ref{def:lifting-pre-corners}) as follows.

Covering means that there exists a grounding $\metachr$ substitution $\sigma$ such that
$\denotesgr{\langle S,T\rangle\sigma}=\denotesgr{\langle S_0\uplus S^+,T\uplus T^+\setminus T^\div\rangle\sigma} = \langle s, t\rangle$,
and, for $k=1,2$, $\denotesgr{\mathit{Rel}_k}=\mathit{rel}'_k$ and $\denotesgr{\post(\mathit{Rel}_k)\sigma}=a_k$,
and the meta-level constraint part of $A\sigma$ holds.

Now $\langle S_0,T_0\rangle$ is defined as a template for $\langle s_0,t_0\rangle$, so there exists a \metachr\ substitution $\sigma'$ such that
$\denotesgr{\langle S,T\rangle\sigma'}= \langle s_0,t_0\rangle$, and thus we can find a CHR substitution $\theta$ such that
$\langle s_0,t_0\rangle\theta=\langle s, t \rangle$.

Let now $s^+=\denotesgr{S^+\sigma}$, $T_0=\denotesgr{t_0\sigma}$,
$T^+=\denotesgr{t^+\sigma}$ and $T^\div=\denotesgr{t^\div\sigma}$;
it follows that $s=s_0\theta\uplus s^+$, $t=t^0\uplus t^+\setminus t^\div$ and, by definition of the \metachr\ version of $\allrars$,
that $t^+\subseteq\allrars(s_0\theta\uplus s^+)\setminus\allrars(s_0)$ and $t^\div\subseteq  \allrars(s_0\theta)$.

We have now that $s=s_0\theta\uplus s^+$ and $t=t_0\uplus t^+\setminus t^\div$, and together with
the soundness parts of Definitions~\ref{def:abstract-eqs-etc} and~\ref{def:abstract-derivation-step},
and the relationship between $\exe$ and $\mexenoargs$ (Definition~\ref{def:metachr}, p.~\pageref{def:metachr}),
it follows that $\lambda$ is subsumed by $\Lambda$.
\end{proof}
We notice the following straight-forward property, indicating that we can use  existing, automatic confluence checkers
(e.g.,~\cite{Raiser-Langbein2010})
to classify further abstract corners as ``trivially joinable'', so only those abstract corners whose joinability critically depend on
$I$ and $\approx$ need to be considered.
\begin{proposition}~\label{prop:reuse-old-conf-checkers}
Consider a program with invariant $I$ and equivalence $\approx$, and with only logical and $I$-complete built-ins,
and let $\mathbf{\Lambda}$ be an abstract critical $\alpha_1$-$I$-corner lifted from a most general critical pre-corner
$\circ \mathrel{\stackrel{R_1}\mapsfrom} \Sigma \mathrel{\stackrel{R_2}\ourmapsto} \circ$.
If the concrete corner $\Sigma_1 \mathrel{\stackrel{R_1}\mapsfrom} \Sigma \mathrel{\stackrel{R_2}\ourmapsto} \Sigma_2$
exists and is joinable modulo $=$ with invariant $\true$,
$\mathbf{\Lambda}$ is joinable modulo $\approx$ with invariant $I$.
\end{proposition}
It does not hold that a program is confluent modulo $\approx$ if and only if all of its abstract critical pairs are joinable.
This is demonstrated by the following example.
\begin{example}[Continuing Example~\ref{ex:different-confluence}, p.~\pageref{ex:different-confluence}]\label{ex:different-confluence-abstract}
Consider the program of Example~\ref{ex:different-confluence}; its two first rules leads to the following
abstract critical corner $\mathbf{\Lambda}^{r_1,r_2}$ (there are no propagation rules, so we leave out the propagation history).
$$ \bigl(\{\texttt{q}(x)\} \where \texttt{1>=}x\land x\texttt{>=-1}\bigr) \;\stackrel{r_1}\mapsfrom\;
     \bigl( \{\texttt{p}(x)\} \where \texttt{1>=}x\land x\texttt{>=-1}\bigr)
    \; \stackrel{r_2}\ourmapsto  \; \bigl( \{\texttt{r}(x)\} \where\texttt{1>=}x, x\texttt{>=-1}\bigr)
$$
We recall the two remaining rules of this program, that may perhaps apply to a state
consisting of a single \texttt{q}  atom.

\medskip\noindent
\hbox to 2em{$r_3$:\hfil}\verb"q(X) <=> 1 >= x, x >= 0  | r(x)"\\
\hbox to 2em{$r_4$:\hfil}\verb"q(X) <=> 0 >= x, x >= -1 | r(x)"

\medskip\noindent
It appears that none of these rules can apply to the left wing state so $\mathbf{\Lambda}^{r_1,r_2}$ is not joinable,
although any concrete corner covered by it is joinable.
\end{example}
This phenomenon which is induced by the presence of non-logical and incomplete predicates motivates
the following.
   
\begin{definition}[Split-joinability]\label{def:splitting}
Assume a set of  \metachr{} formulas, $\{ \Phi_{i}\mid i\in \mathit{Inx}\}$,
for some finite or infinite index set $\mathit{Inx}$, such that
$$\Phi \Leftrightarrow    \bigvee_{i\in \mathit{Inx}}\Phi_{i}.$$
A \emph{splitting} of an abstract corner 
$$A'  \;\mathrel{\mathit{Rel}_1} \;(\Sigma \where \Phi_i) \; \mathrel{\mathit{Rel}_2} A''$$
is the set of  abstract corners
$$ \bigl\{\bigl(\post((\Sigma \where \Phi_i), \mathit{Rel}_1)\; \mathrel{\mathit{Rel}_1} \;(\Sigma \where \Phi_i) \; \mathrel{\mathit{Rel}_2} \; \post((\Sigma \where \Phi_i), \mathit{Rel}_2)\bigr)  \mid i\in \mathit{Inx} \bigr\}.$$
An abstract corner is \emph{split-joinable} modulo $\approx$ whenever it has a splitting $\{\mathbf{\Lambda}_i \mid i\in \mathit{Inx}\}$
such that each $\mathbf{\Lambda}_i$ is either inconsistent or joinable modulo $\approx$.
\end{definition}
The following property follows immediately from the definition.
\begin{proposition}\label{prop:splitting-with-recursion}
For any splitting of an abstract corner $\mathbf{\Lambda}$ into $\{ \mathbf{\Lambda}_i\mid i\in \mathit{Inx}\}$,
it holds that
$$
\denotes{\mathbf{\Lambda}} = \bigcup_{i\in \mathit{Inx}} \denotes{\mathbf{\Lambda}_i}.
$$
\end{proposition}

\begin{example}[continuing Example~\ref{ex:different-confluence-abstract}]
The non-joinable abstract critical corner $\mathbf{\Lambda}^{r_1,r_2}$ is split-joinable using the  disjunction
$(\texttt{1>=}x\land x\texttt{>=0})\;\lor\;
        (\texttt{0>=}x\land x\texttt{>=-1})$.
Notice that neither $\mathbf{\Lambda}^{r_1,r_2}$ nor any member of  its splitting covers
a concrete corner of the form $(\cdots\mapsfrom \{\texttt{p(X)}\} \ourmapsto\cdots)$,
where \texttt{X} is a CHR variable.
\end{example}
We notice that the invariant of groundedness does not in itself make splitting necessary,
see, e.g., the example in Section~\ref{sec:set-example-all-details}, below.
In Section~\ref{sec:infinite-split} below, we show an example of a program that needs an infinite splitting,
but still we can use the results in the present section to show confluence.

\begin{theorem}[Abstract Critical Corner Theorem]\label{thm:abstract-CC-theorem} 
A CHR program with invariant $I$ and equivalence $\approx$
is locally confluent modulo $\approx$
if and only if 
each of its abstract critical $I$-corners is either inconsistent, joinable modulo $\approx$, or split-joinable modulo $\approx$.
\end{theorem}
\begin{proof}
Follows immediately from Critical Corner Theorem, i.e., Theorem~\ref{thm:critCorner}, p.~\pageref{thm:critCorner},
Lemma~\ref{lem:cover-subsume} and Proposition~\ref{prop:splitting-with-recursion}.
\end{proof}
Combining this result with Theorem~\ref{thm:conflu-and-critCorner}, p.~\pageref{thm:conflu-and-critCorner},
we arrive at our following central result.

\begin{theorem}\label{thm:termination+abstract-joinability=confluence}
 A terminating program with invariant $I$ and equivalence relation $\approx$ is confluent 
 if and only if 
 each of its abstract critical $I$-corners is either inconsistent, joinable modulo $\approx$, or split-joinable modulo $\approx$.
\end{theorem}
%
%

\section{Examples}\label{sec:examples}
We show three examples of confluence proofs.
First, we give all details for the very simple but highly motivating example appearing in the Introduction
of this paper.
Next, we consider a more complex program, the Viterbi algorithm expressed in CHR, for which we formalize
invariant and equivalence and give the proof of confluence modulo equivalence.
This is a practically interesting algorithm, and the example also demonstrates that
our framework can deal with nontrivial reasoning about the propagation history.
Finally, we show an example that our method is robust for some cases where a countably infinite splitting
is needed, Section~\ref{sec:infinite-split}.

Confluence modulo equivalence of a CHR version
of the union-find algorithm~\cite{Tarjan:1984:WAS:62.2160},
which has been used as a test case
for  aspects of confluence,
is demonstrated informally by~\cite{DBLP:conf/lopstr/ChristiansenK14}.
A detailed analysis and proof in terms of abstract critical corners is planned to appear in a 
future publication.

\subsection{The Motivating One-line Program shown Confluent Modulo Equivalence}\label{sec:set-example-all-details}
In the Introduction, we motivated  confluence modulo equivalence for CHR 
by a program consisting of the following single rule.
\begin{verbatim}
 set(L), item(A) <=> set([A|L]).
\end{verbatim}
Here we formalize the invariant $I$ and equivalence $\approx$ hinted in Examples~\ref{ex:collect} and ~\ref{ex:collect-inv}, p.~\pageref{ex:collect}-\pageref{ex:collect-inv}, and give a proof of confluence modulo $\approx$.
\begin{itemize}
\item $I(\langle S, T \rangle)$ if and only if
\begin{itemize}
 \item $S = \{{\tt set}(L)\} \uplus Items$ where $L$ is a list of constants, $Items$ is a set of {\tt item}/1 constraints
 whose arguments are constants,
 \item $T = \emptyset$.
\end{itemize}
\item $\langle S,T\rangle \approx \langle S',T'\rangle$,  if and only if 
 \begin{itemize}
  \item $I(\langle S,T\rangle)$ and $I(\langle S',T'\rangle)$,
  \item $S = \mbox{\tt set$(L)$} \uplus \mathit{Items}$ and  $S' = \mbox{\tt set$(L')$} \uplus \mathit{Items}$ such that $L$ and $L'$ are permutations of each other.
 \end{itemize}
 \end{itemize}
We identify the following two most general critical pre-corners for the program.
To give a complete picture, we have not abbreviated the application instances that label the derivation steps,
as we do in most other examples.
%
$$ 
\begin{array}{@{}c@{}}
\Lambda_1 = \;\;
 \bigl(\circ
 \stackrel{\scriptsize \mbox{\tt set(L),item(A)$\,$<=>$\,$set([A|L])}}{\mapsfrom}
 \langle \{\mbox{\tt set(L)}, \mbox{\tt item(A)}, \mbox{\tt item(B)}\}, \emptyset \rangle
 \stackrel{\scriptsize \mbox{\tt set(L),item(B)$\,$<=>$\,$set([B|L])}}{\ourmapsto} 
 \circ \bigr)
 \\
 \Lambda_2 = \;\;
 \bigl(\circ
 \stackrel{\scriptsize \mbox{\tt set(L1),item(A)$\,$<=>$\,$set([A|L1])}}{\mapsfrom}\; 
 \langle \{\mbox{\tt set(L1)},\mbox{\tt item(A)}, \mbox{\tt set(L2)}\}, \emptyset \rangle \; 
 \stackrel{\scriptsize \mbox{\tt set(L2),item(A)$\,$<=>$\,$set([A|L2])}}{\ourmapsto} 
 \circ \bigr)
 \end{array}
$$
We lift $\Lambda_1, \Lambda_2$ to the following abstract critical $I$-corners according to Definition~\ref{def:lifting-pre-corners}, p.~\pageref{def:lifting-pre-corners}. Trivially satisfied meta-level constraints are removed.
 $$ \scriptsize
 \begin{array}{@{}l@{}}
 \mathbf{\Lambda_1} =\;\;\;\;
        \begin{array}{c}
     \xymatrix{ 
     {\begin{array}{l}
     \langle \{\mbox{\tt set}(\ell), \mbox{\tt item}(a), \mbox{\tt item}(b)\} \uplus S, \emptyset \rangle \where \\
     \phantom{\where }\; I(\langle \{\mbox{\tt set}(\ell), \mbox{\tt item}(a), \mbox{\tt item}(b)\} \uplus S, \emptyset \rangle) 
     \end{array}}
 \ar@{|->}[dr]^(.6){\qquad\qquad\qquad \mbox{\scriptsize\tt set($\ell$),item($b$)$\,$<=>$\,$set($[b|\ell]$)}}
 \ar@{|->}[d]_-{\mbox{\scriptsize\tt set($\ell$),item($a$)$\,$<=>$\,$set($[a|\ell]$)}} & \\
 {\begin{array}{l}
     \langle \{\mbox{\tt set}([a|\ell]), \mbox{\tt item}(b)\} \uplus S, \emptyset \rangle \where \\
     \phantom{\where }\; I(\langle \{\mbox{\tt set}([a|\ell]), \mbox{\tt item}(b)\} \uplus S, \emptyset \rangle) 
     \end{array}}
 &
  {\begin{array}{l}
     \langle \{\mbox{\tt set}([b|\ell]), \mbox{\tt item}(a)\} \uplus S, \emptyset \rangle \where \\
     \phantom{\where }\; I(\langle \{\mbox{\tt set}([b|\ell]), \mbox{\tt item}(a)\} \uplus S, \emptyset \rangle) 
     \end{array}} 
 \\ }
 \end{array} \\ \\ \\
 \mathbf{\Lambda_2} =\;\;\;\;
 \begin{array}{c}
     \xymatrix{ 
     {\begin{array}{l}
     \langle \{\mbox{\tt set}(\ell_1), \mbox{\tt set}(\ell_2), \mbox{\tt item}(a)\} \uplus S, \emptyset \rangle \where \\
     \phantom{\where }\; I(\langle \{\mbox{\tt set}(\ell), \mbox{\tt set}(\ell_2), \mbox{\tt item}(a)\} \uplus S, \emptyset \rangle) 
     \end{array}}
 \ar@{|->}[dr]^(.6){\qquad\qquad\qquad \mbox{\scriptsize\tt set($\ell_2$),item($a$)$\,$<=>$\,$set($[a|\ell_2]$)}}
 \ar@{|->}[d]_-{\mbox{\scriptsize\tt set($\ell$),item($a$)$\,$<=>$\,$set($[a|\ell]$)}} & \\
 {\begin{array}{l}
     \langle \{\mbox{\tt set}(\ell_2),\mbox{\tt set}([a|\ell])\} \uplus S, \emptyset \rangle \where \\
     \phantom{\where }\; I(\langle \{\mbox{\tt set}(\ell_2),\mbox{\tt set}([a|\ell])\} \uplus S, \emptyset \rangle) 
     \end{array}}
  &
  {\begin{array}{l}
     \langle \{\mbox{\tt set}(\ell),\mbox{\tt set}([a|\ell_2])\} \uplus S, \emptyset \rangle \where \\
     \phantom{\where }\; I(\langle \{\mbox{\tt set}(\ell),\mbox{\tt set}([a|\ell_2])\} \uplus S, \emptyset \rangle) 
     \end{array}}
 }
     \end{array}
       \end{array}
$$
The abstract corner $\mathbf{\Lambda_2}$ is inconsistent because $I$ does not accept a constraint store with more than one {\tt set} constraint, and $\mathbf{\Lambda_1}$ is shown joinable modulo $\approx$ by the following abstract diagram $\boldsymbol{\Delta}_1$.
   $$ \scriptsize
    {\boldsymbol{\Delta}}_1 = \;\;\;\;
        \begin{array}{c}
     \xymatrix{ 
     {\begin{array}{l}
     \langle \{\mbox{\tt set}(l), \mbox{\tt item}(a), \mbox{\tt item}(b)\} \uplus S, \emptyset \rangle \where \\
     \phantom{\where }\; I(\langle \{\mbox{\tt set}(l), \mbox{\tt item}(a), \mbox{\tt item}(b)\} \uplus S, \emptyset \rangle) 
     \end{array}}
 \ar@{|->}[dr]^(.6){\qquad\qquad\qquad \mbox{\scriptsize\tt set($l$),item($b$)$\,$<=>$\,$set($[b|l]$)}}
 \ar@{|->}[d]_-{\mbox{\scriptsize\tt set($l$),item($a$)$\,$<=>$\,$set($[a|l]$)}} & \\
 {\begin{array}{l}
     \langle \{\mbox{\tt set}([a|\ell]), \mbox{\tt item}(b)\} \uplus S, \emptyset \rangle \where \\
     \phantom{\where }\; I(\langle \{\mbox{\tt set}([a|\ell]), \mbox{\tt item}(b)\} \uplus S, \emptyset \rangle) 
     \end{array}} 
     \ar@{|->}[d]_-{\mbox{\scriptsize\tt set($[a|l]$),item($b$)$\,$<=>$\,$set($[b,a|l]$)}}
 &
  {\begin{array}{l}
     \langle \{\mbox{\tt set}([b|\ell]), \mbox{\tt item}(a)\} \uplus S, \emptyset \rangle \where \\
     \phantom{\where }\; I(\langle \{\mbox{\tt set}([b|\ell]), \mbox{\tt item}(a)\} \uplus S, \emptyset \rangle) 
     \end{array}}
     \ar@{|->}[d]^-{\mbox{\scriptsize\tt set($[b|\ell]$),item($a$)$\,$<=>$\,$set($[a,b|\ell]$)}}
 \\ 
 {\begin{array}{l}
     \langle \{\mbox{\tt set}([b,a|\ell])\} \uplus S, \emptyset \rangle \where \\
     \phantom{\where }\; I(\langle \{\mbox{\tt set}([b,a|\ell])\} \uplus S, \emptyset \rangle) 
     \end{array}}
    \ar@2{~}[r]
 &
  {\begin{array}{l}
     \langle \{\mbox{\tt set}([a,b|\ell]) \} \uplus S, \emptyset \rangle \where \\
     \phantom{\where }\; I(\langle \{\mbox{\tt set}([a,b|\ell]) \} \uplus S, \emptyset \rangle) 
     \end{array}}
 \\  
 }
     \end{array}
$$
The program is terminating since each derivation step reduces the number of \texttt{item} constraints
by one, so by Theorem~\ref{thm:termination+abstract-joinability=confluence} it follows the program is  confluent modulo $\approx$.

\subsection{Confluence  Modulo Equivalence  of the Viterbi Algorithm}\label{sec:viterbi}
The Viterbi algorithm~\cite{Vit67} is an example of a dynamic programming algorithm that
searches for
one optimal solution to a problem among, perhaps, several equally good ones.

A Hidden Markov Model, HMM, is a finite state machine with probabilistic state transitions and
probabilistic emission of a letter from each state.
The \emph{decoding problem} for an observed sequence of emitted letters
$Ls$ is that of finding
a most probable \emph{path} which is a sequence of state transitions that may have produced $Ls$;
see~\cite{DurbinEtAL99} for a background on HMMs and their applications in computational biology.

A decoding problem is typically solved using the Viterbi algorithm~\cite{Vit67} which is an example
of a dynamic programming algorithm that produces solutions for a problem by successively extending
solutions for growing subproblems. While there are potentially exponentially many differents paths
to compare, an early pruning strategy ensures linear time complexity.

The algorithm has been studied in CHR
by~\cite{ChristiansenEtAlCHR2010,DBLP:conf/lopstr/ChristiansenK14} as shown below.
The optimal complexity requires a restriction in the possible derivations, namely that the \texttt{prune}
rule (below) is applied as early as possible. In~\cite{ChristiansenEtAlCHR2010}, it is
demonstrated how such a rule ordering can be imposed by semantics-preserving program
transformations; here we will show
confluence of the program modulo a suitable equivalence (which ensures that limiting the rule order
does not destroy the semantics of the program).

\begin{example}\label{ex:HMM-as-such}
The following diagram shows a very simple HMM with states \texttt{q0}, $\ldots$, \texttt{q3},
emission alphabet $\{\texttt{a},  \texttt{b}\}$ and probabilities indicated for transitions
and emissions.\footnote{An interesting HMM will, of course, have loops so it can produce arbitrary long
sequences. No explicit end states are needed.}

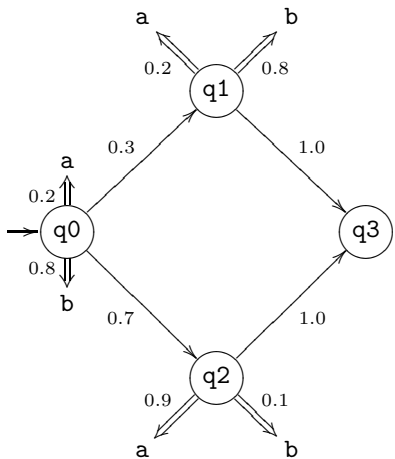
$$
\begin{array}{c} 
{\xymatrix@R=4mm@C=4mm{
&		& {\tt a} 	&  		& {\tt b} 	& \\
&		&	& *++[o][F]{\tt q1}
			   \ar@{=>}[ul]^{0.2}
			   \ar@{=>}[ur]_{0.8}
			   \ar[ddrr]^{1.0}
					&	& \\
& {\tt a}    		& 	&		&	& \\    
\ar@{->}[r]& *++[o][F]{\tt q0}	
\ar@{=>}[u]^{0.2\,}
\ar@{=>}[d]_{0.8\,}
\ar[ddrr]_{0.7}
\ar[uurr]^{0.3}
		& 	&		&	& *++[o][F]{\tt q3}\\
&{\tt b}     		& 	&		&	& \\
&		&	& *++[o][F]{\tt q2}
			   \ar@{=>}[dl]_{0.9}
			   \ar@{=>}[dr]^{0.1}	
			   \ar[uurr]_{1.0}
					&	& \\		
&		& {\tt a } 	&  		& {\tt b} 	& \\		
}}
\end{array}
$$
The different events of transitions and emissions are assumed to be independent.
For example, the sequence \texttt{a$\cdot$b} may be produced
via the path \texttt{q0}$\cdot$\texttt{q1}$\cdot$\texttt{q3} with probability
$0.2*0.3*0.8=0.048$ or 
\texttt{q0}$\cdot$\texttt{q2}$\cdot$\texttt{q3} with probability $0.2*0.7*0.1=0.014$.
For simplicity of the program that follows, it is assumed that an emission is produced when
a state is left (rather than entered).
\end{example}
A specific HMM is encoded as a set of \texttt{trans/3} and \texttt{emit/3} constraints that are not changed 
during program execution.
\begin{example}\label{ex:HMM-encoding}
The HMM of Example~\ref{ex:HMM-as-such} is encoded by the following constraints.
$$
\begin{tabular}{l}
  $\{\texttt{ trans(q0,q1,0.3)}, \texttt{ trans(q0,q2,0.7)},
  \texttt{ trans(q1,q3,1)}, \texttt{ trans(q2,q3,1)},$  \\
 \hbox to 0.6em{}$\texttt{ emit(q0,a,0.2)},  \texttt{ emit(q0,b,0.8)},$ \\ 
 \hbox to 0.6em{}$ \texttt{ emit(q1,a,0.2)},  \texttt{ emit(q1,b,0.8)},$\\  
 \hbox to 0.6em{}$ \texttt{ emit(q2,a,0.9)},  \texttt{ emit(q2,b,0.1)}  \}$
\end{tabular}
$$
\end{example}
The CHR program that implements the Viterbi algorithm is as follows.
{\begin{verbatim}
:- chr_constraint path/4, trans/3, emit/3.

expand @ trans(Q,Q1,PT), emit(Q,L,PE), path([L|Ls],Q,P,PathRev) ==>
   P1 is P*PT*PE  |  path(Ls,Q1,P1,[Q1|PathRev]).

prune @ path(Ls,Q,P1,_) \ path(Ls,Q,P2,_) <=> P1 >= P2 | true.
\end{verbatim}}\noindent
The meaning of a constraint \texttt{path($Ls$,$q$,$p$,$R$)}
is that $Ls$ is a remaining emission sequence to be processed, $q$  the current state of
the HMM, and $p$ the probability of a path $R$ found for the already processed prefix
of the emission sequence. To simplify the program, a path is represented in reverse order.
The decoding of a sequence $Ls$ starting from state $q_0$ is stated by the query
\begin{itemize}
  \item[] {\tt:- }$\textit{HMM}$\!\texttt{,} \texttt{path($Ls$,$q_0$,1,[$q_0$])}.
\end{itemize}
where \textit{HMM} is an encoding of a given
HMM in terms of ground \texttt{trans} and \texttt{emit} constraints;
for each pair of states $q_1,q_2$, 
\textit{HMM} contains at most one constraint of
the form \texttt{trans($q_1$,$q_2$,$\ldots$)}, and for each pair of state $q$ and emission letter $L$,
 \textit{HMM} contains at most one constraint of
the form \texttt{emit($q$,$L$,$\ldots$)}.

The first rule of the program, \texttt{expand},
expands the existing paths and \texttt{prune} removes paths for identical subproblems
(identified by the current HMM state and remaining emission sequence) with lower (or equal) probabilities.
The program is terminating for such queries as any new \texttt{path} constraint introduced by the \texttt{expand} rule
has a first argument shorter than that of its predecessor.
A final state will include one \texttt{path} constraint of optimal probability for each
prefix of the input string (unless the underlying state machine is not capable of generating that string).

The program is not confluent in the classical sense, as the \texttt{prune} rule may nondeterministically
remove one or the other of two alternative \texttt{path} constraints of identical probability for the same
sequence.
In the following we introduce invariant $I$ and equivalence $\approx$
and show the program confluent modulo equivalence.
For simplicity of the definitions and with no loss of generality, we assume a fixed
indexed encoding $\mathit{HMM}$ of a Hidden Markov Model
with initial state \texttt{q0} and fixed input emission  sequence $Ls_0$.

\begin{definition}
$I(\Sigma)$ if and only if
$\langle \textit{HMM} \cup \{(0\colon \texttt{path($Ls_0$,$q_0$,1,[$q_0$])})\}\rangle
\mapstostar
\Sigma$.
\end{definition}
However, in the proof of local confluence below, we will need a more direct
characterization of the possible derivation states
and the interrelations between their constraints.
To this end, we state the following proposition.

\begin{proposition}\label{prop:viterbi-invariant}
An $I$-state is of the form $\langle S\cup\mathit{HMM}, T\rangle$
where 
$S$ is a set of ground \texttt{path} constraints and $T$ a propagation history.

For any $(i\colon \texttt{path($[L|Ls]$,$q$,$P$,$qs$)}) \in S$ for which 
$\{ (i^t\colon\texttt{trans($q$,$q'$,$P^t$)}), (i^e\colon\texttt{emit($q$,$L$,$P^e$)})\} \in \mathit{HMM}$,
then one and only one of the following will be the case.
\begin{enumerate}
  \item \emph{Expansion has not taken place:}\\
          $({\texttt{expand}}@ i^t, i^e, i)\not\in T$
  \item  \emph{Expansion produced and still in the store:}\\
          $({\texttt{expand}}@ i^t, i^e, i)\in T \;\land\;
             \exists i' .\, (i'\colon \texttt{path($Ls$,$q'$,$P'$,$[q'|qs]$)}) \in S$
           where $P'$ is the value of $P$$*$$P^t$$*$$P^e$.
              
  \item  \emph{Expansion produced but pruned by stronger or equal alternative:}\\
          $({\texttt{expand}}@ i^t, i^e, i)\in T \;\land\;
             \not\exists i',P' .\, (i'\colon \texttt{path($Ls$,$q'$,$P'$,$[q'|qs]$)}) \in S$\\
           $\land\;
             \exists P',qs',i' .\, \bigl((i'\colon \texttt{path($Ls$,$q'$,$P'$,$[q'|qs']$)}) \in S
                       \land P' \geq P$$*$$P^t$$*$$P^e \land qs\neq qs'\bigr)$
\end{enumerate}
 \end{proposition}
Notice in case 3, that the \texttt{path} required to exists may either be the stronger (or equal) alternative
that via  \texttt{prune} rule lead to the removal of \texttt{path($Ls$,$q'$,$P'$,$[q'|qs]$)},
or it may be an even stronger (or equal) one, meaning that
several applications of \texttt{prune} have been involved.
%
The uniqueness of \texttt{emit} (\texttt{trans}) constraints in \textit{HMM} for a
fixed $q$ ensures that the constraints $(i'\colon \texttt{path($Ls$,$q'$,$P'$,$[q'|qs]$)})$ in case 2 and 3 are unique and uniquely related to the application record ${\texttt{expand}}@ i_t, i_e, i$.

\begin{proof}
We use induction over the length of the derivation leading to a given $I$-state.
\par\smallskip\noindent
\emph{Base case.} The state
 $\langle \textit{HMM} \cup \{(0\colon \texttt{path($Ls$,$q_0$,1,[$q_0$])})\}\rangle$
matches case 1 in the proposition.
\par\smallskip\noindent
\emph{Step.} Assume an $I$-state $\Sigma=\langle S\cup\mathit{HMM}, T\rangle)$ satisfying the proposition,
and let $\Sigma\ourmapsto\Sigma^*$, where\break
 $\Sigma^*=\langle S^*\cup\mathit{HMM}, T\rangle$.
Two kinds of derivation steps are possible, one for each program rule.

\texttt{expand}: Assume that the \texttt{path} constraint $i{:}\pi\in S$ of $\Sigma$ is
involved in an application of the \texttt{expand} rule. The only difference between
$\Sigma$ and $\Sigma^*$ is that the latter includes
a new \texttt{path} constraint $i^*{:}\pi^*\in S^*$ and a new application record
$({\texttt{expand}}@ i^t, i^e, i^*)\in T^*$.
In $\Sigma^*$, $i{:}\pi$ satisfies condition 2, and $i^*{:}\pi^*$ condition 1;
any other \texttt{expand} constraint in $\Sigma^*$ satisfies the same of 1, 2, 3
as in $\Sigma$.

\texttt{prune}: Assume that the rule applies in $\Sigma$ by the application instance
$i_1{:}\pi_1\texttt{/}i_2{:}\pi_2\mathrel\texttt{<=>}P_1\texttt{>=} P_2\texttt{|}\texttt{true}$.
Thus $S^*=S\setminus\{i_2{:}\pi_2\}$, $T^*=T$.
It holds that, when $i_1{:}\pi_1$ satisfies condition $k$ in $\Sigma$, then it also satisfies
condition $k$ in $\Sigma^*$ for $k=1,2,3$.

The only way that the removal of $i_2{:}\pi_2$ may affect the proposition
is when there is another $(i_2^0{:}\pi_2^0)\in S$ satisfying condition  2 or 3
with  $i_2$ in the role of the existentially quantified index $i$ in either case.
For condition 2, $(i_2^0{:}\pi_2^0)$  satisfies condition 2 or 3 in $\Sigma^*$ with $i'=i_1$;
for condition 3, $(i_2^0{:}\pi_2^0)$ satisfies condition 3 in $\Sigma^*$ with $i'=i_1$.
\end{proof}
Our equivalence relation specifies 
the intuition that
two solutions for the same subproblem are 
equally good 
when they have the same probability.
We recall that a state is defined as an equivalence class
over state representations sharing the same pattern of variable recurrence. 
\begin{definition}
The  $\approx$ is the smallest equivalence relation on $I$-states such that
$\langle S\cup \mathit{HMM}, T \rangle \approx $\break$\langle S'\cup \mathit{HMM}, T\rangle$
if and only if
\begin{itemize}
  \item For any $i\colon\! \mbox{\tt path($Ls$,$q$,$P$,$qs$)}\; \in S$, there is
    an $i\colon\! \mbox{\tt path($Ls$,$q$,$P$,$qs'$)}\; \in S$, and vice versa.
\end{itemize}
\end{definition}
\begin{theorem}
The Viterbi program with invariant $I$ is confluent modulo $\approx$. 
\end{theorem}
%
\par\medskip\noindent\textit{Proof.}\hskip 1ex
According to Theorem~\ref{thm:termination+abstract-joinability=confluence},
we can prove confluence of a CHR program by listing the 
 set of
critical abstract corners and showing each of them joinable or split joinable.

Firstly, we observe that no built-in predicate can appear in an $I$-state (they are only used in guards)
and that the two built predicates \texttt{>=} and \texttt{is} are $I$-complete.
Thus, we have no $\alpha_2$- and $\alpha_3$-corners to consider, leaving only $\alpha_1$-  and $\beta$-corners.
For a better overview, we indicate the overall shapes of corners in the chosen canonical set,
described in full detail below.
There are three $\alpha_1$-corners, one for each possible way that two rules may 
 produce a critical overlap:
\begin{itemize}
  \item [] $\Lambda_1\colon\hbox to 1.9em{} \circ \stackrel{\small\tt prune}\mapsfrom \circ \stackrel{\small\tt prune}\ourmapsto\circ$
  \item[] $\Lambda_2, \Lambda_3 \colon~ \circ \stackrel{\small\tt prune}\mapsfrom \circ \stackrel{\small\tt expand}\ourmapsto\circ$ ~differing in whether or not the constraint being expanded is removed;\\
\phantom{$\Lambda_2, \Lambda_3 \colon\hbox to 0pt{}$}our analysis will show $\Lambda_3$ (expanded constraint removed) is not joinable,
but can be split into\\\phantom{$\Lambda_2, \Lambda_3 \colon\hbox to 0pt{}$}three joinable subcases
$\Lambda_3^{(1)}, \Lambda_3^{(2)}, \Lambda_3^{(3)}$, one for each option in Proposition~\ref{prop:viterbi-invariant}.
\end{itemize}
Two $\beta$-corners are found, one for each clause of the program.
\begin{itemize}
  \item []$\Lambda_4\colon\hbox to 1.9em{} \circ \approx \circ \stackrel{\small\tt prune}\ourmapsto\circ$
  \item []$\Lambda_5\colon\hbox to 1.9em{} \circ \approx \circ \stackrel{\small\tt expand}\ourmapsto\circ$
\end{itemize}
To save space, application steps are labelled  by application records
(rather that application instances) and we leave out also the $id$ function, e.g., writing
\texttt{prune@}$\pi_1\pi_2$ instead of \texttt{prune@}$\id((\pi_1,\pi_2))$.

We abbreviate the writing of the invariant in an abstract state,
writing $(\Sigma\where I\land\cdots)$ instead of $(\Sigma\where I(\Sigma)\land\cdots)$,
where $\Sigma$ is a (perhaps complex) abstract state expression.
We use the following conventions
in expressions that represent propagation histories.
\begin{itemize}
  \item A condition of the form $ra\not\in T$, where $ra$ is a rule application and $T$ a propagation history,
  may be removed in an abstract state expressions when it is clear from context that it is always satisfied.
  This is relevant when $ra=(\langle\textit{rule-id}\rangle@\cdots i\cdots)$ and $T$ is part of a state guaranteed not
  to contain $i$.
   \item When $i$ represents a constraint index and $T$  a propagation history,
   the notation $T\setminus i$ is a shorthand for
   $T\setminus\{ra\mid \mbox{$ra$ is an application record of the form $ra=(\langle\textit{rule-id}\rangle@\cdots i\cdots)$}\}$.
   \end{itemize}
To simplify notation for the description of these corners,
we introduce the following abbreviations; the recurrences of variables are significant.
\begin{eqnarray*}
\tau & = & (i^t\colon\mbox{\tt trans($q$,$q'$,$P^t$)}) \\
\eta & = & (i^e\colon \mbox{\tt emit($q$,$L$,$P^e$)}) \\
\pi_j & = & (i_j\colon \mbox{\tt path([$L$|$LS$],$q$,$P_j$,$qs_j$)})\quad\,\,\mbox{for  $j=1,\ldots,4$} \\
 \pi'_j  & = & (i'_j\colon \mbox{\tt path($LS$,$q'$,$P'_j$,[$q'$|$qs_j$])})
     \quad\mbox{$P'_j$ is the value of $P_j$$*$$P^t$$*$$P^e$ for    $j=1,\ldots,4$}
\end{eqnarray*}
As it appears, $\pi_i'$ is can be derived from  $\pi_i$,
$\tau$ and $\eta$  using the \texttt{expand} rule.
The \texttt{path} constraints $\pi_1,\ldots\pi_4$ all concern the same sub-problem,
identified by the identical first and second argument, {\tt [$L$|$LS$]} and $q$;
and analogously for the $\pi'_i$ constraints.

We consider now the canonical abstract corners one by one and show them (split) joinable.

\smallskip
\noindent$\mathbf{\Lambda}_1$: \emph{Overlap of}  \texttt{prune} \emph{with itself}
   \begin{displaymath} \scriptsize
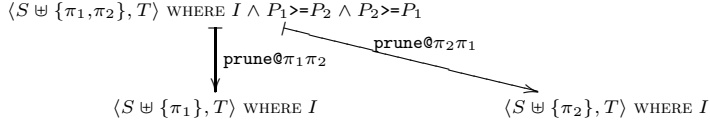
        \begin{array}{c}
     \xymatrix{ \langle S \uplus \{\pi_1,\!  \pi_2\}, T \rangle \where I\land P_1\texttt{>=} P_2 
              \land P_2\texttt{>=} P_1 
 \ar@{|->}[dr]^-{\mbox{\scriptsize\tt prune@$\pi_2\pi_1$}}     
 \ar@{|->}[d]^{\mbox{\scriptsize\tt prune@$\pi_1\pi_2$}}    & \\
\langle S \uplus \{\pi_1\}, T \rangle \where I &
 \langle S \uplus \{ \pi_2\}, T \rangle \where I }
                \\ \\ 
     \end{array}
\end{displaymath}
This extends immediately to a joinability diagram because the two abstract wing states are equivalent.

\smallskip
\noindent$\mathbf{\Lambda}_2$: \emph{Overlap of} \texttt{prune} \emph{and} \texttt{expand}\emph{; expanded constraint not removed}
   \begin{displaymath} \scriptsize
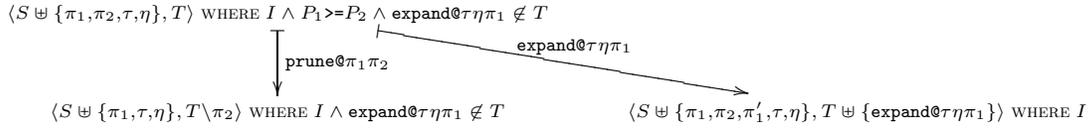
        \begin{array}{c}
     \xymatrix{ \langle S \uplus \{\pi_1, \! \pi_2, \! \tau,\!  \eta\}, T \rangle \where I\land  P_1\texttt{>=} P_2
           \land \texttt{expand@}\tau\eta\pi_1 \not\in T
 \ar@{|->}[dr]^-{\mbox{\scriptsize\tt expand@$\tau\eta\pi_1$}}     
 \ar@{|->}[d]^{\mbox{\scriptsize\tt prune@$\pi_1\pi_2$}}    & \\
\langle S \uplus \{\pi_1,\!  \tau,\!  \eta\}, T\backslash\pi_2 \rangle \where I \land \texttt{expand@}\tau\eta\pi_1 \not\in T &
 \langle S \uplus \{\pi_1,\!  \pi_2,\!  \pi'_1,\! \tau,\!  \eta\}, T\uplus \{\texttt{expand@}\tau\eta\pi_1\} \rangle \where I  }
                \\ \\ 
     \end{array}
   \end{displaymath}
In this case, the two rules commute and the corner joins in one and the same abstract state.
     \begin{displaymath} \scriptsize
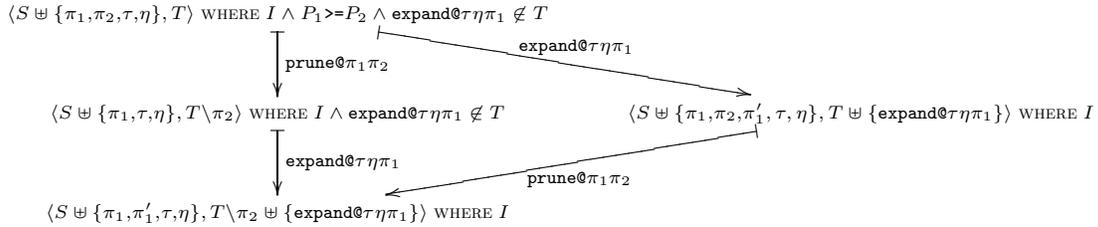
        \begin{array}{c}
     \xymatrix{  \langle S \uplus \{\pi_1,\!  \pi_2,\!  \tau, \! \eta\}, T \rangle \where I\land  P_1\texttt{>=} P_2
           \land \texttt{expand@}\tau\eta\pi_1 \not\in T
 \ar@{|->}[dr]^-{\mbox{\scriptsize\tt expand@$\tau\eta\pi_1$}}     
 \ar@{|->}[d]^{\mbox{\scriptsize\tt prune@$\pi_1\pi_2$}}    & \\
\langle S \uplus \{\pi_1,\!  \tau,\!  \eta\}, T\backslash\pi_2  \rangle \where I \land \texttt{expand@}\tau\eta\pi_1 \not\in T 
\ar@{|->}[d]^-{\mbox{\scriptsize\tt expand@$\tau\eta\pi_1$}} 
&
\langle S \uplus \{\pi_1,\!  \pi_2,\!  \pi'_1,\tau, \eta\}, T\uplus \{\texttt{expand@}\tau\eta\pi_1\} \rangle \where I  
 \ar@{|->}[dl]^{\mbox{\scriptsize\tt prune@$\pi_1\pi_2$}}
 \\
 \langle S \uplus \{\pi_1,\!  \pi'_1,\!  \tau,\!  \eta\}, T\backslash\pi_2 \uplus \{\texttt{expand@}\tau\eta\pi_1\} \rangle \where I 
 &
 }
     \end{array}
   \end{displaymath}

\noindent$\mathbf{\Lambda}_3$: \emph{Overlap of} \texttt{prune} \emph{and} \texttt{expand}\emph{; expanded constraint removed}
    \begin{displaymath} \scriptsize
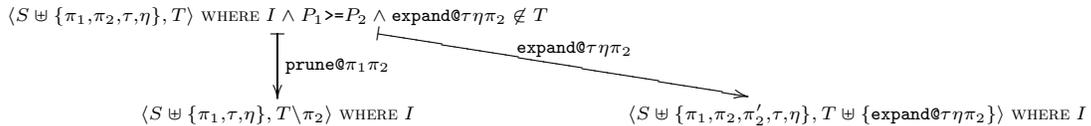
        \begin{array}{c}
     \xymatrix{ \langle S \uplus \{\pi_1,\! \pi_2,\!  \tau,\!  \eta\}, T \rangle \where I\land  P_1\texttt{>=} P_2
           \land \texttt{expand@}\tau\eta\pi_2 \not\in T
 \ar@{|->}[dr]^-{\mbox{\scriptsize\tt expand@$\tau\eta\pi_2$}}     
 \ar@{|->}[d]^{\mbox{\scriptsize\tt prune@$\pi_1\pi_2$}}    & \\
\langle S \uplus \{\pi_1,\!  \tau,\!  \eta\}, T\backslash\pi_2  \rangle \where I  &
\langle S \uplus \{\pi_1,\!  \pi_2,\!  \pi'_2,\! \tau,\!  \eta\}, T\uplus \{\texttt{expand@}\tau\eta\pi_2\} \rangle \where I  }
                \\ \\ 
     \end{array}
   \end{displaymath}
This abstract corner is not joinable as different derivations are possible depending on which of the three cases in Proposition~\ref{prop:viterbi-invariant} that holds for the path constraint $\pi_1$. This suggests a splitting of the corner into three new corners, that we can show joinable as follows. Hence, the corner is not joinable but split joinable. For reasons of space, we show only the related abstract joinability diagrams; the corners can be identified at the top.

\medskip
\noindent{$\mathbf{\Lambda}_3^{(1)}$: \emph{Split of} $\mathbf{\Lambda}_3$\emph{;} $\pi_1$ \emph{applicable}
   \begin{displaymath} \scriptsize
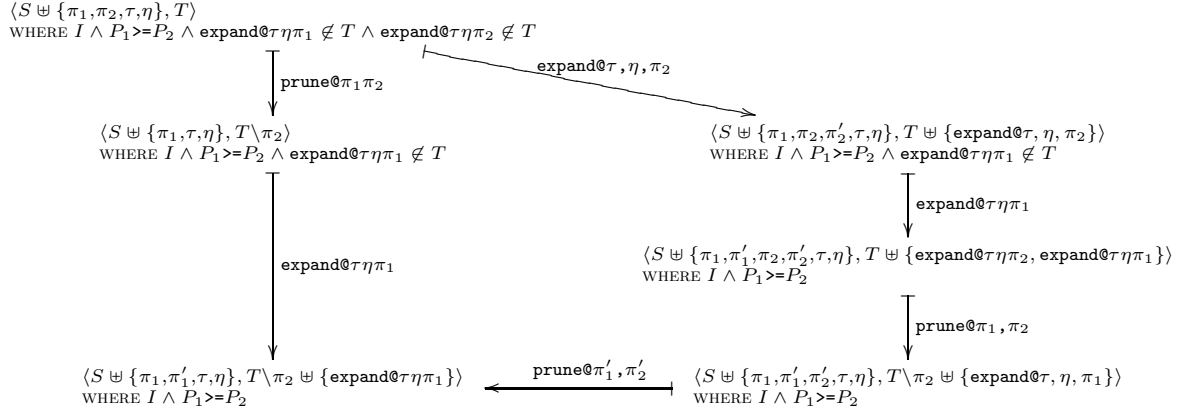
        \begin{array}{c}
     \xymatrix{
          {\begin{array}{l}
               \langle S \uplus \{\pi_1,\! \pi_2,\! \tau,\! \eta\}, T \rangle \\
                \where I\land  P_1\texttt{>=} P_2
               \land \texttt{expand@}\tau\eta\pi_1 \not\in T
               \land \texttt{expand@}\tau\eta\pi_2 \not\in T
            \end{array}}
 \ar@{|->}[dr]^-{\mbox{\scriptsize\tt expand@$\tau$,$\eta$,$\pi_2$}}     
 \ar@{|->}[d]^{\mbox{\scriptsize\tt prune@$\pi_1\pi_2$}}    & \\
{\begin{array}{l}\langle S \uplus \{\pi_1,\! \tau,\! \eta\}, T\backslash\pi_2  \rangle\\
 \where I\land  P_1\texttt{>=} P_2
               \land \texttt{expand@}\tau\eta\pi_1 \not\in T\end{array}}
\ar@{|->}[dd]^-{\mbox{\scriptsize\tt expand@$\tau\eta\pi_1$}} 
&
{\begin{array}{l}
\langle S \uplus \{\pi_1,\! \pi_2, \!\pi'_2,\!\tau, \!\eta\}, T\uplus \{\texttt{expand@}\tau,\eta,\pi_2\} \rangle   \\
\where I\land  P_1\texttt{>=} P_2
               \land \texttt{expand@}\tau\eta\pi_1 \not\in T 
\end{array}}
 \ar@{|->}[d]^-{\mbox{\scriptsize\tt expand@$\tau\eta\pi_1$}}
 \\
 %
&
{\begin{array}{l}
\langle S\uplus \{\pi_1,\! \pi'_1,\! \pi_2,\! \pi'_2,\!\tau,\! \eta\}, T\uplus \{\texttt{expand@}\tau\eta\pi_2,  \texttt{expand@}\tau\eta\pi_1 \} \rangle\\
\where I\land  P_1\texttt{>=} P_2
              \end{array}}
 \ar@{|->}[d]^{\mbox{\scriptsize\tt prune@$\pi_1$,$\pi_2$}} 
\\
{\begin{array}{l}\langle S \uplus \{\pi_1,\!\pi'_1, \!\tau, \!\eta\}, T\backslash\pi_2 \uplus \{\texttt{expand@}\tau\eta\pi_1\} \rangle\\
\where I\land  P_1\texttt{>=} P_2\end{array}}
& 
{\begin{array}{l}\langle S\uplus \{\pi_1, \!\pi'_1,\! \pi'_2,\!\tau,\! \eta\}, T\backslash\pi_2 \uplus \{\texttt{expand@}\tau,\eta,\pi_1 \} \rangle\\
 \where  I\land  P_1\texttt{>=} P_2 
\ar@{|->}[l]_{\mbox{\scriptsize\tt prune@$\pi'_1$,$\pi'_2$}} \end{array}}
 }
     \end{array}
   \end{displaymath}

\noindent{$\mathbf{\Lambda}_3^{(2)}$: \emph{Split of} $\mathbf{\Lambda}_3$\emph{;} $\pi_1$ \emph{already expanded into} $\pi_1'$\emph{;}
$\pi_1'$ \emph{still in state}

 \begin{displaymath} \scriptsize
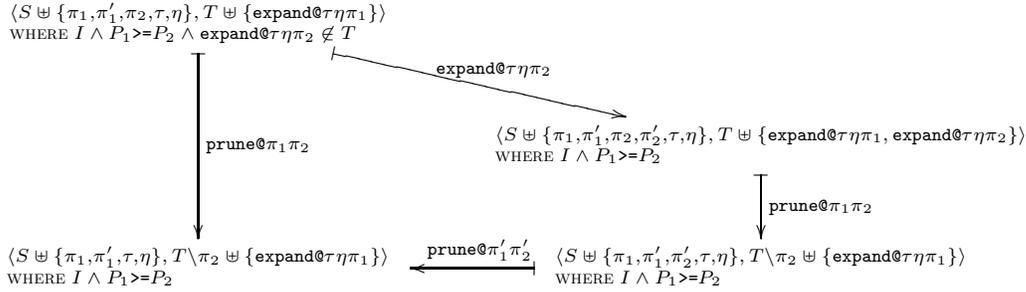
        \begin{array}{c}
     \xymatrix{  {\begin{array}{l}\langle S\uplus \{\pi_1,\! \pi'_1,\! \pi_2,\! \tau, \!\eta\}, T\uplus \{\texttt{expand@}\tau\eta\pi_1 \} \rangle\\
     \where I \land  P_1\texttt{>=} P_2
     \land \texttt{expand@}\tau\eta\pi_2 \not\in T
     \end{array}}
 \ar@{|->}[dr]^-{\mbox{\scriptsize\tt expand@$\tau\eta\pi_2$}}     
 \ar@{|->}[dd]^{\mbox{\scriptsize\tt prune@$\pi_1\pi_2$}}    & \\
 &
 {\begin{array}{l}\langle S\uplus \{\pi_1,\! \pi'_1,\! \pi_2,\! \pi'_2,\!\tau,\! \eta\}, T\uplus \{\texttt{expand@}\tau\eta\pi_1,   \texttt{expand@}\tau\eta\pi_2\} \rangle\\
  \where I \land  P_1\texttt{>=} P_2  \end{array}}
 \ar@{|->}[d]^{\mbox{\scriptsize\tt prune@$\pi_1\pi_2$}} 
 \\
{\begin{array}{l} \langle S\uplus \{\pi_1,\! \pi'_1,\! \tau,\! \eta\}, T\backslash\pi_2\uplus \{\texttt{expand@}\tau\eta\pi_1 \} \rangle\\ \where I \land  P_1\texttt{>=} P_2  \end{array}}
 &
{\begin{array}{l} \langle S\uplus \{\pi_1,\! \pi'_1,\! \pi'_2,\!\tau,\! \eta\}, T\backslash\pi_2\uplus \{\texttt{expand@}\tau\eta\pi_1 \} \rangle\\
 \where I \land  P_1\texttt{>=} P_2 \end{array}}
 \ar@{|->}[l]_{\mbox{\scriptsize\tt prune@$\pi'_1\pi'_2$}} 
 }
     \end{array}
   \end{displaymath}   
Notice for the last {\tt prune@$\pi'_1\pi'_2$} step, that the application history
has no mentioning of $\pi'_2$, as the only event, since it was produced, 
is the step labelled   {\tt prune@$\pi_1\pi_2$}.

\medskip
\noindent{$\mathbf{\Lambda}_3^{(3)}$: \emph{Split of} $\mathbf{\Lambda}_3$\emph{;} $\pi_1$ \emph{already expanded into} $\pi_1'$\emph{;}
$\pi_1'$ \emph{already removed}

\smallskip\noindent
As given by Proposition~\ref{prop:viterbi-invariant}, option 3, this implies the presence in the common ancestor
state of a \texttt{path} constraints $\pi_3$, with sufficiently high probability to 
have pruned $\pi_1'$ as well as a possible $\pi_2'$ (expanded from  $\pi_2'$ using
$\tau$ and $\eta$). We can thus write this abstract corner and
expand it to an abstract joinability diagram as follows.
   \begin{displaymath} \scriptsize
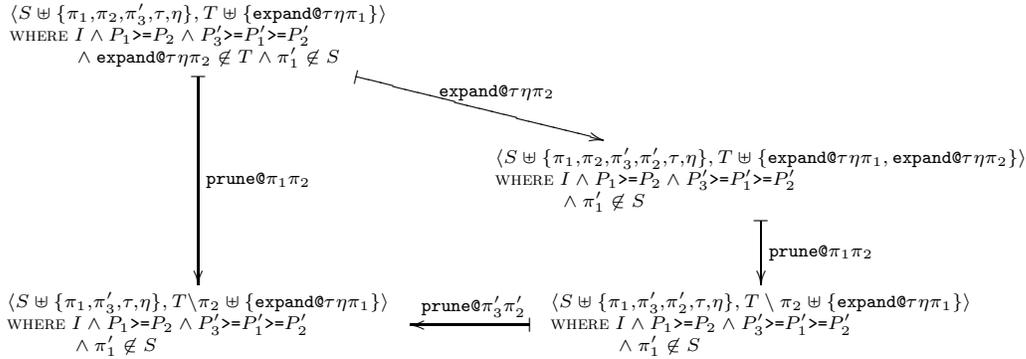
        \begin{array}{c}
     \xymatrix{ 
     {\begin{array}{l}\langle S \uplus \{\pi_1,\! \pi_2,\! \pi'_3,\! \tau,\! \eta\}, T\uplus \{\texttt{expand@}\tau\eta\pi_1 \} \rangle\\
      \where I \land  P_1\texttt{>=} P_2 \land P'_3 \texttt{>=} P_1'\texttt{>=} P_2'\\
     \hbox to 3em{}\land \texttt{expand@}\tau\eta\pi_2 \not\in T
     \land \pi_1'\not\in S
      \end{array}}
 \ar@{|->}[dr]^-{\mbox{\scriptsize\tt expand@$\tau\eta\pi_2$}}     
 \ar@{|->}[dd]^{\mbox{\scriptsize\tt prune@$\pi_1\pi_2$}}    & \\
 &
{\begin{array}{l} \langle S \uplus \{\pi_1,\! \pi_2,\! \pi'_3, \!\pi'_2,\!\tau,\! \eta\}, T\uplus \{\texttt{expand@}\tau\eta\pi_1,  \texttt{expand@}\tau\eta\pi_2\} \rangle\\
      \where I \land  P_1\texttt{>=} P_2 \land P'_3 \texttt{>=} P_1'\texttt{>=} P_2'\\
     \hbox to 3em{}\land \pi_1'\not\in S
 \end{array}}
 \ar@{|->}[d]^{\mbox{\scriptsize\tt prune@$\pi_1\pi_2$}} 
 \\
 {\begin{array}{l}\langle S\uplus \{\pi_1,\! \pi'_3,\! \tau,\! \eta\}, T\backslash\pi_2\uplus \{\texttt{expand@}\tau\eta\pi_1 \} \rangle  \\
 \where
        I \land  P_1\texttt{>=} P_2 \land P'_3 \texttt{>=} P_1'\texttt{>=} P_2'\\
     \hbox to 3em{}\land \pi_1'\not\in S
   \end{array}}
 &
{\begin{array}{l} \langle S\uplus \{\pi_1,\! \pi'_3,\! \pi'_2,\!\tau,\! \eta\}, T\setminus\pi_2\uplus \{\texttt{expand@}\tau\eta\pi_1 \} \rangle\\
      \where I \land  P_1\texttt{>=} P_2 \land P'_3 \texttt{>=} P_1'\texttt{>=} P_2'\\
     \hbox to 3em{}\land \pi_1'\not\in S
\end{array}}
 \ar@{|->}[l]_{\mbox{\scriptsize\tt prune@$\pi'_3\pi'_2\;\;$}} 
 }
     \end{array}
   \end{displaymath}  
For the last {\tt prune@$\pi'_3\pi'_2$} step, the application history
has no mentioning of $\pi'_2$, as the only event since it was produced 
is the step labelled   {\tt prune@$\pi_1\pi_2$}.

This finishes the proof that $\mathbf{\Lambda}_3$ is split joinable.
Now we turn to the canonical abstract $\beta$-corners of which there are two,
$\mathbf{\Lambda}_4$ and  $\mathbf{\Lambda}_5$,
one for each program rule.

\medskip
\noindent{$\mathbf{\Lambda}_4$: \emph{Equivalence and the} \texttt{expand} \emph{rule}

\smallskip\noindent
For the two equivalent states on the left side, it holds that
{\tt expand@$\tau\eta\pi_2$} = {\tt expand@$\tau\eta\pi_1$}
and that $S_i=\mathit{HMM}\uplus S_i'$, $i=1,2$ where
$S'_1$ and $S'_2$ consist of pairwise similar \texttt{path} constraints with identical index and that may differ
only in their last arguments, and similarly for $\pi_1$ and $\pi_2$.
\begin{displaymath} \scriptsize \begin{array}{c}
\xymatrix{ 
\langle S_1 \uplus \{\pi_1,\! \tau,\! \eta\}, T\rangle \where I \land \texttt{expand@}\tau\eta\pi_1\not\in T
\ar@2{~}[d] \ar@{|->}[dr]^-{\mbox{\scriptsize\tt expand@$\tau\eta\pi_1$}}
\\
\langle S_2 \uplus \{\pi_2,\! \tau,\! \eta\}, T \rangle \where I \land \texttt{expand@}\tau\eta\pi_2\not\in T 
\ar@{|->}[dr]_-{\mbox{\scriptsize\tt expand@$\tau\eta\pi_2$}} &
\langle S \uplus \{\pi_1,\! \pi'_1,\!\tau,\! \eta\}, T\uplus \{\texttt{expand@}\tau\eta\pi_1\} \rangle \where I
\ar@2{~}[d] 
\\
& \langle S_2 \uplus \{\pi_2,\! \pi'_2,\!\tau,\! \eta\}, T\uplus \{\texttt{expand@}\tau\eta\pi_2\} \rangle \where I
}
\end{array}\end{displaymath}  
To see that the lower equivalence holds, we notice that the indices of
 $\pi_1'$ and $\pi_2'$ can be chosen identical (and different from any other index used),
 and they may differ only in their last arguments.

\medskip
\noindent{$\mathbf{\Lambda}_5$: \emph{Equivalence and the} \texttt{prune} \emph{rule}

\smallskip\noindent
For the two equivalent states on the left side, it holds that
{\tt expand@$\tau\eta\pi_{i+2}$} = {\tt expand@$\tau\eta\pi_i$}, $i=1,2$,
and that $\pi_{i+2}$ and $\pi_{i}$ , $i=1,2$, may differ
only in their last arguments.
Furthermore, $S_i=\mathit{HMM}\uplus S_i'$, $i=1,2$ where
$S'_1$ and $S'_2$  consist of pairwise similar \texttt{path} constraints with identical index and that may differ
only in their last arguments.
%
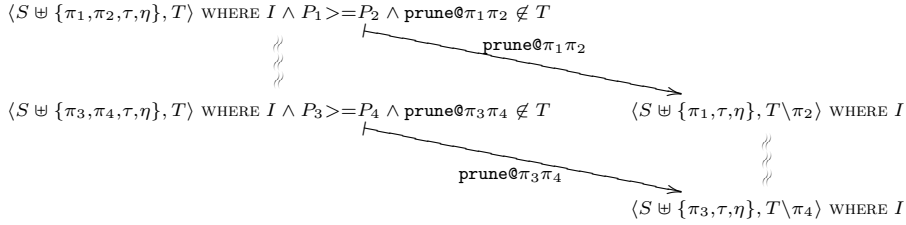
\begin{displaymath} \scriptsize \begin{array}{c}
\xymatrix{ 
\langle S \uplus \{\pi_1,\! \pi_2,\!\tau,\! \eta\}, T \rangle \where 
I\land P_1{\tt>=}P_2\land  \mbox{\scriptsize\tt prune@$\pi_1\pi_2$}\not\in T
\ar@2{~}[d] \ar@{|->}[dr]^-{\mbox{\scriptsize\tt prune@$\pi_1\pi_2$}}
\\
\langle S\uplus \{\pi_3,\!\pi_4,\! \tau,\! \eta\}, T \rangle \where 
     I\land P_3{\tt>=}P_4\land  \mbox{\scriptsize\tt prune@$\pi_3\pi_4$}\not\in T
\ar@{|->}[dr]_-{\mbox{\scriptsize\tt prune@$\pi_3\pi_4$}} &
 \langle S \uplus \{\pi_1,\! \tau,\!\eta\}, T\backslash \pi_2 \rangle \where I
\ar@2{~}[d] 
\\
&  \langle S \uplus \{\pi_3,\! \tau,\!\eta\}, T\backslash \pi_4 \rangle \where I
}
\end{array}\end{displaymath}  
Thus the set of abstract, critical corners have been shown joinable or split joinable;
by termination and Theorem~\ref{thm:termination+abstract-joinability=confluence}, the program is confluent modulo $\approx$.
\hskip 1em$\Box$
\subsection{Countably Infinite Splitting}\label{sec:infinite-split}
Here we show a program whose proof of confluence needs an infinite splitting of 
an abstract critical corner.
The following CHR program is intended for queries of the form
\texttt{start,} \texttt{c(s$^n$(0))},
where \texttt{s$^n$(0)} denotes the $n$th successor of {\tt 0} for any $n\ge 0$, e.g., \texttt{s$^2$(0)} = \texttt{s(s(0))}.

{\small\begin{verbatim}
easy    @   start          <=> easy.
hard    @   start          <=> hard.
done    @   c(X), easy     <=> c(0), end.
step    @   hard \ c(s(X)) <=> c(X).
finally @   c(0) \ hard    <=> end.
\end{verbatim}}\noindent
The first step in such a derivation will introduce either an \texttt{easy}
or a \texttt{hard} constraint.
In case of \texttt{easy}, the derivation terminates after one additional step in the state  $\{\texttt{c(0)}, \texttt{end}\}$.
In case of \texttt{hard}, the derivation terminates after $n+1$ steps in the same state,
so the program is confluent (modulo trivial ${\approx}={=}$) under the invariant implied by the intended initial states.

We can specify the invariant as follows, using the unary meta-level predicate $\mathit{succ}(N)$, satisfied if and only if $N$
of the form \texttt{s$^n$(0)} for an arbitrary natural number $\ge0$.
$I(\langle S,T \rangle)$ holds if and only if
\begin{itemize}
 \item $S = \{R, \texttt{c($N$)}\}$ where $R \in \{{\tt start}, {\tt hard}, {\tt easy},{\tt end} \}$ and  $\mathit{succ}(N)$,
 \item $T=\emptyset$.
\end{itemize}
There exists only one consistent abstract critical $I$-corner $\mathbf{\Lambda}$, and it is
based on the overlap of the rules \texttt{easy} and \texttt{hard}.
Notice that the invariant  has been unfolded, which is a semantics-preserving transformation. 
$$ \mathbf{\Lambda} = 
\begin{array}{c}
\xymatrix{ 
  \langle \{{\tt start},{\tt c}(n)\} , \emptyset \rangle \where \mathit{succ}(n)
 \ar@{|->}[dr]^-{\qquad \mbox{\scriptsize\tt hard@start}}     
  \ar@{|->}[d]_{\mbox{\scriptsize\tt easy@start}}    
 & \\
\langle \{{\tt easy},{\tt c}(n)\} , \emptyset \rangle \where \mathit{succ}(n) &
\langle \{{\tt hard},{\tt c}(n)\} , \emptyset \rangle \where \mathit{succ}(n) \\ 
}
\end{array}
$$
The abstract critical corner $\mathbf{\Lambda}$ is not joinable
as no single joinability diagram applies for all concrete corners covered by  $\mathbf{\Lambda}$.
Therefore we split $\mathbf{\Lambda}$ 
using the infinite disjunction $\mathit{succ}(n) \Leftrightarrow n=\texttt{0}\lor n=\texttt{s(0)}\lor\cdots$.

This leads to a countably infinite set of abstract corners $\mathbf{\Lambda}_0$,  $\mathbf{\Lambda}_1$, \ldots, where
$$ \mathbf{\Lambda}_i = 
\begin{array}{c}
\xymatrix{ 
  \langle \{{\tt start},\texttt{c(s$^n$(0)}\} , \emptyset \rangle
 \ar@{|->}[dr]^-{\qquad\mbox{\scriptsize\tt hard@start}}     
  \ar@{|->}[d]_{\mbox{\scriptsize\tt easy@start}}    
 & \\
\langle \{{\tt easy},\texttt{c(s$^n$(0)}\} , \emptyset \rangle&
\langle \{{\tt hard},\texttt{c(s$^n$(0)}\} , \emptyset \rangle  \\ 
}
\end{array}
$$
Each such abstract corner can be extended into a joinability diagram $\boldsymbol{\Delta}_i$, each having $i+4$ abstract states
and the same number of abstract derivation steps.
For a better overview, we indicate only the shapes of these diagrams;
the actual states are uniquely determined by the rules applied.
$$
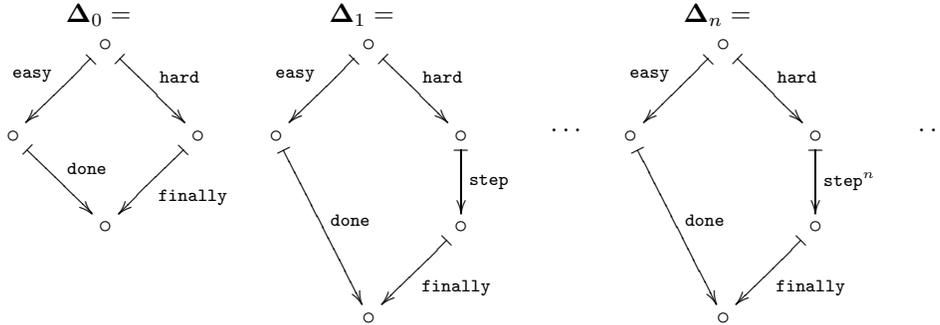
\begin{array}{llllll}
 \qquad\;\;\boldsymbol{\Delta}_0= &   \qquad\;\;\boldsymbol{\Delta}_1= &  & \qquad\;\;\boldsymbol{\Delta}_n= &  \\
 \xymatrix{
 &
 \circ 
 \ar@{|->}[dr]^-{\mbox{\scriptsize\tt hard}} 
 \ar@{|->}[dl]_{\mbox{\scriptsize\tt easy}}
 &\\
 \circ \ar@{|->}[dr]^-{\mbox{\scriptsize\tt done}}   & &  \circ \ar@{|->}[dl]^-{\mbox{\scriptsize\tt finally}} \\
  & \circ & \\
 } 
 & 
   \xymatrix{
 &
 \circ 
 \ar@{|->}[dr]^-{\mbox{\scriptsize\tt hard}} 
 \ar@{|->}[dl]_{\mbox{\scriptsize\tt easy}}
 &\\
 \circ \ar@{|->}[ddr]^-{\mbox{\scriptsize\tt done}}  & &   \circ \ar@{|->}[d]^-{\mbox{\scriptsize\tt step}}   \\
  & &   \circ \ar@{|->}[dl]^-{\mbox{\scriptsize\tt finally}}\\
  &\circ& \\
 } 
  & \xymatrix{ \\ \ldots\\ \\ \\}
  &
   \xymatrix{
 &
 \circ 
 \ar@{|->}[dr]^-{\mbox{\scriptsize\tt hard}} 
 \ar@{|->}[dl]_{\mbox{\scriptsize\tt easy}}
 &\\
 \circ \ar@{|->}[ddr]^-{\mbox{\scriptsize\tt done}}  & &   \circ \ar@{|->}[d]^-{\mbox{\scriptsize\tt step}^{\scriptscriptstyle n}}   \\
  & &   \circ \ar@{|->}[dl]^-{\mbox{\scriptsize\tt finally}}\\
  &\circ& \\
 } 
%
   & \xymatrix{ \\ \ldots\\ \\ \\}\\
\end{array}
$$
Thus  $\mathbf{\Lambda}$  is split joinable, the program is terminating (no derivation starting from
a state containing $\texttt{c(s$^n$(0))}$ includes more that $n+2$ steps),
and by Theorem~\ref{thm:termination+abstract-joinability=confluence} it follows that the program is  confluent (modulo =).

We notice here that confluence is due to the invariant; without invariant, we would get instead of $\mathbf\Lambda$, the  corner
$\langle \{{\tt easy},{\tt c}(x)\}\rangle\mapsfrom\langle \{{\tt start},{\tt c}(x)\} , \emptyset\rangle\ourmapsto\langle \{{\tt hard},{\tt c}(x)\}\rangle$.
It is neither joinable
nor split joinable.

\section{Conclusions and Future Work}\label{sec:discussion} 
The aim of this paper is both theoretical and practical.
Practical as it points forward to methods for proving highly useful properties of realistic CHR programs that may
identify possible  optimizations and contribute to correctness proofs; and theoretical
since it provides a firm basis for understanding
the notion of confluence modulo equivalence applied in the context of CHR.

We have demonstrated the relevance of confluence modulo equivalence for
Constraint Handling Rules, which may also inspire to apply the concept
to other systems with nondeterministic choice and parallelism.
This may be approached either by migrating our results to other
types of derivation systems, or using the fact that programs and systems of many such paradigms can
be mapped directly into CHR programs; see an overview
in the book by Fr\"uhwirth~\cite{fru_chr_book_2009}.

We introduced a new operational semantics for CHR that includes non-logical and incomplete built-ins
and, as we have argued, this semantics is in many respects more in accordance with concrete implementations
of CHR that what is seen in earlier work.

We introduced the idea of a logical meta-language \metachr{} specifically intended
for reasoning about CHR programs,
their semantics and their proofs of confluence modulo equivalence. These proofs are
reified as collections of abstract joinability diagrams.
A main advantage of this approach
is that we can parameterize such proofs, i.e., diagrams, by meta-variables
constrained at the meta-level to stand for, say, variables  or nonvariable terms of CHR.
In our approach, this is essential for handling non-logical and incomplete built-ins correctly.


Our work is an improvement of the state-of-art in confluence proving for CHR~\cite{DBLP:conf/cp/Abdennadher97,DBLP:conf/cp/AbdennadherFM96,DBLP:journals/constraints/AbdennadherFM99,DBLP:conf/iclp/DuckSS07} in several ways: we generalize to modulo equivalence, we handle a larger and more realistic class of CHR programs, and for many programs we can reduce to a finite number of proof cases where \cite{DBLP:conf/iclp/DuckSS07} needs infinitely many, even for simple invariants such as groundness.   
%
The foundational
works by Abdennadher et al~\cite{DBLP:conf/cp/Abdennadher97,DBLP:conf/cp/AbdennadherFM96,DBLP:journals/constraints/AbdennadherFM99} and Duck et al~\cite{DBLP:conf/iclp/DuckSS07}
use ordinary substitutions and inclusion of more constraints as their way to explain
how their abstract cases, called critical pairs, cover large classes of concrete such pairs, each
required to be joinable to ensure  confluence.
The use of the same language for abstract and concrete cases is quite limiting
for what can be done at the abstract level, and which causes 
the mentioned problem of inifinitely many proof cases.
Taking the step that we do, introducing an explicit meta-language with meta-level constraints,
eliminates this problem.

The use of a formal language provides a firm basis for automatic or
semi-automatic support for deriving actual proofs,
and our future plans include the development of such an implemented system.
This  requires a better understanding of how in general
to construct abstract post states, given a state and an abstract derivation step;
this is an important topic in our forthcoming research.
It is obvious to incorporate an existing confluence
checker in such a system in order to identify and eliminate those $\alpha_1$
corners that are joinable even when invariant and equivalence are ignored.

One practical issue that needs to be understood better is how to cope with infinite
splittings which have been exposed in our examples.
It may be considered to allow  meta-variables in \metachr{} to
range over entire sub-derivations, suitably constrained at the meta-level.
This may give rise to abstract diagrams that  cover (in the formal sense we have defined)
a range of differently shaped concrete diagrams.
This potential is indicated informally in a diagram shown in
Section~\ref{sec:infinite-split}, with a component indicating an entire sub-derivation,
written as $\stackrel{\texttt{step}^n}\ourmapsto$,
so that we could show (still informally) an infinite set of corners joinable
with a single argument.

A more detailed analysis of  $\beta$-corners is desirable.
We did not assume or impose any specific way of defining an equivalence, which means that
any abstract $\beta$-corner needs to be considered as critical as soon as the equivalence is non-trivial.
Huet~\cite{DBLP:journals/jacm/Huet80} has shown a lemma for term rewriting systems, which will be interesting to adapt for CHR (see our~Lemma~\ref{lemma:alpga-and-gamma}, p.~\pageref{lemma:alpga-and-gamma}).
It applies ${\approx}  = (\vdash\!\dashv)^{*}$ for some symmetric relation $\vdash\!\dashv$.
Such a relation may be specified by a finite number of cases, as in a system of equations or logical equivalences.
Here it seems possible to split each of our   $\beta$-corners into a  number
of sub-cases, one for each case of the inductive definition of $\vdash\!\dashv$.

%

\bibliographystyle{spmpsci}
\bibliography{CHR}

\input{AppendixProofs.tex}

\end{document}

%% file: AppendixProofs.tex
\appendix
\section{Proof of Lemma~\ref{lem:critCorner}: the Critical Corner Lemma}
We recall the notation
$\textit{all-relevant-app-recs}(S)$, Definition~\ref{def:state}, p.~\pageref{def:state}, that refers to the set
of all application records for rules of the current program taking indices
from the constraint store $S$.

\begin{proof} \emph{(Lemma~\ref{lem:critCorner}, p.~\pageref{lem:critCorner})}
We consider a program with invariant $I$ and equivalence $\approx$, and we will go through
the possible ways that an $I$-corner $\lambda$ can be non-joinable and in each such case point out
a most general pre-corner $\Lambda$ that subsumes $\lambda$.

\medskip\noindent
{\large{$\alpha_1$:}}\\
Let $\lambda=(\Sigma_1 \stackrel{R_1}\mapsfrom \Sigma \stackrel{R_2}\ourmapsto\Sigma_2)$
be an $\alpha_1$-corner that is not joinable
with application instances $R_k=$\break$(r_k\colon A_k\backslash B_k\texttt{<=>}g_k|C_k)$ for $k=1,2$.
Let now $H_k=A_k\cup B_k$,  $k=1,2$, and $\mathit{Overlap}=(B_1\cap H_2)\cup(B_2\cap H_1)$.

In case $\mathit{Overlap}=\emptyset$,  none of the application instances of $\lambda$ 
remove any constraint from the common ancestor state that prevents the other one from being successively applied.
Thus there exists some state $\Sigma'$ such that
$\Sigma_1\stackrel{R_2}\ourmapsto\Sigma'\stackrel{R_1}\mapsfrom\Sigma_2$
which means $\lambda$ is joinable.

Assume now that $\mathit{Overlap}\neq\emptyset$ and we proceed as follows to produce a
most general $\alpha$-pre-corner $\Lambda$ as follows.
We select two most general application pre-instances
$$
R_k^0\; = \; \bigl(r_k\colon  A_k^0\backslash B_k^0\texttt{<=>}g_k^0|C_k^0\bigr),\quad i=1,2
$$
in such a way such that, for $k=1,2$, the indices 
in $R_k^0$ and $R_k$ are pairwise identical, compared in the order they appear.
Define also $H_k^0=A_k^0\cup B_k^0$,  $k=1,2$.

Let now, for $k=1,2$, $\mathit{Overlap}^0_k$ be the set of constraints in $R_k^0$
whose indices coincide with those of $\mathit{Overlap}$.
Since $\mathit{Overlap}^0_1$ and $\mathit{Overlap}^0_2$ have the common instance
$\mathit{Overlap}$, there exists a most general unifier $\sigma$ of $\mathit{Overlap}^0_1$ and $\mathit{Overlap}^0_2$;
let furthermore $\theta$ be a smallest substitution that such that
$\mathit{Overlap}^0_1\sigma\theta=\mathit{Overlap}^0_2\sigma\theta=\mathit{Overlap}$. 

Noticing that $(g_1^0\sigma, g_2^0\sigma)$ is satisfiable (by $\theta$), we can
define now the following most general critical $\alpha_1$-pre-corner, that we argue below
subsumes $\lambda$.
\begin{eqnarray*}
\Lambda^0 & \;=\;
  & \bigl(  \circ\stackrel{R_1^0\sigma}\mapsfrom
               \langle S^0, T^0\rangle
               \stackrel{R_2^0\sigma}\ourmapsto\circ
       \bigl),\quad \text{where }\\ S^0  &\;=\;&  H_1^0\sigma\cup H_2^0\sigma \\
T^0 & \;=\; &
     \textit{all-relevant-app-recs}(H_1^0\sigma\cup H_2^0\sigma)
                \;\setminus\; \{r_1@id(A_1B_1), r_2@id(A_2B_2) \}  
\end{eqnarray*}
We should check that, if $r_1=r_2$, then 
$A^0_1\sigma \neq A^0_2\sigma$ or $B^0_1\sigma \neq B^0_2\sigma$ (cf.~Def.~\ref{def:pre-corners})
 in order for the indicated $\Lambda^0$
to actually be a most general $\alpha_1$-pre-corner: if this is not the case,
$\Sigma_1$ and $\Sigma_2$ would be identical and thus $\lambda$ joinable; contradiction.

Let now $\Sigma=\langle s, t\rangle$ (i.e., we name the parts of the common ancestor in $\lambda$)
and we can define $s^+$, $t^+$ and $t^\div$ as follows
\begin{eqnarray*}
s^+ & = & s\setminus S^0\theta \\
t^+  &= & t\setminus T^0 \\
t^\div  &=  & T^0\setminus t  
\end{eqnarray*}
By construction, $R_k^0\sigma\theta=R_k$, $k=1,2$, and 
we can show that the following properties hold. 
\begin{eqnarray*}
s & = & S^0\theta \uplus s^+  \\
t  &= & T^0 \uplus t^+ \setminus t^\div \\
%
%
t^+  & \subseteq & \allrars(S\theta\uplus s^{+})\setminus\allrars(S\theta)\\ 
t^\div  &\subseteq  & \allrars(S\theta)
\end{eqnarray*}
Thus, the conditions of Definition \ref{def:subsumption-by-gcpc} are satisfied, proving that $\Lambda^0$ subsumes $\lambda$.

\medskip\noindent
{\large{$\alpha_2$:}}\\
Let $\lambda=(\Sigma_1 \stackrel{R}\mapsfrom \Sigma \stackrel{b}\ourmapsto\Sigma_2)$
be an $\alpha_2$ corner that is not joinable,
with application instance $R=(r\colon A\backslash B\texttt{<=>}g|C)$, $g$ non-logical or $I$-incomplete, and built-in $b$.
Define now the following most general application pre-instance
$$
R^0\; = \; \bigl(r\colon  H^0\texttt{<=>}g^0|C^0\bigr)
$$
in such a way such that
the indices 
in $R^0$ and $R$ are pairwise identical, compared in the order they appear,
and let furthermore $b^0$ be a most general indexed built-in atom with same predicate
and index as $b$.
We define now the following most general critical $\alpha_2$-pre-corner that we will argue subsumes $\lambda$.
\begin{eqnarray*}
\Lambda^0\ & \;=\;
  & \bigl(  \circ\stackrel{R^0}\mapsfrom
               \langle H^0, T^0\rangle
               \stackrel{b^0}\ourmapsto\circ
       \bigl),\quad \text{where} \\
T^0 & \;=\; &
     \textit{all-relevant-app-recs}(H^0)
                \;\setminus\; \{r@id(H^0) \}  
\end{eqnarray*}
Let  $\theta$ be a smallest substitution such that $R^0\theta=R$ and $b^0\theta=b$.
If $\vars(b)\cap\vars(g)=\emptyset$, 
there would exists a state $\Sigma'$ such that
$\Sigma_1\stackrel{b}\ourmapsto\Sigma'\stackrel{R}\mapsfrom\Sigma_2$; contradiction.
The remaining arguments to show that $\Lambda^0$ subsumes $\lambda$ are exactly as
for the $\alpha_1$ case.

\medskip\noindent
{\large{$\alpha_3$:}}\\
Let $\lambda=(\Sigma_1 \stackrel{b_1}\mapsfrom \Sigma \stackrel{b_2}\ourmapsto\Sigma_2)$
be an $\alpha_3$-corner that is not joinable,
with built-ins $b_1,b_2$, where $b_1$ is non-logical or $I$-incomplete.
We define now the following most general critical $\alpha_2$-pre-corner that we will argue subsumes $\lambda$.
\begin{eqnarray*}
\Lambda^0\ & \;=\;
  & \bigl(  \circ\stackrel{b_1^0}\mapsfrom
               \langle H^0, \emptyset\rangle
               \stackrel{b_2^0}\ourmapsto\circ
       \bigl) 
\end{eqnarray*}
To show that $\Lambda^0$ subsumes $\lambda$, let $\theta$ be a smallest substitution such that
$b_k^0\theta=b_k$, $k=1$, and proceed exactly as in the $\alpha_1$ case (with $T_0=\emptyset$).
We should also notice that it must hold that $\vars(b_1)\cap\vars(b_2)\neq\emptyset$ as otherwise
$b_1$ and $b_2$ would commute and $\lambda$ be joinable.

\medskip\noindent
{\large{$\beta_1$:}}\\
Let $\lambda=(\Sigma_1 \approx \Sigma \stackrel{R}\ourmapsto\Sigma_2)$
be a $\beta$-corner that is not joinable,
where $R$ is an application instance
with application instance $R=(r\colon A\backslash B\texttt{<=>}g|C_k)$.
Define now the following most general application pre-instance
$$
R^0\; = \; \bigl(r\colon  H^0\texttt{<=>}g^0|C^0\bigr)
$$
in such a way such that
the indices 
in $R^0$ and $R$ are pairwise identical, compared in the order they appear.
The proof that the following most general critical $\beta_1$-pre-corner subsumes $\lambda$
is similar to the previous cases.
\begin{eqnarray*}
\Lambda^0\ & \;=\;
  & \bigl(  \circ\approx
               \langle H^0, T^0\rangle
               \stackrel{R^0}\ourmapsto\circ
       \bigl),\quad \text{where} \\
T^0 & \;=\; &
     \textit{all-relevant-app-recs}(H^0)
                \;\setminus\; \{r@id(H^0) \}  
\end{eqnarray*}

\medskip\noindent
{\large{$\beta_2$:}}\\
Analogous to the $\beta_1$ case and omitted.
\end{proof}

%% file: ConfMedEqHenningMaja_111116.bbl
\begin{thebibliography}{10}
\providecommand{\url}[1]{{#1}}
\providecommand{\urlprefix}{URL }
\expandafter\ifx\csname urlstyle\endcsname\relax
  \providecommand{\doi}[1]{DOI~\discretionary{}{}{}#1}\else
  \providecommand{\doi}{DOI~\discretionary{}{}{}\begingroup
  \urlstyle{rm}\Url}\fi

\bibitem{DBLP:conf/cp/Abdennadher97}
Abdennadher, S.: Operational semantics and confluence of constraint propagation
  rules.
\newblock In: G.~Smolka (ed.) CP, Constraint Programming, \emph{Lecture Notes
  in Computer Science}, vol. 1330, pp. 252--266. Springer (1997)

\bibitem{DBLP:conf/lopstr/AbdennadherF03}
Abdennadher, S., Fr{\"{u}}hwirth, T.W.: Integration and optimization of
  rule-based constraint solvers.
\newblock In: M.~Bruynooghe (ed.) Logic Based Program Synthesis and
  Transformation, 13th International Symposium {LOPSTR} 2003, Uppsala, Sweden,
  August 25-27, 2003, Revised Selected Papers, \emph{Lecture Notes in Computer
  Science}, vol. 3018, pp. 198--213. Springer (2003)

\bibitem{DBLP:conf/cp/AbdennadherFM96}
Abdennadher, S., Fr{\"u}hwirth, T.W., Meuss, H.: On confluence of {C}onstraint
  {H}andling {R}ules.
\newblock In: E.C. Freuder (ed.) CP, \emph{Lecture Notes in Computer Science},
  vol. 1118, pp. 1--15. Springer (1996)

\bibitem{DBLP:journals/constraints/AbdennadherFM99}
Abdennadher, S., Fr{\"{u}}hwirth, T.W., Meuss, H.: Confluence and semantics of
  constraint simplification rules.
\newblock Constraints \textbf{4}(2), 133--165 (1999)

\bibitem{Aho72}
Aho, A.V., Sethi, R., Ullman, J.D.: Code optimization and finite
  {Church-Rosser} systems.
\newblock In: R.~Rustin (ed.) Design and Optimization of Compilers, pp.
  89--106. Prentice-Hall (1972)

\bibitem{DBLP:conf/sigmod/AikenWH92}
Aiken, A., Widom, J., Hellerstein, J.M.: Behavior of database production rules:
  Termination, confluence, and observable determinism.
\newblock In: M.~Stonebraker (ed.) Proceedings of the 1992 {ACM} {SIGMOD}
  International Conference on Management of Data, San Diego, California, June
  2-5, 1992., pp. 59--68. {ACM} Press (1992)

\bibitem{DBLP:journals/aaecc/AptMP94}
Apt, K.R., Marchiori, E., Palamidessi, C.: A declarative approach for
  first-order built-in's of {P}rolog.
\newblock Appl. Algebra Eng. Commun. Comput. \textbf{5}, 159--191 (1994)

\bibitem{BaderNipkow1999}
Baader, F., Nipkow, T.: Term rewriting and all that.
\newblock Cambridge University Press (1999)

\bibitem{DBLP:journals/tplp/BetzRF10}
Betz, H., Raiser, F., Fr{\"{u}}hwirth, T.W.: A complete and terminating
  execution model for constraint handling rules.
\newblock {TPLP} \textbf{10}(4-6), 597--610 (2010)

\bibitem{ChristiansenTPLP2005}
Christiansen, H.: {CHR Grammars}.
\newblock Int'l Journal on Theory and Practice of Logic Programming
  \textbf{5}(4-5), 467--501 (2005)

\bibitem{ChristiansenEtAlCHR2010}
Christiansen, H., Have, C.T., Lassen, O.T., Petit, M.: The {V}iterbi algorithm
  expressed in {C}onstraint {H}andling {R}ules.
\newblock In: P.~{Van Weert}, L.~{De Koninck} (eds.) Proceedings of the 7th
  International Workshop on {C}onstraint {H}andling {R}ules, Report CW 588, pp.
  17--24. Katholieke Universiteit Leuven, Belgium (2010)

\bibitem{DBLP:conf/lopstr/ChristiansenK14}
Christiansen, H., Kirkeby, M.H.: Confluence modulo equivalence in {C}onstraint
  {H}andling {R}ules.
\newblock In: M.~Proietti, H.~Seki (eds.) Logic-Based Program Synthesis and
  Transformation - 24th International Symposium, {LOPSTR} 2014, Canterbury, UK,
  September 9-11, 2014. Revised Selected Papers, \emph{Lecture Notes in
  Computer Science}, vol. 8981, pp. 41--58. Springer (2014)

\bibitem{drabent-report-1997}
Drabent, W.: {A Floyd-Hoare method for Prolog}.
\newblock Tech. Rep. 2(13), Link\"oping Electronic Articles in Computer and
  Information Science (1997)

\bibitem{DuckSBH04}
Duck, G.J., Stuckey, P.J., de~la Banda, M.J.G., Holzbaur, C.: The refined
  operational semantics of {C}onstraint {H}andling {R}ules.
\newblock In: B.~Demoen, V.~Lifschitz (eds.) Proc. Logic Programming, 20th
  International Conference, ICLP 2004, \emph{Lecture Notes in Computer
  Science}, vol. 3132, pp. 90--104. Springer (2004)

\bibitem{DBLP:conf/iclp/DuckSS07}
Duck, G.J., Stuckey, P.J., Sulzmann, M.: Observable confluence for {C}onstraint
  {H}andling {R}ules.
\newblock In: V.~Dahl, I.~Niemel{\"a} (eds.) ICLP, \emph{Lecture Notes in
  Computer Science}, vol. 4670, pp. 224--239. Springer (2007)

\bibitem{DurbinEtAL99}
Durbin, R., Eddy, S., Krogh, A., Mitchison, G.: Biological Sequence Analysis:
  Probabilistic Models of Proteins and Nucleic Acids.
\newblock Cambridge University Press (1999)

\bibitem{DBLP:journals/tcs/FalaschiGMP97}
Falaschi, M., Gabbrielli, M., Marriott, K., Palamidessi, C.: Confluence in
  concurrent constraint programming.
\newblock Theor. Comput. Sci. \textbf{183}(2), 281--315 (1997)

\bibitem{FruehwirthRaiserEds2011}
Fr\"uhwirth, T., Raiser, F. (eds.): {{C}onstraint {H}andling {R}ules},
  Compilation, Execution, and Analysis.
\newblock Books on Demand GmbH, Norderstedt (2011)

\bibitem{DBLP:conf/iclp/Fruhwirth93}
Fr{\"{u}}hwirth, T.W.: User-defined constraint handling.
\newblock In: D.S. Warren (ed.) Logic Programming, Proceedings of the Tenth
  International Conference on Logic Programming, Budapest, Hungary, June 21-25,
  1993, pp. 837--838. {MIT} Press (1993)

\bibitem{DBLP:journals/lncs/Fruhwirth94}
Fr{\"{u}}hwirth, T.W.: {C}onstraint {H}andling {R}ules.
\newblock In: A.~Podelski (ed.) Constraint Programming: Basics and Trends,
  Ch{\^{a}}tillon Spring School, Ch{\^{a}}tillon-sur-Seine, France, May 16 -
  20, 1994, Selected Papers, \emph{Lecture Notes in Computer Science}, vol.
  910, pp. 90--107. Springer (1994)

\bibitem{fruehwirth-98}
Fr{\"u}hwirth, T.W.: Theory and practice of {{C}onstraint {H}andling {R}ules}.
\newblock Journal of Logic Programming \textbf{37}(1-3), 95--138 (1998)

\bibitem{fru_chr_book_2009}
Fr{\"u}hwirth, T.W.: {C}onstraint {H}andling {R}ules.
\newblock Cambridge University Press (2009)

\bibitem{DBLP:journals/tplp/Haemmerle12}
Haemmerl{\'e}, R.: Diagrammatic confluence for {{C}onstraint {H}andling
  {R}ules}.
\newblock TPLP \textbf{12}(4-5), 737--753 (2012)

\bibitem{Hill94meta-programming-in-logic}
Hill, P., Gallagher, J.: Meta-programming in logic programming.
\newblock In: Handbook of Logic in Artificial Intelligence and Logic
  Programming, pp. 421--497. Oxford Science Publications, Oxford University
  Press (1994)

\bibitem{DBLP:journals/aai/HolzbaurF00a}
Holzbaur, C., Fr{\"{u}}hwirth, T.W.: A {PROLOG} {C}onstraint {H}andling {R}ules
  compiler and runtime system.
\newblock Applied Artificial Intelligence \textbf{14}(4), 369--388 (2000)

\bibitem{DBLP:journals/jacm/Huet80}
Huet, G.P.: Confluent reductions: Abstract properties and applications to
  {K}nuth systems: Abstract properties and applications to {T}erm {R}ewriting
  {S}ystems.
\newblock J. ACM \textbf{27}(4), 797--821 (1980)

\bibitem{KnuthBendix1970}
Knuth, D., Bendix, P.: Simple word problems in universal algebras.
\newblock In: J.~Leech (ed.) Computational Problems in Universal Algebras, pp.
  263--297. Pergamon Press (1970)

\bibitem{Raiser-Langbein2010}
Langbein, J., Raiser, F., Fr{\"u}hwirth, T.W.: A state equivalence and
  confluence checker for {CHRs}.
\newblock In: P.V. Weert, L.D. Koninck (eds.) Proceedings of the 7th
  International Workshop on {C}onstraint {H}andling {R}ules, Report CW 588, pp.
  1--8. Katholieke Universiteit Leuven, Belgium (2010)

\bibitem{DBLP:journals/tcs/MayrN98}
Mayr, R., Nipkow, T.: Higher-order rewrite systems and their confluence.
\newblock Theor. Comput. Sci. \textbf{192}(1), 3--29 (1998)

\bibitem{Newman42}
Newman, M.: On theories with a combinatorial definition of ``equivalence''.
\newblock Annals of Mathematics \textbf{43}(2), 223--243 (1942)

\bibitem{DBLP:journals/jfp/Niehren00}
Niehren, J.: Uniform confluence in concurrent computation.
\newblock J. Funct. Program. \textbf{10}(5), 453--499 (2000)

\bibitem{DBLP:conf/ccl/NiehrenS94}
Niehren, J., Smolka, G.: A confluent relational calculus for higher-order
  programming with constraints.
\newblock In: J.~Jouannaud (ed.) Constraints in Computational Logics, First
  International Conference, CCL'94, Munich, Germant, September 7-9, 1994,
  \emph{Lecture Notes in Computer Science}, vol. 845, pp. 89--104. Springer
  (1994)

\bibitem{RaiserEtAl2009}
Raiser, F., Betz, H., Fr{\"u}hwirth, T.W.: Equivalence of {CHR} states
  revisited.
\newblock In: F.~Raiser, J.~Sneyers (eds.) Proc. 6th International Workshop on
  {C}onstraint {H}andling {R}ules, Report CW 555, pp. 33--48. Katholieke
  Universiteit Leuven, Belgium (2009)

\bibitem{RaiserTacchella2007}
Raiser, F., Tacchella, P.: On confluence of non-terminating {CHR} programs.
\newblock In: K.~Djelloul, G.J. Duck, M.~Sulzmann (eds.) {C}onstraint
  {H}andling {R}ules, 4th Workshop, CHR 2007, pp. 63--76. Porto, Portugal
  (2007)

\bibitem{personalCommTomSchrijversFeb2016}
Personal communication with {T}om {S}chrijvers (February 2016)

\bibitem{SchrijversDemoen2004}
Schrijvers, T., Demoen, B.: {The K.U.Leuven CHR System: Implementation and
  Application}.
\newblock In: T.~Fr\"uhwirth, M.~Meister (eds.) First Workshop on {C}onstraint
  {H}andling {R}ules: Selected Contributions, pp. 1--5. Ulmer
  Informatik-Berichte, Nr. 2004-01 (2004)

\bibitem{DBLP:series/lncs/5388}
Schrijvers, T., Fr{\"u}hwirth, T.W. (eds.): {{C}onstraint {H}andling {R}ules},
  Current Research Topics, \emph{Lecture Notes in Computer Science}, vol. 5388.
\newblock Springer (2008)

\bibitem{DBLP:journals/jacm/Sethi74}
Sethi, R.: Testing for the {Church-Rosser} property.
\newblock J. ACM \textbf{21}(4), 671--679 (1974)

\bibitem{DBLP:journals/tplp/SneyersWSK10}
Sneyers, J., Weert, P.V., Schrijvers, T., Koninck, L.D.: As time goes by:
  {{C}onstraint {H}andling {R}ules}.
\newblock TPLP \textbf{10}(1), 1--47 (2010)

\bibitem{Tarjan:1984:WAS:62.2160}
Tarjan, R.E., van Leeuwen, J.: Worst-case analysis of set union algorithms.
\newblock J. ACM \textbf{31}(2), 245--281 (1984)

\bibitem{Vit67}
Viterbi, A.J.: Error bounds for convolutional codes and an asymptotically
  optimum decoding algorithm.
\newblock IEEE Transactions on Information Theory \textbf{13}, 260--269 (1967)

\end{thebibliography}
